\documentclass{article}

\usepackage{graphicx} 

\usepackage[letterpaper, left=1in, right=1in, top=0.9in, bottom=0.9in]{geometry}
\usepackage[american]{babel}
\usepackage[normalem]{ulem}
\usepackage{amsmath, amssymb, amsthm}
\usepackage{mdframed,lipsum}
\usepackage{thmtools, thm-restate}
\usepackage{tikz}
\usetikzlibrary{calc}

\usepackage[shortlabels]{enumitem}
\usepackage{mdframed}
\usepackage{bbm}
\usepackage{bm}
\usepackage{cancel}
\usepackage{xcolor}
\usepackage{mathtools}
\usepackage{algorithmic}
\usepackage[procnumbered,ruled,vlined,linesnumbered]{algorithm2e}
\usepackage{float}
\usepackage{varwidth}
\usetikzlibrary{decorations.pathreplacing}
\usepackage{subcaption} 
\usepackage{multirow}
\usepackage{natbib} 
\usepackage{enumitem}

\SetKwInput{KwData}{Input}
\SetKwInput{KwResult}{Output}
\SetKwInput{KwGlobalVar}{Global variables}
\SetKw{Break}{break}

\usepackage[bookmarks,colorlinks,breaklinks]{hyperref}
\hypersetup{urlcolor=blue, colorlinks=true, citecolor=green!50!black, linkcolor=blue}
\usepackage[capitalize,nosort,nameinlink]{cleveref}
\usepackage{authblk}
\usepackage{centernot}

\newcommand{\pluseq}{\mathbin{\mbox{+=}}}

\declaretheorem[numberwithin=section,refname={Theorem,Theorems},Refname={Theorem,Theorems}]{theorem}

\declaretheorem[numberlike=theorem]{lemma}

\declaretheorem[numberlike=theorem]{corollary}
\declaretheorem[numberlike=theorem]{definition}
\declaretheorem[numberlike=theorem]{claim}
\declaretheorem[numberlike=theorem,style=remark]{remark}

\declaretheorem[numberlike=theorem, refname={Observation,Observations},Refname={Observation,Observations},name={Observation}]{observation}

\theoremstyle{definition}

\DeclareFontFamily{U}{mathb}{\hyphenchar\font45}
\DeclareFontShape{U}{mathb}{m}{n}{
      <5> <6> <7> <8> <9> <10> gen * mathb
      <10.95> mathb10 <12> <14.4> <17.28> <20.74> <24.88> mathb12
}{}
\DeclareSymbolFont{mathb}{U}{mathb}{m}{n}
\DeclareMathSymbol{\llcurly}{3}{mathb}{"CE}
\DeclareMathSymbol{\ggcurly}{3}{mathb}{"CF}

\definecolor{darkgreen}{RGB}{0,100,0}

\newcommand{\stageD}{main case D}

\newcommand{\stages}{main cases}
\newcommand{\stage}{main case}

\newcommand{\StageA}{Main Case A}

\newcommand{\Stage}{Main Case}

\newcommand{\pzero}{basic}
\newcommand{\pone}{height-one}
\newcommand{\ptwo}{simple}
\newcommand{\pthree}{simple height-one}
\newcommand{\Pthree}{Simple Height-One}

\newcommand{\weakresents}{weak most-resents}
\newcommand{\weakresent}{weak most-resent}
\newcommand{\weakresented}{weak most-resented}

\newcommand{\efxf}{EFX-feasible}

\newcommand{\X}{\ensuremath{\mathbf{X}}}
\newcommand{\Y}{\mathbf{Y}}

\newcommand{\C}{\mathbf{C}}
\newcommand{\xp}{\ensuremath{\mathbf{X'}}}
\newcommand{\yp}{\mathbf{Y'}}

\newcommand{\xz}{\mathbf{X^{\prime \prime}}}

\newcommand{\efxfenvyll}[1]{\llcurly_{#1}}

\newcommand{\doesnotefxenvyll}[1]{\centernot\llcurly_{#1}}

\newcommand{\lowerval}[1]{\prec_{#1}}
\newcommand{\lowereqval}[1]{\preceq_{#1}}
\newcommand{\greaterval}[1]{\succ_{#1}}
\newcommand{\greatereqval}[1]{\succeq_{#1}}

\newcommand{\efx}{\textsf{EFX}}
\newcommand{\ef}{\textsf{EF}}

\newcommand{\dumpp}[2]{$#2 \gets #1$}  
\newcommand{\dumpu}[2]{$#2 \gets^* #1$}

\newcommand{\sumsq}{\ensuremath{\cup}}
\newcommand{\takes}{\ensuremath{=}}
\newcommand{\bigsumsq}{\ensuremath{\bigcup}}

\newcommand{\mms}{\textsf{MMS}}
\newcommand{\pmms}{\textsf{PMMS}}
\newcommand{\eefx}{\textsf{EEFX}}
\newcommand{\mxs}{\textsf{MXS}}
\newcommand{\efl}{\textsf{EFL}}
\newcommand{\efr}{\textsf{EFR}}
\newcommand{\rmms}{\textsf{RMMS}}
\newcommand{\effx}{\textsf{EF2X}}

\newcommand{\efkx}{\textsf{EFkX}}

\newcommand{\zat}{Z_1^{(t)}}
\newcommand{\zbt}{Z_2^{(t)}}
\newcommand{\zatt}{Z_1^{(t+1)}}
\newcommand{\zbtt}{Z_2^{(t+1)}}
\newcommand{\yat}{Y_1^{(t)}}
\newcommand{\ybt}{Y_2^{(t)}}

\newcommand{\ya}{Y_1}
\newcommand{\yb}{Y_2}

\newcommand{\yap}{Y'_1}
\newcommand{\ybp}{Y'_2}

\newcommand{\caseA}{Case 1}
\newcommand{\caseB}{Case 2}

\title{$\efx$ Allocations Exist on Multi-Graphs}

\author[1]{Mahyar Afshinmehr \thanks{mahyar.afshinmehr@cs.ox.ac.uk}}
\author[2]{Arash Ashuri\thanks{aashuri@mpi-inf.mpg.de}}
\author[3]{Pouria Mahmoudkhan \thanks{pouriamksh@gmail.com}}
\author[4]{Kurt Mehlhorn \thanks{mehlhorn@mpi-inf.mpg.de}}
\author[5]{Amir Mohammad Shahrezaei \thanks{s.a.m.shahrezaei@gmail.com}}

\affil[1]{University of Oxford}
\affil[2,4]{Max Planck Institute for Informatics}
\affil[3,5]{Sharif University of Technology}

\date{}

\begin{document}

\maketitle

    We study the fair allocation of indivisible goods among agents, with a focus on limiting envy. 
    A central fairness notion is \emph{envy-freeness up to any good} ($\efx$), 
    which requires that any envy toward another agent vanishes after the removal of any single 
    good from the latter’s bundle. The existence of $\efx$ allocations is considered a major open problem in fair division. So far, it has only been established in limited settings.
        
    \citet{CFKS23} proved the existence of $\efx$ allocations for \emph{graphical valuations}. In this setting, the agents correspond to the nodes of an underlying graph, and the goods correspond to the
    edges, and any good has positive value only for the endpoints of the corresponding edge. 
    Their proof crucially relies on the restriction that the graph is \emph{simple}, meaning that for any pair of agents, there is at most one good that has value to both. For \emph{multi-graph 
    valuations}, where multiple goods may be valued by the same pair of agents, only partial results are known.
    \citet{ARS24} and \citet{KKSS25} obtained $2/3$ and $\sqrt{2}/2$ 
    approximations of $\efx$, respectively; \citet{KSS24} established existence under restricted 
    additive valuations; and \citet{AAMM25} proved existence under the assumption that the shortest 
    cycle containing non-parallel edges has length at least $4$.
    
    In this paper, we resolve this open problem by proving the existence of 
    $\efx$ allocations for multigraph instances under  cancelable valuations, a strict superclass of additive valuation functions.
    Our proof is algorithmic and computes such allocations in polynomial time when the valuation functions are cancelable.
    This work contributes to the small number of $\efx$ existence results that apply to an arbitrary 
    number of agents.

\section{Introduction}
    
    Allocating indivisible resources fairly is a fundamental problem in economics and computer science. Formal models of fair division are used to study settings such as inheritance division, public resource allocation, assignment problems, and the distribution of indivisible goods in markets and online platforms \citep{etkin2007spectrum,moulin2004fair,vossen2002fair,budish2012multi,pratt1990fair}. The study of fair division originates with the work of \citet{S49} and has since developed across economics, mathematics, computer science, and social choice theory. See \citep{AABRLMVW23,brams1996fair,brandt2016handbook,robertson1998cake} for general overviews.
    
    A central fairness requirement is envy-freeness. When goods are indivisible, however, envy-free allocations need not exist. Consider two agents and a single good that both agents like. The good has to be allocated to one of the agents, and the other will be envious. This limitation has led to the introduction of relaxed fairness notions that retain the spirit of envy-freeness while avoiding impossibility.
    
    The first studied relaxation is \emph{envy-freeness up to one good} ($\ef1$). An allocation is $\ef1$ if, for every pair of agents, any envy can be eliminated by removing a single good from the envied agent’s bundle. The notion was formally defined by \citet{B10} and was shown earlier to be achievable by \citet{LMMS04}. The existence of $\ef1$ allocations is now well understood.
    
    More recently, attention has shifted to a stronger requirement, \emph{envy-freeness up to any good} ($\efx$), introduced by \citet{CKMPSW19}. An allocation satisfies $\efx$ if, for every pair of agents, removing any single good from one agent’s bundle eliminates the other agent’s envy. This property strictly strengthens $\ef1$ and has emerged as a central benchmark for fairness in the allocation of indivisible goods.
    
    Recently, \cite{AMMSW26} showed that $\efx$ allocations do not always exist even for $3$ agents with monotone valuation functions. 
    The existence of $\efx$ allocations for more restricted classes of valuation functions, such as additive valuations, remains a major open problem. Nevertheless, significant progress has been made by establishing existence guarantees for several restricted valuation domains. These include instances with identical monotone valuations \citep{PR20}, binary valuations \citep{BSY23}, monotone valuations with at most two distinct types \citep{M24}, and additive valuations with at most three distinct types \citep{PGNV25}. For three-agent instances, $\efx$ allocations are known to exist for additive valuations \citep{CGM24}, and these results have been extended to nice cancelable valuations \citep{BCFF21} and $\mms$-feasible valuations \citep{ACGMM23}.

    \citet{CFKS23} showed that $\efx$ allocations exist when valuations can be represented via a simple graph, where goods correspond to edges in a graph, agents correspond to nodes, and each agent values only the edges incident to her. 
    Following the introduction of graph valuations, a natural question is whether $\efx$ allocations exist under \emph{multi-graph valuations}, where multiple goods may connect the same pair of agents. This question was explicitly raised by \citet{CFKS23}. The multi-graph setting is significantly more expressive and introduces additional challenges compared to the simple graph case.
    
    Several recent works have established existence results for $\efx$ allocations on multigraphs under specific structural conditions. Independently \citet{ADKMR25,SS25,BP24} proved that $\efx$ allocations exist when the underlying multigraph contains no odd cycles. In addition, \citet{SS25} showed that $\efx$ allocations exist whenever the shortest cycle using non-parallel edges has a length of at least six. \citet{BP24} further proved that if the underlying graph is $t$-colorable and its shortest cycle has length at least $2t-1$, then an $\efx$ allocation is guaranteed to exist. Moreover, \citet{ADKMR25} established the existence of $\efx$ allocations on multi-graphs when the underlying graph is a single cycle. More recently, \citet{AAMM25} showed that $\efx$ allocations exist for all multigraphs whose shortest cycle using non-parallel edges has length at least four.
    
    In another line of research, approximation guarantees for $\efx$ have been studied in the multigraph setting. \citet{ARS24} showed that every additive multi-graph instance admits a $\frac{2}{3}$-$\efx$ allocation. This bound was recently improved by \citet{KKSS25}, who proved that every such instance admits a $\frac{1}{\sqrt{2}}$-$\efx$ allocation. \citet{KSS24} showed that exact $\efx$ allocations exist on multi-graphs when the valuation functions are restricted additive.
    
    Additional work has focused on algorithmic and complexity aspects of fair orientations, where each good is allocated to agents that value it. \citet{CFKS23} showed that deciding whether an $\efx$ orientation exists or not on simple graphs is NP-complete. \citet{ZWLL24} studied the mixed manna setting with both goods and chores, proving that deciding the existence of $\efx$ orientations on simple graphs with additive valuations is NP-complete. \citet{ZM25} established a connection between the existence of $\efx$ orientations and the chromatic number of the underlying graph. \citet{BGRS25} refined this picture by showing that bipartiteness is a tight boundary for tractability, identifying hardness even for graphs that are close to bipartite. More recently, \citet{DEGK25} proved that $\ef1$ orientations always exist for monotone valuations and can be computed in pseudo-polynomial time. Finally, \citet{KNPMV25} studied envy-free orientations in graphs and multi-graphs, with a focus on parameterized complexity, complementing the existing literature on $\efx$ orientations.

    \subsection{Our Contributions}
    We study the existence of $\efx$ allocations for multi-graph valuations, a setting that generalizes graphical valuations introduced by \citet{CFKS23}. While $\efx$ allocations are known to exist for simple graphs, extending these guarantees to multigraphs has remained an open problem, and prior progress relied either on approximation guarantees or additional structural assumptions.
    
    Our main result resolves this open question for cancelable valuations, a strict superclass of additive valuation functions. We provide a constructive proof that $\efx$ allocations always exist in multigraph instances when agents have cancelable valuation functions, for an arbitrary number of agents. This establishes the first exact $\efx$ existence result for multigraph valuations without restricting the underlying graph structure. We show that $\efx$ allocations can be computed in polynomial time on all multigraphs under cancelable valuation functions.
    
    Our work complements and goes beyond previous results in the multigraph setting. In contrast to earlier guarantees that required cycle-length conditions \citep{AAMM25}, restricted valuation classes \citep{KSS24}, or only achieved approximate notions of $\efx$ \citep{ARS24,KKSS25}, our result establishes exact $\efx$ existence under a standard and widely studied valuation model.

\subsection{Technical Overview}

We provide a brief overview of the main ideas behind our algorithm for computing $\efx$ allocations on mult-igraphs.
The algorithm proceeds in three phases and operates in the space of partial $\efx$ allocations.

\paragraph{\textbf{Phase One: Constructing Unit Bundles.}}
    
    For every pair of agents $i$ and $j$, let $E_{i,j}$ denote the set of edges between them.
    We compute two different partitions of $E_{i,j}$ into two bundles.
    
    From agent $i$'s perspective, we partition $E_{i,j}$ into two $\efx$-feasible\footnote{Given a partition $\mathcal{Y}=\langle Y_1,\dots,Y_\ell\rangle$ of a set $S$, we say that a bundle $Y_t$ is $\efx$-feasible for agent $i$ if she weakly prefers $Y_t$ to every other bundle in the partition after the removal of any single item from the competing bundle.} bundles, denoted by $(a_{j,i}, b_{j,i})$, where agent $j$ weakly prefers $a_{j,i}$ to $b_{j,i}$.
    Symmetrically, we compute another partition $(a_{i,j}, b_{i,j})$ from agent $j$'s perspective.
    We refer to these four bundles collectively as the \emph{unit bundles} associated with $E_{i,j}$.
    
    These unit bundles remain fixed throughout the algorithm. Moreover, except for one special case handled later, whenever an agent $\ell$ receives an item from $E_{i,j}$, she receives an entire unit bundle. During the execution of the algorithm, we dynamically choose which of the two partitions of $E_{i,j}$ to use and whether to assign all unit bundles in the partition to the endpoints of $E_{i,j}$. If the partition $(a_{j,i},b_{j,i})$ is used, only $a_{j,i}$ can be assigned to $j$ and only $b_{j,i}$ can be assigned to $i$. We also refer to these bundles as their share of the partition. If the partition $(a_{j,i},b_{j,i})$ is used and both endpoints receive their share, $j$ will not envy $i$, $i$ may envy $j$, but does not strongly envy $j$.\footnote{Removing any good from agent $j$ would eliminate the envy towards agent $j$.}
 
\paragraph{\textbf{Phase Two: A Structured Partial Allocation.}}
    
    In the second phase, we construct a partial allocation in which every agent receives exactly one incident unit bundle. The allocation satisfies the following properties:
    \begin{enumerate}
    \item the partial allocation is $\efx$;
    \item every agent weakly prefers her own bundle to every unallocated incident unit bundle; and
    \item the envy graph\footnote{The envy graph associated with a partial allocation has one node per agent, and an edge from agent $i$ to agent $j$ whenever $i$ envies $j$’s bundle relative to her own.} has no directed path of length greater than one. Moreover, every agent is envied by at most one other agent, and, if $i$ envies $j$, $j$ owns $a_{j,i}$ and nothing else. 
    \end{enumerate}
    
    The phase begins by allowing each agent, in some arbitrary order, to select her most preferred unallocated incident unit bundle.
    Whenever the envy graph contains a directed path of length at least two, we carefully choose one of these paths, and perform a rerouting step: every agent on the path except the last receives the bundle she envies, while the last agent receives her most preferred unallocated incident unit bundle. This operation preserves $\efx$ while strictly shortening long envy chains.
    Consequently, at the end of the phase, the envy graph consists of a collection of stars (see \cref{fig 1}) since the length of its longest path is at most one, and since every agent is envied by at most one other agent. The latter holds since we assign exactly one unit bundle to each agent, and hence only the other end of the corresponding edge may envy.

    \begin{center}
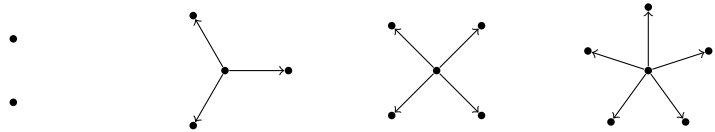
\begin{figure}[H]
    \centering

\begin{tikzpicture}[scale=0.7, every node/.style={circle,fill=black,inner sep=1pt}] 
\node (A) at (0,0) {};
\foreach \i in {0,120,240}{
  \node (A\i) at (\i:1.2) {};
  \draw[->] (A) -- (A\i);
}

\node (B) at (4,0) {};
\foreach \i in {45,135,225,315}{
  \node (B\i) at ($(B)+(\i:1.2)$) {};
  \draw[->] (B) -- (B\i);
}

\node (C) at (8,0) {};
\foreach \i in {90,162,234,306,18}{
  \node (C\i) at ($(C)+(\i:1.2)$) {};
  \draw[->] (C) -- (C\i);
}
\node (D) at (-4,0.6) {};
\node (E) at (-4,-0.6) {};
\end{tikzpicture}
    \caption{This is an example of how the envy graph may look after phase two.}
    \label{fig 1}
\end{figure}
\end{center}

    From now on, our algorithm preserves the three  properties above, but we allow non-envied agents
    to own more than one unit bundle. We still need to allocate the remaining unallocated items. At this point, for any pair $(i,j)$, either no item in $E_{i,j}$ is assigned, or one of the unit bundles in one of the partitions is already assigned, or both unit bundles in one of the partitions are already assigned. In the second case, we must still assign the other bundle in the partition. In the first case, we must still choose the partition, we do so arbitrarily, and then assign its unit bundles.
    
    If, at this point, the envy graph contains no envy edges, then the remaining unit bundles can be allocated easily. For every pair of agents $(i,j)$, we assign one unit bundle in $E_{i,j}$ to $i$ and the other to $j$, while ensuring that each agent retains the unit bundle already assigned to her in Phase Two. In other words, if $E_{i,j}$ is completely unallocated, we choose one of the partitions arbitrarily and assign its bundles as described above. If a bundle in one of the partitions is already assigned to an endpoint, we assign the other bundle to the other endpoint. If both bundles in a partition are already assigned, we do nothing.
    The resulting allocation is $\efx$ since each agent already weakly prefers her allocated bundle to every remaining incident unit bundle (Property (2)).
    
    The more interesting case occurs when some agents are envied. In this situation, every non-envied agent $i$ can still safely receive one of the unallocated unit bundles in $E_{i,j}$ per
    every other agent $j$, while preserving $\efx$. The challenge is therefore to allocate the remaining unit bundles incident to envied agents.

    To address this issue, consider an envied agent $j$. Suppose agent $i$ envies $j$. Then $j$ currently holds the bundle $a_{j,i}$, and for every agent $k\notin \{i,j\}$, at least one unit bundle in $E_{k,j}$ is fully unallocated. Define
    $Z_j = \bigcup_{k\notin \{i,j\}} Z_{k,j}$,
    where $Z_{k,j}$ is the most valuable fully unallocated unit bundle in $E_{k,j}$ from agent $j$'s perspective.

    If agent $j$ weakly prefers $b_{i,j}\cup Z_j$ to her current bundle, then we can allocate
    $b_{i,j}\cup Z_j$ to $j$ and give $j$'s previous bundle to agent $i$. This strictly decreases the number of envied agents. Hence, after exhaustively applying this operation, every remaining envied agent $j$ strictly prefers her current bundle to $b_{i,j}\cup Z_j$.

    This observation suggests a natural completion strategy: for each envied agent $j$, we attempt to allocate the remaining unit bundles incident to $j$ to agents who do not value them. We refer to this operation as \emph{dumping}. Intuitively, dumping allows us to dispose of low-value bundles without creating new envy. Similar ideas were previously proven to be necessary even in simple graph settings \cite{CFKS23}.

    The difficulty is that careless dumping may create new envy relations. To avoid this issue, we identify non-envied agents that can safely absorb dumped bundles while remaining non-envied. Certain structural configurations of the envy graph guarantee that such agents exist.

    For example, suppose there exist two non-envied agents $s$ and $t$ such that the bundle of agent $s$ consists solely of a unit bundle in $E_{s,t}$. We call this pair a support pair.
    Then every envied agent $j\notin\{s,t\}$ can safely dump the bundles in $Z_j$ (agent $j$'s share) to agent $s$ without causing $j$ to envy $s$. Recall that $j$ weakly prefers her current bundle to $b_{i,j} \cup Z_j$. 
    For a non-envied agent $p$ we guarantee that she possesses a unit bundle in $E_{j,p}$ that she weakly prefers over the unit bundle from $E_{j,p}$ that is potentially dumped on $s$. Hence, $p$ does not envy $s$. 
    Such configurations enable us to complete the allocation while preserving $\efx$.

    One issue with this is that if $j_1$ and $j_2$ are both envied agents, they cannot both dump their share in $E_{j_1,j_2}$ to $s$, since then agent $s$ gets more than $Z_{j_1}$ w.r.t.~agent $j_1$.
    The existence of two disjoint support pairs allows us to circumvent this issue and to allocate all remaining unit bundles. However, if there exists only one support pair, we still can use it to allocate some of the remaining unit bundles, and find a new way to allocate the remaining ones based on the
    other properties of the allocation.

    \paragraph{\textbf{Phase Three: Completing the Allocation.}}
    
    The final phase systematically allocates the remaining unit bundles incident to envied agents.
    
    Suppose agent $i$ envies agent $j$.
    Recall that this implies that $j$ owns exactly $a_{j,i}$.
    For every agent $k\notin \{i,j\}$, we proceed as follows:
    \begin{enumerate}
    \item If agent $k$ is not envied by $i$, then we dump agent $j$'s share in $E_{j,k}$ to agent $i$;
    \item Otherwise, we allocate one of the remaining unit bundles in $E_{i,j}$ to agent $i$ and allocate the other bundle to another non-envied agent.
    \end{enumerate}

    We have depicted the proposed allocation in \cref{fig:phasethree}.
    In some cases, this direct approach fails. However, the reason for failure always reveals additional structure in the envy graph, which can then be exploited to construct a complete $\efx$ allocation using variants of the arguments developed in Phase Two.

    \begin{figure}
        \centering

\begin{center}
\begin{tikzpicture}[
    scale=1.3,
    every node/.style={
        circle,
        draw,
        inner sep=2pt,
        minimum size=10mm
    },
    main/.style={minimum size=11mm},
    envy/.style={->, thick},
    dashededge/.style={dashed, thick},
    pointydashed/.style={dashed, ->, thick}
]

\node[main] (i0) at (0,0) {$i_0$};
\node[main] (i1) at (3,2) {$i_1$};
\node[main] (i2) at (3,-2) {$i_2$};

\node[main] (j0) at (8,2) {$j_0$};
\node[main] (j1) at (8,-2) {$j_1$};

\node (a12) at (3,0.7) {$a_{i_1,i_2}$};
\node (b12) at (3,-0.7) {$b_{i_1,i_2}$};

\node (a1j0) at (4.7,2) {$b_{j_0,i_1}$};
\node (b1j0) at (6.3,2) {$a_{j_0,i_1}$};

\node (a2j1) at (4.7,-2) {$a_{i_2,j_1}$};
\node (b2j1) at (6.3,-2) {$b_{i_2,j_1}$};

\draw[envy] (i0) -- (i1);
\draw[envy] (i0) -- (i2);
\draw[envy] (j0) -- (j1);


\draw[dashededge] (i1) -- (a12);
\draw[dashededge] (a12) -- (b12);
\draw[dashededge] (b12) -- (i2);

\draw[dashededge] (i1) -- (a1j0);
\draw[dashededge] (a1j0) -- (b1j0);
\draw[dashededge] (b1j0) -- (j0);

\draw[dashededge] (i2) -- (a2j1);
\draw[dashededge] (a2j1) -- (b2j1);
\draw[dashededge] (b2j1) -- (j1);


\draw[pointydashed] (a12) -- (i0);
\draw[pointydashed] (b12) -- (j0);

\draw[pointydashed] (a1j0) -- (i0);
\draw[pointydashed] (b1j0) -- (j0);

\draw[pointydashed] (a2j1) -- (i0);
\draw[pointydashed] (b2j1) -- (j0);

\end{tikzpicture}
\end{center}

            \caption{This figure illustrates an example of the allocation proposed in Phase Three. In this example, there are five agents, namely $i_0,i_1,i_2,j_0,j_1$. Agent $i_0$ envies both $i_1$ and $i_2$, while agent $j_0$ envies $j_1$. Envy relations are represented by solid directed edges.
            Dashed directed edges indicate the recipient of each unit bundle in the proposed allocation. Dashed non-directed lines indicate the agents to whom the corresponding unit bundles belong.\protect\\
            The pair $(i_2,j_1)$ illustrates case one twice: $i_0$ envies $i_2$, but does not envy $j_1$ and $j_0$ envies $j_1$, but does not envy $i_2$. We select one of the partitions of $E_{i_2,j_1}$ arbitrarily, say $(a_{i_2,j_1}, b_{i_2,j_1})$, and assign $i_2$'s share to $i_0$ and $j_1$'s share to $j_0$.\protect\\
            The pair $(i_1,i_2)$ illustrates case two: $i_0$ envies $i_1$ and $i_2$. We select one of the partitions of $E_{i_1,i_2}$ arbitrarily, say $(a_{i_1,i_2}, b_{i_1,i_2})$, and assign $i_1$'s share to $i_0$ and $i_2$'s share to another non-envied agent, here $j_0$.
            }
        \label{fig:phasethree}
    \end{figure}
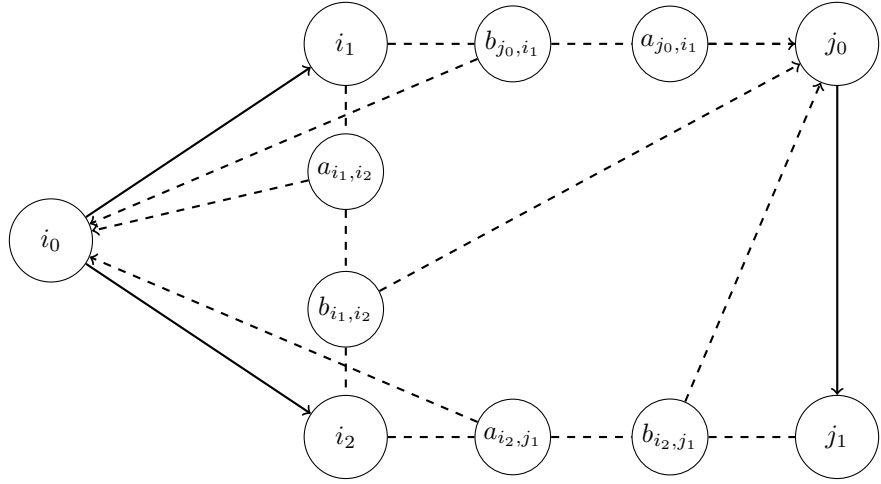

    \subsection{Further Related Work}
    Beyond exact $\efx$, a substantial body of work studies relaxations that are achievable in more general valuation settings. One prominent direction considers multiplicative approximations of $\efx$. For subadditive valuations, \citet{PR20} proved the existence of $1/2$-$\efx$ allocations, while for additive valuations, \citet{AMN20} established the existence of $1/\phi$-$\efx$ allocations, where $\phi \approx 1.618$. Stronger guarantees are known under additional restrictions. In particular, $2/3$-$\efx$ allocations were shown to exist when agents’ valuations fall into at most four types \citep{HMN25}, when the number of agents is at most seven, \citep{ARS24}, eight \citep{FKW26} and when agents agree on their top $n$ items \citep{MS23}. Further improvements for restricted settings were obtained by \citet{BKP24}.
    
    Another relaxation is $\efkx$, which requires envy to disappear after the removal of any $k$ goods from an envied bundle. Existence of $\effx$ allocations has been established for four agents with cancelable valuations \citep{AGS25}, for any number of agents with restricted additive valuations \citep{ARS22}, and for $(\infty,1)$-bounded valuations \citep{KSS24}. 
    Moreover, \citet{AGS25} showed that $\effx$ allocations can be computed in polynomial time for three agents with cancelable valuation function. In contrast, the known algorithms for computing $\efx$ allocations generally run in pseudo-polynomial time.
    Recently, \citet{FKW26} showed that for every $k \geq 2$, a $\frac{k+1}{k+2}$-$\efkx$ allocation always exists. They further proved that $\efkx$ graph orientations do not always exist, and that deciding whether such orientations exist is NP-complete.

    Several alternative relaxations of $\efx$ have also been proposed, including $\efl$ \citep{BBMN18}, $\efr$ \citep{FHLSY21}, and the notions $\eefx$ and $\mxs$ \citep{CGRSV22}, which were later extended to monotone valuations by \cite{HANR25}. Related work has also examined the compatibility of different fairness guarantees, such as the coexistence of $2/3$-$\mms$ with $\ef1$ \citep{AR25}, the combination of $\mxs$ and $\efl$ \citep{AG25}, and allocations satisfying $\rmms$ and $\efl$ \citep{F25}, and both $\eefx$ and $\efl$ \citep{AMMR26}.
    
    Another approach allows discarding a limited number of goods to charity and seeks $\efx$ allocations over the remaining items. This model was initiated by \citet{CGH19}, who showed that an $\efx$ partial allocation always exists and achieves at least half of the optimal Nash social welfare. \citet{CKMS20} proved that for $n$ agents with monotone valuations, $\efx$ can be achieved by donating up to $n-1$ goods while ensuring that no agent envies the donated bundle. Subsequent work reduced the number of donated goods under additional assumptions: \citet{BCFF21} showed that $n-2$ goods suffice for nice cancelable valuations, and \citet{M24} extended this result to monotone valuations. Also, \citet{BCFF21} proved that for four agents, an $\efx$ allocation exists with at most one donated good. Further improvements were obtained by allowing approximate guarantees, reducing the number of donated goods while achieving $(1-\varepsilon)$-$\efx$ \citep{CGMMM21,ACGMM23,BBK22,CSJS23}.
    
    Related works have also explored fairness notions that are strictly stronger than $\efx$. One such notion is $\pmms$, introduced by \citet{CKMPSW19}. While $\pmms$ implies $\efx$, \citet{BMP25} showed that $\pmms$ allocations do not always exist, even for three agents, a setting in which $\efx$ allocations are known to exist.
    
    A parallel literature studies fair allocation of indivisible chores, where agents are endowed with cost functions. For additive cost functions, \citet{ZW24} obtained an $O(n^2)$-approximation for $\efx$, which was later improved to a $4$-$\efx$ guarantee for any number of agents by \citet{GMQ25}. Independently, \citet{CS24} and \citet{afshinmehr2024approximateefxexacttefx} showed the existence of $2$-$\efx$ allocations for three agents, and \citet{garg2025existence2efxallocationschores} extended this guarantee to any number of agents. Exact $\efx$ allocations for chores are known only in highly restricted settings; in fact, \citet{CS24} showed that with monotone cost functions, instances without $\efx$ allocations exist. Recently, \citet{mahara2025existencefairefficientallocation} showed the existence of an allocation that is both $\ef1$ and Pareto optimal for additive valuation functions.
    More recently, \citet{HT26} showed that there exists four additive cost functions for which
    no $\efx$ allocation exists.
    
    Another line of research considers fairness over time rather than in a single allocation. \citet{CL26} studied temporal fair division of indivisible goods across multiple rounds, introducing temporal analogues of $\ef1$, $\efx$, and maximin share, and establishing sharp possibility and impossibility results both with and without scheduling.
    
    Fairness can also be defined in a share-based manner, where each agent is guaranteed a minimum value that depends only on her own valuation and the set of goods. The most prominent notion in this class is the \emph{maximin share} ($\mms$). An allocation is said to be an $\alpha$-approximate MMS allocation if every agent receives a bundle worth at least an $\alpha$ fraction of her maximin share. The first non-trivial approximation guarantee of $2/3$ was established by \citet{KPW18}, and was later simplified and shown to be computable in polynomial time by a sequence of works including \cite{AMNS17,GHS18,BK20,GMT19}. Subsequent results pushed the approximation frontier to $3/4 + O(1/n)$ \citep{GHS18,GT20,AGST23}. More recent breakthroughs established approximation guarantees strictly exceeding $3/4$, including $3/4 + 3/3836$ by \citet{AG24} and $10/13$ by \citet{HKSS25}.

\section{Preliminaries}\label{sec:prelims}

    An instance of \emph{discrete fair division} is given by a tuple
    $\langle N, M, \{v_i\}_{i \in N} \rangle$, where
    $N = [n] = \{1,2,\dots,n\}$ denotes the set of agents,
    $M = [m]$ is a finite set of $m$ indivisible goods, and
    $\{v_i\}_{i \in N}$ is a profile of valuation functions.
    Each agent $i \in N$ has an associated valuation function
    $v_i : 2^M \rightarrow \mathbb{R}_{\ge 0}$. 
    

    Agent $i$ is said to strictly prefer a bundle $T$ to a bundle $S$
    if $v_i(T) > v_i(S)$, denoted by $S \prec_i T$.
    Similarly, $S \preceq_i T$ denotes weak preference, meaning that
    $v_i(S) \le v_i(T)$.
    For notational convenience, given a set $S \subseteq M$ and an item
    $g \in S$, we write $S \setminus g$ to denote $S \setminus \{g\}$, $S \cup g$ to denote $S \cup \{g\}$,
    and use $g$ as shorthand for the singleton set $\{g\}$.

   \paragraph{Types of valuation functions.} 
    We consider valuation functions that are monotone, i.e., for any $S\subseteq T\subseteq M$, $v(S) \leq v(T)$. 
    One of the most well-studied classes of valuation functions is \emph{additive}: a valuation function $v$ is additive if the value for any bundle $S\subseteq M$ is equal to the sum of the values of its goods, i.e., $v(S)=\sum_{g\in S} v(g)$.
    In this paper, we focus on the more general class of \emph{cancelable} valuations: a valuation function $v$ is \emph{cancelable}~\citep{BCFF21} 
    if for any two bundles $S, T \subset M$ 
    and any good \ $g \in M \setminus (S \cup T)$,
    if $v (S \cup g) > v (T \cup g)$, then $v (S) > v (T)$, i.e., removing the same good from two different bundles would not change the relative preference between the two. It is easy to verify
    that for any cancelable valuation $v$ and bundles $S, T, R$ such that $R \subseteq M \setminus (S \cup T)$ we have the following two properties, which we heavily use throughout the paper: 
    \begin{align*}
        v(S \cup R) > v(T\cup R)  &\Rightarrow v(S) > v(T) \,;\\
        v(S\cup R) \leq v(T\cup R) &\Leftarrow  v(S) \leq v(T)\,.
    \end{align*}

\begin{observation}
\label{cancp}
    If agent $i$ has a cancelable valuation function and for some bundles $S,T,Q,R$ 
    we have $S\lowereqval{i} T$, $Q\lowereqval{i} R$, and 
    $(S\cup T)\cap(Q\cup R)=\emptyset$,
    then $S\cup Q \lowereqval{i}  T \cup R$.
\end{observation}

\begin{proof}
    Since $S\lowereqval{i} T$, and $Q$ is disjoint from $S$ and $T$, 
    by cancelability it holds that $S \cup Q \lowereqval{i} T \cup Q$. 
    Similarly, since $Q \lowereqval{i} R$, and $T$ is disjoint from $Q$ and $R$, 
    by cancelability it holds that $ T \cup Q \lowereqval{i} T \cup R$. By combining these two, we have that
    $S \cup Q \lowereqval{i} T \cup Q \lowereqval{i} T \cup R.$
\end{proof}

\paragraph{Multigraph Instances.}
A fair division instance $\mathcal{I} = \langle [n], [m], \{v_i\}_{i \in [n]} \rangle$
defined on a \emph{multigraph}\footnote{A multigraph may contain multiple edges between the same pair of vertices.}
is represented by a multigraph $G = (V, E)$.
In this representation, the $n$ agents correspond to the vertices in $V$,
while the $m$ indivisible goods correspond to the edges in $E$.

The valuations satisfy the following condition: for every agent
$i \in [n]$ and every subset of goods $S \subseteq [m]$, the value of $S$
to agent $i$ depends only on the goods incident to $i$, namely, $v_i(S) = v_i(S \cap E_i)$, where $E_i \subseteq E$ denotes the set of edges incident to agent $i$.
For any pair of agents $i$ and $j$, let $E_{i,j}$ denote the set of edges
connecting $i$ and $j$; note that $E_{i,j} = E_{j,i}$. Moreover, for every set of agents $A \subseteq N$, we define $E_A = \bigcup_{i \in A} E_i$ to be the set of goods incident to an agent in the set $A$. Similarly, $E_{A, B} = \bigcup_{i \in A, j \in B} E_{i, j}$ denotes the set of goods between agents in sets $A, B$. For notational simplicity, given a subset of agents $A \subseteq N$ and a single agent $i \notin A$,  and we write $E_{i, A}$ to denote $E_{\{i\}, A}$.

 \paragraph{Allocations and Orientations.}
A \emph{partial allocation,} or simply \emph{allocation,} $\X = \langle X_1, X_2, \ldots, X_n \rangle$
is an ordered collection of pairwise disjoint subsets of $M$.
That is, for any two distinct agents $i,j \in N$, we have
$X_i, X_j \subseteq M$ and $X_i \cap X_j = \emptyset$.
The set $X_i$ represents the bundle of goods allocated to agent $i$
under allocation $\X$. An allocation $\X$ is \emph{complete} if every good
is allocated, i.e.,$\bigcup_{i \in [n]} X_i = [m]
$. In this work, we always mention completeness explicitly. 

A \emph{partial orientation} is a partial allocation satisfying the
additional constraint that each agent is assigned only goods incident
to them; formally, for all $i \in [n]$, we have $X_i \subseteq E_i$.

\paragraph{Envy, $\efx$, and Related Notions.}
Let $i \in N$ be an agent and let $S,T \subseteq M$ be two bundles of goods.
Agent $i$ \emph{prefers} $T$ over $S$ (or $i$ \emph{envies} $T$ with respect to $S$ ) if $v_i(T) > v_i(S)$. 
We say that $i$ \emph{strongly envies} bundle $T$ with respect to bundle $S$ if there exists a good
$g \in T$ such that $v_i(T \setminus g) > v_i(S)$. Consider an allocation $\X = \langle X_1, X_2, \ldots, X_n \rangle$.
Agent $i$ is said to envy (respectively, strongly envy) agent $j \in N$
if agent $i$ envies (respectively, strongly envies) $X_j$ with respect to $X_i$.
An allocation $\X$ is called \emph{$\efx$} if no agent strongly envies any
other agent. Let $\Y =  (Y_1, \ldots, Y_k)$ be a partition of a subset of $M$. A bundle $Y_\ell$ is \emph{$\efx$-feasible} for an agent $i$ in $\Y$ if
agent $i$ with $Y_\ell$ does not strongly envy any bundle in $\Y$. Finally, for fair division instances defined on a multigraph,
an allocation $\X$ is an \emph{$\efx$ orientation} if it is an $\efx$ allocation
and satisfies $X_i \subseteq E_i$ for all $i \in [n]$.

    Next, we now dive into specific definitions, notations, and some routines that we use during our proof. We use the concept of \emph{unit bundles} introduced by \citet{AAMM25}, however, defined somewhat differently. Our modified definition is more convenient.

\begin{definition}[Unit Bundle]\label{def:ub}\label{obs:b-notexfenvy-a}\label{obs: ub-ineq}
	For any two agents $i$ and $j$, let $a_{i,j} \sqcup b_{i,j} = E_{i,j}$ be a partition of $E_{i,j}$ such that both $a_{i, j}$ and $b_{i, j}$ are $\efx$-feasible for agent $j$ in the partition $(a_{i, j}, b_{i, j})$, and
	$a_{i,j}  \greatereqval{i}  b_{i,j}$. Note that the second index indicates the agent for which the bundles are EFX, and the first index determines which bundle is called the $a$-bundle. 
	Define  $a_{j,i}$ and $b_{j,i}$ analogously, i.e., $a_{ji}$ and $b_{ji}$ are EFX-feasible for agent $i$ and $a_{ji} \greatereqval{j} b_{ji}$. We call these bundles 
    \emph{unit bundles} between $i,j$ if they additionally satisfy the following inequalities:
    \[  v_i(a_{ij}) \ge \max(v_i(a_{ji}),v_i(b_{ji})) \ge \min(v_i(a_{ji}),v_i(b_{ji})) \ge v_i(b_{ij}) \]
    and 
    \[  v_j(a_{ji}) \ge \max(v_j(a_{ij}),v_j(b_{ij})) \ge \min(v_j(a_{ij}),v_j(b_{ij})) \ge v_j(b_{ji}),\]
    i.e., the partition with respect to $i$ ($j$) is at least as balanced for $i$ ($j$) as the partition with respect to $j$ ($i$).
\end{definition}
    In \cref{sec:computing-unit-bundles}, we provide an algorithm, and in \cref{cor:unitbundleoutput}, we show that such unit bundles are computable for 
    cancelable valuation functions with this algorithm.
    Moreover, in \cref{lem:polyunitbundle}, we show that our algorithm terminates in polynomial time when the valuation functions are cancelable.
    We compute the two decompositions into unit bundles for each pair of agents in a preprocessing step and use them throughout the algorithm. 

\begin{remark} 
    Note that for additive valuations, there is a simple way to get unit bundles. This is to let $(a_{i,j},b_{i,j})$ be a partition of $E_{i,j}$ such that its minimum-valued bundle with respect to agent $j$ has 
    the highest possible value for agent $j$, i,e., for every other partition $(S,T)$ of $E_{i,j}$, we have $$\min (v_{j}(a_{i,j}), v_j(b_{i,j})) \geq \min (v_{j}(S), v_j(T)),$$
    and since the valuation function is additive, it gives us:
    $$\max (v_{j}(a_{i,j}), v_j(b_{i,j})) \leq \max (v_{j}(S), v_j(T)).$$
    Then, after finding this partition, denote the bundles in the partition such that $a_{i,j} \greatereqval{i} b_{i,j}$.\\
    Similarly, we construct two other unit bundles  $(a_{j,i},b_{j,i})$ using the same approach, but by inverse roles for agent $i$ and $j$. The above construction is simple, but requires exponential time. 

    If the valuations are cancelable, then among all the partitions that maximize the minimum-valued bundle, we should choose a partition that minimizes the maximum-valued bundle.
\end{remark}

Unit bundles are called as such because they will not be broken up throughout the algorithm except in one case in \cref{sec:d}. In fact, we keep the following property:

\begin{definition}[Unitary Allocation]\label{def:unitary}
    We call an allocation \emph{unitary} for agent $i$ if:
    \begin{itemize}
        \item For every agent $j\ne i$, agent $i$ may hold either $a_{i,j}$ or $b_{j,i}$ or nothing from $E_{i,j}$, i.e., $X_i\cap E_{i,j}  \in \{a_{i,j},b_{j,i},\emptyset\}$. 
        \item $v_i(X_i) \ge v_i(B)$ for every unallocated unit bundle $B$\footnote{A unit bundle is unallocated, if all goods inside it are unallocated. If a unit bundle is not unallocated, either it is fully allocated, or it is partially allocated because some other unit bundle is allocated.}. 
    \end{itemize}
    We call an allocation \emph{unitary}, if it is unitary for every agent $i$.
\end{definition}


    Note that in a unitary orientation, if $i$ and $j$ both hold parts of $E_{ij}$, then they either hold $(a_{ij}, b_{ij})$ or $(b_{ji}, a_{ji})$. In the former case, $i$ will not envy $j$ and $j$ will not strongly envy $i$. In the latter case, $j$ will not envy $i$ and $i$ will not strongly envy $j$.

\begin{lemma}\label{lem:minimum value}
    If $\X$ is a unitary orientation, then for any agents $i$ and $j$, 
    we have $X_j \greatereqval{j} b_{i,j}$ and $X_j \greatereqval{j} b_{j,i}$. 
\end{lemma}
\begin{proof}
    Because $\X$ is unitary, we have $X_j \cap E_{i,j} \in \{\emptyset, a_{j,i}, b_{i,j}\}$.
    If $X_j \cap E_{i,j} \in \{a_{j,i}, b_{i,j}\}$, using the fact that $a_{j,i} \greatereqval{j} b_{i,j}$, we obtain $X_j \greatereqval{j} b_{i,j}$.
    
    Next, suppose $X_j \cap E_{i,j} =\emptyset$. Because $\X$ is unitary, we have $X_i \cap E_{i,j} \in \{\emptyset, a_{i,j}, b_{j,i}\}$. Since $\X$ is an orientation, at least one of the unit bundles $b_{i,j}$ and $a_{j,i}$ is unallocated.
    Since $\X$ is unitary, we have $X_j \greatereqval{j} B$ for every unallocated unit bundle $B$. Hence, 
    $X_j \greatereqval{j} \arg \min_{S \in \{ b_{i,j},a_{j,i}\}} v_j(S) = b_{i,j}$.
    Moreover, since $b_{i,j} \greatereqval{j} b_{j,i}$, it follows that 
    $X_j \greatereqval{j} b_{j,i}$.
\end{proof}

    Next, we define ``resent", a weaker form of envy that is used in all of our algorithms.

\begin{definition}[Resent]
    In a allocation, we say an agent $u$ \emph{resents} an agent $v$ if $a_{u,v} \greaterval{u} X_u$.
    In this case, we also say that $u$ resents the bundle $a_{u,v}$.
    Moreover, we say that an agent $u$ \emph{most-resents} an agent $v$ if
    $u$ resents $v$, and $a_{u,v} \greatereqval{u} a_{u,j}$ for any other agent $j$. 
   \end{definition}

\begin{remark}
    In a unitary orientation, if $u$ resents $v$, then $a_{u,v} \greaterval{u} X_u$.
    As a result, since valuation functions are monotone, $a_{u,v}$ is not allocated to $u$. Moreover, by the second condition of unitary allocation, 
    $a_{u,v}$ cannot be unallocated.
    Since the allocation is an orientation and $a_{u,v}  \subseteq E_{u,v}$, we get that some of $a_{u,v}$
    is allocated to $v$.
    Since the allocation is unitary, by the first condition, we get $X_v\cap E_{u,v}\in \{\emptyset, a_{v,u},b_{u,v}\}$, and since $a_{u,v} \cap b_{u,v} = \emptyset$, 
    we get $X_v \cap E_{u,v} = a_{v,u}$.  
    Hence, if $u$ resents $v$, then $u$ prefers $a_{u,v}$ over her current bundle, but $v$ owns $a_{v,u}$ and hence $a_{u,v}$ is currently not available for $u$. 

    Note that if $j$ is not resented by $u$ then $X_u \greatereqval{u} a_{uj}$. Therefore, writing $a_{uv} \greatereqval{u} a_{uj}$ for any other $j$ or for any other $j$ resented by $u$ are equivalent.
\end{remark}

\begin{lemma}\label{lemma:resent-envy}
    In a unitary orientation \X, if an agent $j$ envies an agent $i$, then $j$ also resents $i$.
\end{lemma}
\begin{proof}
    Since the allocation is an orientation, $X_i \subseteq E_i$, and therefore, $X_i \cap E_j = X_i \cap E_{i,j}$.
    If agent $j$ envies agent $i$, then by the definition of multigraph instances, we have $X_j \lowerval{j} X_i \cap E_j = X_i \cap E_{i,j}$.
    Combining this with \cref{lem:minimum value}, we obtain $b_{i,j} \lowerval{j} X_i \cap E_{i,j}$.

    Moreover, since the allocation is unitary, we have $X_i \cap E_{i,j} \in \{\emptyset, a_{i,j}, b_{j,i}\}$.
    Using the fact that $b_{j,i} \lowereqval{j} b_{i,j}$ and $b_{i,j} \lowerval{j} X_i \cap E_{i,j}$, 
    we get that $X_i \cap E_j = X_i \cap E_{i,j} = a_{i,j}$.
    Moreover, we have $a_{j,i} \greatereqval{j} a_{i,j}$, so we get $X_j \lowerval{j} X_i \cap E_{i,j} = a_{i,j} \lowereqval{j} a_{j,i}$.
    Hence, the proof is complete.
\end{proof}


\begin{definition}[Resent Graph]
    Given an allocation $\X$, the \emph{resent graph} of $\X$, denoted by $G_r(\X)$, is a directed graph defined as follows. The vertex set of $G_r(\X)$ is $N$, with one vertex corresponding to each agent. For any pair of agents $i,j \in N$, there is a directed edge from $i$ to $j$ if and only if agent $i$ resents agent $j$.
    We denote such a directed edge by $i \to j$.
\end{definition}
    Throughout the paper, the term “path” refers to a directed path.

\begin{definition}\label{def:p}
A partial allocation $\X$ is called \pzero\ if it has the following properties:
\begin{itemize}
    \item $\X$ is an orientation;
    \item $\X$ is unitary;
    \item The resent graph $G_r(\X)$ is a forest.
    \item For any pair of agents $i, j \in N$, if $i \to j$, then $X_j = a_{j, i}$. 
\end{itemize}
    The allocation $\X$ is called \pone\ if it is \pzero\ and every tree in the resent-graph has height at most one. The allocation $\X$ is called \ptwo\ if it is \pzero\ and each agent receives at most one unit bundle. The allocation $\X$ is called \pthree\ if it is both \pone\ and \ptwo.

\begin{observation}\label{obs:basicEFX}
    If an allocation $\X$ is \pzero, then it is \efx.
\end{observation}
\begin{proof}
    Since $\X$ is \pzero, it is a unitary orientation, so if $i$ envies $j$, then $i$ resents $j$, too.
    Hence, $X_j= a_{j,i}$, and by \cref{lem:minimum value}, $X_i \greatereqval{i} b_{j,i}$, so agent $i$ does not strongly envy agent $j$.
\end{proof}


\end{definition}


\begin{observation}\label{obs:resented-bundle}
    If allocation $\X$ is \pzero, every agent is resented by at most one agent.
\end{observation}
\begin{proof}
    Given two agents $i$ and $j$ such that $i$ resents $j$, by the fourth property of \pzero\ allocation, we get $X_j = a_{j,i}$.
    Since $X_j =  a_{j,i}$, we get that no agent other than $i$ resents $j$.
\end{proof}

We want to show that a \ptwo\ allocation always exists, and, in fact, we use such an allocation as a starting point of our algorithm. To show that such an allocation always exists, we fix an order on our agents and sequentially ask each agent to pick an unallocated unit bundle of maximum value. Specifically, whenever an agent $i$ comes to pick her favorite unit bundle, for any other agent $j$ preceding $i$ in our order, $j$ might only hold $a_{j, i}$ from the set $E_{i,j}$, as she had the chance to pick her bundle greedily. To formalize the procedure, we define the \textsc{Choose($i$)} function for any agent $i \in N$ as follows:


\begin{definition}[Choose Function]
Given an allocation $\X$ and an agent $i$, such that $X_i = \emptyset$, we define \textsc{Choose($i$)} to return the most valuable unallocated unit bundle for agent $i$. 
\end{definition}

	

    Next, we provide and algorithm, and we show that it computes a \ptwo\ allocation in the following lemma.

\begin{algorithm}
\caption{\sc{Greedy Orientation}}\label{alg:p2-maker}
\KwIn{A multigraph instance $\langle N, M, \{v_i\}_{i \in N} \rangle$}
\KwOut{A \ptwo\ allocation $\X$}

\For{$(i \gets 1$;  $i \le n$ ;  $i  \pluseq  1)$}{

    $X_{i} \gets \textsc{Choose($i$)}$ 
}

\Return $\X$

\end{algorithm}

\begin{lemma}\label{lem:p2-maker}
    \Cref{alg:p2-maker} constructs a  \ptwo\ allocation. 
\end{lemma}
\begin{proof}
To show that the allocation is \ptwo, we prove its properties one by one. Consider an arbitrary pair of agents $i$ and $j$. We will show that agent $i$ does not strongly envy agent $j$. If $i$ comes before $j$ in our ordering ($i < j$), then by our greedy construction, $i$ will not envy $j$. If $j$ precedes $i$ in the ordering, then $i$ may only envy agent $j$ if $X_j \cap E_{i, j} \neq \emptyset$. In fact, by the definition of the \textsc{Choose} function, since $j$ picks her bundle greedily, we must have $X_j \cap E_{i, j} = a_{j, i}$. Then, when $i$ comes to pick her favorite bundle, the unit bundle $b_{j, i}$ is unallocated, and by the greedy nature of construction, agent $i$ receives a unit bundle with a value of at least $v_i(b_{j,i})$, that is, $v_i(X_i) \ge v_i(b_{j,i})$. Therefore, by \Cref{obs:b-notexfenvy-a}, agent $i$ does not strongly envy agent $j$. Therefore, the allocation is an $\efx$ orientation. 

Next, we show that the allocation is unitary. Note that by our greedy construction, for every unallocated unit bundle $B$ and every agent $i$ we have that $v_i(X_i) \ge v_i(B)$. Moreover, for every two agents $i$ and $j$, if $X_i \cap E_{i, j} \neq \emptyset$, then by construction, only one of the following cases can happen:

\begin{itemize}
    \item \textbf{Case 1: $i < j$:} Then, $X_i = a_{i, j}$.
    \item \textbf{Case 2: $i > j$ and $X_j \cap E_{i, j} = \emptyset$:} Then, $X_i = a_{i, j}$.
    \item \textbf{Case 3: $i > j$ and $X_j \cap E_{i, j} \neq \emptyset$:} Then, $X_i = b_{j, i}$.
\end{itemize}
Therefore, the allocation is unitary. 

Next, we show that the resent graph does not contain a cycle. Consider an arbitrary subset $A$ of agents, and let $k$ be the agent in this subset who comes first in our order. By our greedy construction, agent $k$ does not resent any other agent in this subset. Therefore, $G_r(\X)$ is a forest. 

Finally, each agent receives at most one unit bundle, which  completes our proof.
\end{proof}

\section{Existence of $\efx$ Allocations on Multigraph Instances}

    We provide an algorithmic proof for the existence of $\efx$ allocations on multigraphs under cancelable valuations. We start with an initial \ptwo\ allocation, which exists by \Cref{lem:p2-maker}. Next, we modify this allocation using \Cref{alg:treebreaker}, discussed in \Cref{sec:p3-maker}, and obtain a \pthree\ allocation. Finally, for the \pthree\ allocation, we distinguish eight cases. For each one, we show how to obtain a complete $\efx$ allocation. 

    For a \pzero\ allocation $\X$, we define six sets.

\begin{definition}
	For an \pzero\ allocation $\X$, we define the following:
    \begin{itemize}
        \item $R(\X)$ is the set of all resented agents in $\X$.
        \item $A_i(\X)$ is the union of the best unallocated unit bundles from the $E_{i,j}$ with $X_i\cap E_{i,j}=\emptyset$, i.e.,
				$$A_i(\X)= \bigsumsq_{j: \ X_i\cap E_{i,j}= X_j\cap E_{i,j} =\emptyset} a_{i,j}
						\  \sumsq \ \bigsumsq_{j: \ X_i\cap E_{i,j}=\emptyset \text{ and } X_j\cap E_{i,j}\ne \emptyset} E_{i, j} \setminus X_j.$$
        \item $B_i$ is the union of all $b_{j,i}$, i.e., $B_i(\X)=\bigsumsq_{j: \ j\ne i} b_{j,i}$.
        \item $C_i(\X)$ is the union of all $a_{i,j}$, where $j$ is a resented vertex, i.e., $C_i(\X) = \bigsumsq_{j \in R(\X)} a_{i,j}$.
        \item $D_i(\X)$ is the union of all $b_{j,i}$, where $j$ is a resented vertex, i.e., $D_i(\X) = \bigsumsq_{j \in R(\X)} b_{j,i}$.
        \item $R_i(\X)$ is the set of agents resented by $i$. We drop the allocation and write $R_i$ when the allocation is clear from the context. 
    \end{itemize}
\end{definition}

   \begin{observation}
       For every \pzero\ allocation, and every agent $i$, $A_i(\X)$ is unallocated in $\X$.
       Moreover, for every resented agent $i$, $C_i(\X),$ and $D_i(\X)$ are unallocated in $\X$.
   \end{observation}

    Next, we define our \stages\ and state our main theorem with its proof.

\begin{definition}
    For a \pthree\ allocation $\X$, we distinguish eight \stages. A case is only applied, if none of the preceding \stages\ applies.

    \begin{itemize}
        \item \textbf{\StageA:} There exist four agents $i, j, k, \ell$ such that $k\to i$, $\ell \to j$, and
        \begin{align*}
            D_j(\X) &\sumsq a_{j,k} \greaterval{j} X_j,\\
            D_i(\X) &\sumsq a_{i,\ell} \hspace{0.3mm} \greaterval{i} X_i.
        \end{align*}

        \item \textbf{\Stage\ B:} There exist two agents $i, j$ such that $j\to i$ and
        $$A_j(\X) \sumsq X_j \greatereqval{j} a_{j,i} .$$

        \item \textbf{\Stage\ C:} There exist four agents  $i, j, k, \ell$ such that $k\to i$, $\ell \to j$, and
        $$D_j(\X) \sumsq a_{j,k} \greaterval{j} X_j.$$

        \item \textbf{\Stage\ D:} There exists at most one resent tree with non-zero height in allocation $\X$.

        \item \textbf{\Stage\ E:} There exist two agents $i, j$ such that $j\to i$ and
        $$A_i(\X) \sumsq b_{j,i} \greaterval{i} X_i .$$

        \item \textbf{\Stage\ F:} There exist three agents $i, j, k$ such that $j \to i$, $k$ is non-resented, $k$ does not have $a_{k,j}$, and
        $$a_{k,j} \sumsq D_k(\X) \greaterval{k} X_k.$$
        
        \item \textbf{\Stage\ G:} There exist four agents $i, j, k, \ell$ such that $k\to i$, $\ell \to j$, and
        \begin{align*}
            [C_j(\X)\cap E_{j,R_k}] \sumsq [D_j(\X) \setminus E_{j,R_k}] \sumsq a_{j,k} &\greaterval{j} X_j,\\
            [C_i(\X)\cap E_{i,R_\ell}] \hspace{0.6mm} \sumsq [D_i(\X)\setminus E_{i,R_\ell}] \hspace{0.6mm} \sumsq a_{i,\ell} \hspace{0.6mm} 
			&\greaterval{i} X_i.
        \end{align*}
        \item \textbf{\Stage\ H:} None of the \stages\ A to G applies.
    \end{itemize}
\end{definition}

Next, we provide our main result.

\begin{theorem}
	For any multigraph instance, there exists a complete $\efx$\ allocation under cancelable  valuations. Moreover, such allocations can be computed in polynomial time. 
\end{theorem}
\begin{proof}
    We first run \Cref{alg:p2-maker} and obtain a \ptwo\ allocation by \Cref{lem:p2-maker}. Next, we run \Cref{alg:treebreaker} and obtain a \pthree\ allocation by \Cref{lem:p3-maker}. Finally, the \pthree\ allocation falls into one of the \stages\ A, B, C, D, E, F, G, or H. In either case, we obtain a complete $\efx$\ allocation by \cref{lem: stageA}, \cref{lem: stageB}, \cref{lem: stageC}, \cref{lem: stageD}, \cref{lem: stageE}, \cref{lem: stageF}, \cref{lem: stageG}, and \cref{lem: stageH}, respectively. It is clear that all our algorithms terminate in polynomial time. 
    Moreover, we can compute the unit bundles in polynomial time by 
    \cref{lem:polyunitbundle}. Hence, the entire construction runs in polynomial time.
 \end{proof}


 \paragraph{Structure of the Paper:} The sections of the main paper are numbered, and the sections of the appendix are labelled A, B, \ldots.  In \cref{sec:p3-maker}, we introduce some algorithms used throughout the paper and use one of them for computing a \pthree\ allocation. The missing proofs of this section are postponed to \cref{sec:tree_appendix}. For the \pthree\ allocation, we distinguish the eight cases defined above. All cases have the same structure. We first perform a couple of local update rules (none in case H) and then allocate the remaining unallocated goods through a phase called ``dumping." In \cref{sec:dump} we provide some general dumping rule structures that most of the time occur in our dumping phase, and prove some general properties that hold after our dumping phase. More precisely, these properties are essential in proving the correctness of the dumping phase designed for each \stage. The missing proofs of this section are deferred to \cref{sec:dump_appendix}. In \cref{sec:a} and \cref{sec:b}, we deal with cases A and B, respectively. The discussion of the other cases are deferred to the corresponding Sections of the appendix.

\section{ Creating a \Pthree\ Allocation}\label{sec:p3-maker}
    In this section, we introduce algorithms for transforming an allocation with certain properties into a \pone\ allocation, and for transforming a \ptwo\ allocation into a \pthree\ allocation.
    The algorithms are based on the concept of a \emph{critical path}. For an allocation $\X$, let $T(\X)$ be the set of trees in $G_r(\X)$. Moreover, let $T_0(\X)$ be the trees in $T(\X)$ with height zero, 
	which we call trivial trees.

\begin{definition}[Critical Path]
	For a tree $T \in T(\X)$ with root $r$ and height at least two, 
	the \emph{critical path} of $T$ is a directed path that starts at $r$ and ends at a leaf, constructed with the following properties:
\begin{enumerate}
    \item The second vertex on the path is the (unique) non-leaf child $i$ of $r$ that maximizes $v_r(a_{r,i})$.
    \item Recursively, for each vertex $i$ on the path, except for $r$, the next vertex is the child $j$ of $i$ that maximizes $v_i(a_{i,j})$, i.e., the child that $i$ most-resents. 
\end{enumerate}
	Note that the second property does not have the non-leaf condition mentioned in the first property. The length of such a critical path is at least two by the first property.
\end{definition}

    For a tree $T$, let $r(T)$ be the root of $T$, 
    and let $H(T)$ be the length of a longest path in $T$.  Next, we introduce the notion of \emph{breaking} a resent tree with height at least two via 
	\cref{alg:break}.

\begin{algorithm}
\caption{\textsc{BreakTree}$(\X, T)$}\label{alg:break}
\KwIn{A \pzero\ allocation $\X$ and tree $T$ with $H(T)\geq 2$}
\KwOut{A \pzero\ allocation $\X$ and end vertex $i_k$}

	Let $i_1\to\dots\to i_k$ be the critical path of $T$

\For{$(\ell \gets 1$;  $\ell < k$ ;  $\ell  \pluseq  1)$}{
    $X_{i_\ell} \takes a_{i_\ell,i_{\ell+1}}$ 
}

	$X_{i_k} \takes \textsc{Choose($i_k$)}$\\
	\Return $\X,i_k$
\end{algorithm}

    Algorithm \textsc{BreakTree} proceeds as follows. Let $T$ be a resent tree, and let 
	$(i_1 = r(T), i_2, \ldots, i_k)$ be its critical path. We unallocate all bundles currently allocated to agents on this path, allocate to each agent $i_\ell$, for $\ell \in [1, \ldots, k-1]$, $a_{i_\ell, i_{\ell+1}}$, and finally allocate to $i_k$ an unallocated unit bundle of highest value to her.

\begin{lemma}\label{lem:breaktree}
	Suppose $\X$ is a \pzero\ allocation, $T$ is a resent tree in $G_r(\X)$ with $H(T)\geq 2$, and $i_1\to\dots\to i_k$ is the critical path of $T$.
    Then, after the execution of \textsc{BreakTree$(\X, T)$}:
	\begin{enumerate}
    \item Allocation remains an unitary orientation.
	\item Agents $i_2,\ldots, i_{k-1}$ do not resent any agent. 
	\item Agents $i_1,i_2,\ldots,i_{k-2},i_k$ (every agent on the critical path except $i_{k-1}$) are non-resented.
    \item For every $t\ne i_k$, if $t$ did not resent an agent $u$, then $t$ does not resent $u$ in the resulting allocation.
	\item The resent graph remains a forest.
	\end{enumerate}
\end{lemma}
\begin{proof}
    \begin{enumerate}
    \item For every agent $i$ not in the critical path, her bundle does not change,
    so for every other agent $j$, since allocation was unitary, we still have 
    $X_i \cap E_{i,j}= \{a_{i,j}, b_{j,i},\emptyset\}$.
    Moreover, each agent $i_\ell$, for $\ell \in [1, \ldots, k-1]$, gets her most preferred unit bundle $a_{i_\ell, i_{\ell+1}}$ by definition of critical path and \cref{lem:minimum value}.
    
    For every agent $i\ne i_k$, the value of her bundle does not decrease, and for every agent $i\ne i_1$,
    the only unit bundle incident to agent \(i\) that may have become unallocated is \(a_{i,i_1}\). However, since \(i_1\) was non-resented,
    we had $X_i \succeq_i a_{i,i_1}$.
    In addition, the only unit bundle incident to agent $i_1$ is the one she released, but she does not resents that because she is getting a more valuable bundle.
    Therefore, for every agent $i\ne i_k$, the first two condition of unitary allocation holds.

    Additionally, $i_k$ gets the unallocated unit bundle of highest value to her.
    If it is in $E_{i,j}$, then if $a_{i_k,j}$ was unallocated, she gets $a_{i_k,j}$. Otherwise, 
    we should have that agent $j$ possessed $a_{j,i_k}$ since the unitary condition holds for every other agent. Thus, agent $i_k$ gets $b_{j,i_k}$. Hence, the unitary condition holds for agent $i_k$ as well, so allocation remains unitary and orientation.

    \item Since allocation remains unitary orientation, we can talk about resent. 	
    Agents $i_2,\ldots,i_{k-1}$ get their most preferred unit bundle by definition of critical path and \cref{lem:minimum value}. Hence, 
    they do not resent anyone anymore.

    \item For every $\ell \in \{1,\ldots, k-2\}$, agent $i_\ell$ gets $a_{i_\ell,i_{\ell+1}}$, so
    the only agent who may resent $i_\ell$ is agent $i_{\ell+1}$. However, since agent $i_{\ell+1}$
    was previously resented by $i_\ell$, we had 
    $X_{i_{\ell+1}}= a_{i_{\ell+1},i_\ell} \greatereqval{i_{\ell+1}} a_{i_\ell,i_{\ell+1}}$,
    so since agent $i_{\ell+1}$ is getting a more valuable bundle from before, $i_{\ell+1}$ does
    not resent agent $i_\ell$. Hence, $i_\ell$ is non-resented.
    Moreover, since agent $i_k$ gets an unallocated unit bundle, and since allocation was unitary,
    no one envied that unallocated unit bundle, so no one resents agent $i_k$.

    \item 
    Agents $i_2,\ldots, i_{k-1}$ do not resent any agent, and agent $i_1$ may only resent the ones she resented before. 
    Moreover, $i_1,\ldots, i_{k-2}, i_k$ are non-resented, and $i_{k-1}$ is only resented by $i_k$,
    so except from agent $i_k$, no new resent is created.
    
    \item Since there was no cycles in the resent graph before, and new resent is possible only from agent $i_k$, if there is a new cycle, 
    it should contain agent $i_k$, but since $i_k$ is non-resented, this may not occur. \qedhere
    \end{enumerate}
\end{proof}

    
	Note that the root $i_1$ may still resent other agents, 
	since she receives only her most preferred bundle among those belonging to non-leaf agents,
    but no path starting from agent $i_1$ will have a length greater than one anymore.
    By construction, $i_k$ resents $i_{k-1}$ after the procedure, and may also resent additional agents.


    Next, we present the \cref{alg:treereduce} that we use both in the algorithm of constructing a \pthree\ allocation 
    and when we want to turn a given allocation with special properties to a \pone\ allocation.

\begin{algorithm}
\caption{\textsc{Reduce Trees}}\label{alg:treereduce}
\KwIn{A \pzero\ allocation $\X$, and an agent $i$.}
\KwOut{A \pzero\ allocation $\X$}
		Let $T$ be the resent tree rooted at $i$\\
    \While{$H(T)>1$}{
		$\X,j\gets  \textsc{BreakTree}(T)$  \\
		Let $T$ be the resent tree rooted at $j$
    }
\Return $\X$
\end{algorithm}

	Next, we present the  \cref{alg:treebreaker} which turns a \ptwo\ allocation into a \pthree\ allocation.

\begin{algorithm}
\caption{\textsc{Remove Trees}}\label{alg:treebreaker}
\KwIn{A \ptwo\ allocation $\X$}
\KwOut{A \pthree\ allocation $\X$}

\While{there exists a $T\in T(\X)$ with $H(T)>1$}{    
	\textsc{Reduce Trees}($\X$,$r(T)$)
}
\Return $\X$
\end{algorithm}

\begin{restatable}{lemma}{lemtreebreakersupport}\label{lem:treebreaker support}
    Given a \pzero\ partial allocation $\X$, and a partition of agents into 3 parts $H,G,\{h_0\}$, and an agent $p\in H$ such that:
    \begin{enumerate}
        \item $\forall r \in N\setminus H$, for $i= \arg\max_{\ell \in H \setminus p} v_r(a_{r,\ell})$,
        unit bundle $a_{i,r}$ is not allocated to agent $i$.
		\item Agent $p$'s bundle is a unit bundle in $E_{p,h_0}$, and agent $p$ does not resent anyone.
        \item $h_0$ is non-resented.
		\item Every agent in $N\setminus H$ only holds one unit bundle.
    \end{enumerate}
	The execution of \textsc{Reduce Trees}$(\X,h_0)$ terminates after at most $n$ iterations of the while loop, 
	and at the end, the agents in $H$ still hold their initial bundle, 
    every agent in $N\setminus H$ possesses exactly one unit bundle,
    and $h_0$ remains non-resented.
\end{restatable}

\begin{restatable}{lemma}{localtreebreakerends}\label{local tree breaker ends}
    	Given a \ptwo\ allocation $\X$ and $T\in T(\X)$, the execution of \textsc{Reduce Trees}$(\X,r(T))$ terminates after at most $n$ iterations, and allocation remains \ptwo. 
\end{restatable}

	We define a potential function that is decreased by each iteration of the  while loop in \Cref{alg:treebreaker}.
    For an agent $i$ and allocation $\X$, let $D_i(\X)$ be the depth of $i$ in the resent tree to which she belongs, i.e., the length of the longest path towards agent $i$ in $G_r(\X)$.
    Moreover, let $\phi_i(\X)= \max((D_i(\X)-1),0)$, and let
	$$\phi(\X)= \sum_{i\in [n]} \phi_i(\X).$$


\begin{restatable}{lemma}{lemtreei}\label{lem:p3-maker}
        Executing \cref{alg:treebreaker} on a \ptwo\ allocation $\X$ produces a \pthree\ allocation in polynomial time, assuming that the unit bundles are given as input.
\end{restatable}

		Next, we present a technical lemma on our tree-breaking procedure, which we use throughout the paper.


\begin{restatable}{lemma}{lemtreeiii}\label{lem:treebreaker2}
    Given a \pzero\ partial allocation $\X$, and a partition of agents into 3 parts $H,G,\{h_0\}$, and an agent $p\in H$ such that:
    \begin{enumerate}
		\item $\forall j \in N\setminus H$, for $i= \arg\max_{\ell \in H \setminus p} v_j(a_{j,\ell})$,
        unit bundle $a_{i,j}$ is not allocated to agent $i$.
		\item $X_p=a_{p,h_0}$, and agent $p$ does not resent anyone.
		\item $h_0$ is not resented, and every path in the resent graph that does not include $h_0$ has length 1.
		\item Every agent in $N\setminus H$ only holds one unit bundle.
    \end{enumerate}
    Then, after execution of \textsc{Reduce Trees}$(\X,h_0)$, we get an allocation such that:
    \begin{itemize}
        \item The allocation is \pone.
        \item The agents in $H$ still hold their initial bundle, and if $h \in H\setminus p$ was non-resented, she remains non-resented.
		\item For every $h\in H$, and for every $t\notin \{h,h_0\}$, no allocated unit bundle in $E_{h,t}$ gets unallocated.
        \item $h_0$ remains non-resented.
    \end{itemize}
\end{restatable}

\section{Dumping Lemmas}\label{sec:dump}

    In each \stage, our algorithm first executes some update rules and then allocates the remaining (unallocated) unit bundles in a ``dumping phase". 
    
\begin{definition}[Global Dumping Properties]\label{pbd}
    The following properties hold before the dumping (except in \cref{lem: stageD>4} and in \cref{lem: stageD=4}): 
\begin{enumerate} 	
	\item \label{pbd1} The allocation is \pone.
	\item \label{pbd2} For any $p$ and $q$: If $p$ resents $q$, $X_q = a_{q,p} \greatereqval{q} b_{q,p} \sumsq A_q(\X)$.
\end{enumerate}
\end{definition}

Next, we introduce the concept of \emph{support pairs}.

\begin{definition}(Support Pair)
	Given a \pone\ allocation $\X$, a pair of agents $(s,t)$ is called a support pair if
    \begin{itemize}
        \item both agents $s$ and $t$ are non-resented.
        \item $X_s \subseteq E_{s, t}$, i.e., $X_s \in \{a_{st}, b_{ts}\}$.
    \end{itemize}
\end{definition}

	With a slight abuse of notation, during our algorithm, whenever a unit bundle $U$ is to be allocated to agent $s$, we denote such an action by \dumpp{U}{s} which simply means we add $U$ to $X_s$, i.e. $X_s \gets X_s \cup U$. Moreover, for any pair of agents $i, j$, and an agent $s$, we use \dumpu{a_{i, j}}{s} to denote the following procedure: we take any unit bundle from $E_{i,j}$ that agent $s$ may already possess from her, 
	and then allocate $a_{i,j}$ to $s$, i.e., $X_s = (X_s \setminus E_{i, j}) \cup a_{i, j}$. For any pair of agents $i, j$, a unit bundle $U$ from $E_{i, j}$, and an agent $s$, we use \dumpp{U}{s} only when $s$ does not possess any unit bundle other than $U$ from $E_{i,j}$. 
	So, \dumpu{a_{i,j}}{s} and \dumpp{a_{i,j}}{s} emphasize that $s$ may or may not already possess a unit bundle in $E_{i,j}$. 

    The dumping phases follow some general rules for allocating (dumping) the remaining unit bundles. At the beginning of a dumping phase, there may be several support pairs. For some of them, the dumping phase will exploit their support pair property; the agents in the other pairs are treated just as the other non-resented agents. We say that the former pairs are used by the dumping phase.
    Let $\X$ and $\xp$ denote the allocation before and after the dumping phase. 
    Next, we introduce some general structures that our dumping phase follows (except in \cref{lem: stageD>4} and in \cref{lem: stageD=4}).

\begin{definition}\label{general dumping rules}
\begin{enumerate}
	\item \label{pd7} For every two non-resented agents $p$ and $r$, we select one of the partitions of $E_{p,r}$ and assign one of the  unit bundles in it to $p$ and the other to $r$. 
    
    \item \label{pd8} For any support pair $(s,t)$ used by the dumping phase, $X_s$ already contains a bundle from $E_{s,t}$; we give the other bundle in the partition to $t$ (if it does not already have it). For a pair $(s,p)$ with $p \not= t$, we perform \dumpu{a_{p,s}}{p}. If $p$ holds $b_{p,s}$ in $\X$, it is replaced. We also give $b_{p,s}$ to $s$. 
    Moreover, at the end of the dumping phase, $X_s \cap E_t= X_s' \cap E_t$, i.e., agent $s$ does not get any additional unit bundle that is incident to $t$ and her allocated unit bundle from $E_{s, t}$ does not change.

    	\item Let $p$ be a non-resented agent.
    \begin{enumerate}
        \item \label{pd3} If $a_{p,q}\subseteq X_p$ for some $q$, then $a_{p,q}\subseteq X'_p$.
	\item \label{pd4} If $b_{q,p}\subseteq X_p$ for some $q$, then either
	 		$a_{p,q}\subseteq X'_p$ or $b_{q,p}\subseteq X'_p$. 
	\item Let $q$ be any resented agent. 
      \begin{enumerate} 
      \item \label{pd5} 
    If $p\to q$, \dumpp{b_{q,p}}{p}.  Note that in this situation, $q$ already owns $a_{p,q}$. 
      \item \label{pd6} If $r\to q$ and $r \not= p$, \dumpu{a_{p,q}}{p}.
      \end{enumerate}
    \end{enumerate}

    \item \label{pd2} The above does not specify the allocation of all unit bundles. For the remaining unit bundles, it is guaranteed that for any pair $(p,q)$ of agents, one of the partitions is used, and its bundles are assigned to distinct non-resented agents, i.e.,
    for every two distinct agents $i$ and $j$, any agent $\ell$ possesses at most one unit bundle from $E_{i,j}$ in $\xp$.

    	\item \label{pd1} If an agent is resented in $\X$, her bundle does not change. 
      \end{enumerate}
\end{definition}

\begin{restatable}{lemma}{lemdumprule}\label{lem:dumprule}
    If $\X$ is an allocation satisfying the dumping properties (\cref{pbd}),
    and the dumping phase follows the general structure above (Definition \ref{general dumping rules}), then after the dumping:
\begin{enumerate}
	\item \label{efx1} For every agent $p$, we have $X'_p \greatereqval{p} X_p$, where $\xp$ is the allocation after the dumping.
	\item \label{efx2} If agent $q$ is resented in $\X$, then $q$ is not strongly envied by any agent in $\xp$.
    \item \label{efx3} For every resent $p\to q$ in $\X$, agent $q$ does not envy agent $p$ in $\xp$. 
	\item \label{efx4} For any support pair $(s,t)$ used by the dumping phase, at the end, 
						  no one envies agent $s$, and agent $s$ does not strongly envy anyone. 
\end{enumerate}
\end{restatable}

Next, we introduce our first update rule that we use several times.

\begin{definition}[Update Rule U1\label{up rule:u1}] 
	Given a \pone\ allocation and $i \to j$ such that 
	$b_{j,i} \sumsq A_j(\X) \greaterval{j} X_j = a_{j,i}$, let ($\xp$ denotes the allocation after the update)
	\begin{align*}
		X'_j&\takes b_{i,j} \sumsq A_j(\X) \\
		X'_i&\takes a_{i,j}
	\end{align*}
\end{definition}

\begin{restatable}{lemma}{lemcyclesupport}\label{lem:cycle_support}
Let $\X$ be a height-one allocation, and let $i$ and $j$ be agents such that $i \to j$ and $b_{j,i} \sumsq A_j(\X) \greaterval{j} X_j = a_{j,i}$. Then after executing update rule (U1), 
	$(i,j)$ is a support pair,  $j$ does not resent anyone, no new resent relation is created, and the allocation remains height-one. Moreover, the process of executing (U1) as long as possible terminates.
\end{restatable}


The following lemma accounts for one of the possible structures that allows us to construct a complete $\efx$ allocation via a specific dumping phase that follows the general structure above.

\begin{restatable}{lemma}{lemtwosupport}\label{lem:two_support}
    Let $\X$ be a height-one allocation. If there exists two disjoint support pairs, then we can construct a complete \efx\ allocation in polynomial time. 
\end{restatable}

\section{\Stage\ A}\label{sec:a}
	In this section, we show that given an allocation in \stage\ A, we can construct a complete \efx\ allocation.	
	First, we provide a support lemma that we use to prove the existence of an \efx\ allocation in \stage\ A.

\begin{lemma}\label{lem:R}
	If there is a \pone\ allocation $\X$,
	and for some distinct agents $k,i,\ell,$ and $j$ we have $k\to i$, $\ell \to j$, and:
        \begin{align*}
            D_j(\X) &\sumsq a_{j,k} \greaterval{j} X_j,\\
            D_i(\X) &\sumsq a_{i,\ell} \hspace{0.3mm} \greaterval{i} X_i.
        \end{align*}
 	Then, we can construct a complete \efx\ allocation in polynomial time.
\end{lemma}
\begin{proof}
    We update allocation $\X$ as follows ($\xp$ denotes the allocation after the update):
    \begin{align*}
        X'_i &\takes  a_{i,\ell} \sumsq D_i(\X) \\ 
        X'_k &\takes a_{k,i}\\ 
        X'_j &\takes a_{j,k} \sumsq (D_j(\X) \setminus E_{i, j}) \cup a_{j, i}\\ 
        X'_\ell &\takes a_{\ell,j}.
    \end{align*}
    
	First, we show that this allocation is possible. Note that agents $i$ and $j$ drop what they have in $E_{i,k}$ and $E_{j,\ell}$, so
	allocating $a_{k,i}$ and $a_{\ell, j}$ to agents $k$ and $\ell$ would be possible. 
	Next, note that agents $k$ and $\ell$ drop anything they may have in $E_{j,k}$ and $E_{i,\ell}$, so
	allocating $a_{j,k}$ and $a_{i,\ell}$ to agents $j$ and $i$ would be possible. 
    Moreover, since $i$ and $j$ were resented in $\X$, we get $D_i(\X)\cap E_{i,j} = b_{j,i}$.
    Thus, using the fact $D_i(\X)\cap D_j(\X) \subseteq E_{i,j}$, we get that
    $((D_j(\X) \setminus E_{i, j}) \cup a_{j, i} ) \cap D_i(\X) = a_{j,i} \cap b_{j,i} = \emptyset$.
    Thus, allocating $D_i(\X)$ and $(D_j(\X) \setminus E_{i, j}) \cup a_{j, i}$ to agents $i$ and $j$ is possible.

	Note that since allocation was \pone, and since $k\to i$, we get $X_i =a_{i,k}$.
	As a result, we get:
	$$X'_k = a_{k,i} \lowereqval{i} a_{i,k} = X_i \lowerval{i}  a_{i,\ell} \sumsq D_i(\X) = X'_i,$$
	so agent $i$ would not resent agent $k$. Moreover, for every agent $r \notin \{i,k\}$, we have $X'_k \cap E_{r}= \emptyset$,
	so no one resents agent $k$. Similarly, for every agent $r \notin \{\ell, j\}$, $X_\ell'\cap E_r = \emptyset$, and hence $r$ does not resent $\ell$. Moreover, agent $j$ does not resent $\ell$, as we have 
    $$X'_\ell = a_{\ell,j} \lowereqval{j} a_{j,\ell} = X_j \lowerval{j} 
    a_{j,k} \sumsq D_j(\X) =a_{j,k} \sumsq (D_j(\X) \setminus E_{i, j}) \cup b_{i,j} \lowereqval{j} a_{j,k} \sumsq (D_j(\X) \setminus E_{i, j}) \cup a_{j, i} = X'_j.$$
    Therefore, agents $k$ and $\ell$ are non-resented in $\xp$. 

    Next, we show that agents $i$ and $j$ become non-resented in $\xp$.
    Note that by \cref{obs:resented-bundle}, agent $\ell$ did not resent $i$ in $\X$, so $a_{\ell,i}  \lowereqval{\ell} X_\ell$. Hence,
    since the agent $\ell$ was non-resented in $\X$, we get $X'_i \cap E_\ell = a_{i,\ell}  \lowereqval{\ell} a_{\ell,i}  \lowereqval{\ell}  X_\ell \lowereqval{\ell}  X'_\ell$.
	Thus, agent $\ell$ does not resent agent $i$. 
    Any other agent $q$ does not resent $i$ since $X'_i \cap E_q \in \{\emptyset, b_{q,i}\}$
    and since by \cref{lem:minimum value}, $X_q\greatereqval{q} b_{i,q} \greatereqval{q} b_{q,i}$. Therefore, agent $i$ becomes non-resented. The similar argument works for $j$ except the case where we analyze agent $i$. Note that we have 
    $$X_j' \cap E_i = a_{j, i} \lowereqval{i} a_{i, j} \lowereqval{i} X_i\lowereqval{i} X_i',$$ 
    where the first inequality comes from \Cref{obs: ub-ineq}, and the second inequality comes from the fact $i$ did not resent $j$ in $\X$ by \cref{obs:resented-bundle}. Therefore, agent $i$ does not resent agent $j$, and agent $j$ is now non-resented.

    Since all four agents $i, j, k, \ell$ are non-resented, by our construction on their modified bundles, it is clear that the pairs $(k,i)$ and $(j,\ell)$ are two disjoint support pairs. As a result, By \cref{lem:two_support}, we can construct a complete \efx\ allocation.
\end{proof}

\begin{corollary}\label{lem: stageA}
	Given an allocation in \stage\ A, we can construct a complete \efx\ allocation.
\end{corollary}

	\paragraph{Not Case A:} From now on, we assume that we are not in Case A, i.e., whenever we have a \pone\ allocation, then for every four distinct agents $k,i,\ell,$ and $j$ 
	such that $k\to i$, $\ell \to j$, at least one of the following holds:
        \begin{enumerate}
            \item $D_j(\X) \sumsq a_{j,k} \lowereqval{j} X_j,$ 
            \item $D_i(\X) \sumsq a_{i,\ell} \hspace{0.3mm} \lowereqval{i} X_i$.
        \end{enumerate}
Throughout our proof, without loss of generality, we assume the above condition for every four agents $k\to i$, $\ell \to j$. For any four such agents, we define:

\begin{definition}\label{lem:strong_A}
	Given a \pone\ allocation, for every four distinct agents $k,i,\ell,$ and $j$ such that $k\to i$, and $\ell \to j$, we define $R(i,j)\in \{k,\ell\}$ 
        \begin{enumerate}
            \item $R(i,j)=k$ if $D_j(\X) \sumsq a_{j,k} \lowereqval{j} X_j,$ 
            \item $R(i,j)=\ell$ otherwise, i.e., $D_j(\X) \sumsq a_{j,k} \greaterval{j} X_j,$ 
					and $D_i(\X) \sumsq a_{i,\ell} \hspace{0.3mm} \lowereqval{i} X_i$.
        \end{enumerate}
\end{definition}

\section{\Stage\ B}\label{sec:b}
	In this section, we show that given an allocation in \stage\ B, we can construct a complete \efx\ allocation. First, we introduce our second update rule, which we use in multiple \stages.


\begin{definition}[Weak Resent]
    In a \pone\ allocation, we say agent $p$ \weakresents\ agent $q$, if 
    $p$ is non-resented, $q$ does not resent anyone,
    $X_q =a_{q,p}$, and for any agent $j$, $a_{p,q}\greatereqval{p} a_{p,j}$,
\end{definition}        

    Note that in a \pone\ allocation, if $p$ resents $q$, then $p$ \weakresents\ some agent $q'$, which
    is the agent $p$ most-resents.

\begin{definition}[Update Rule U2\label{up rule:u2}] 
    Given a \pone\ allocation \X, a support pair $(s,t)$, and two distinct agents 
    $p,q \not\in \{s,t\}$ such that:
    \begin{enumerate}
        \item Every agent except $t$ has exactly one unit bundle,
        \item $p$ \weakresents\ $q$.
        \item $a_{t,p} \sumsq (B_t\setminus E_{t,p}) \greaterval{t} X_t$.  
    \end{enumerate}
    Update as follows: (Use $\xp$ to denote the allocation right before executing \textsc{Reduce Trees}.)
    \begin{align*}
		&\text{Agents } t,p, \text{and } q \text{ drop their previous bundles}\\
		&X'_t \takes  a_{t,p} \sumsq (B_t\setminus E_{t,p})\\
		&X'_p \takes a_{p,q}\\
        &\forall j\notin\{p,t\}  \text{ with } X_j \subseteq E_{j,t}: X'_j = a_{j,t}\\
		&X'_q \takes \textsc{Choose}(q)\\
		&\text{Run }\textsc{Reduce Trees}(\xp, q).
    \end{align*}
\end{definition}

\begin{lemma}\label{lem:t update}
	Update rule (U2) is well-defined, and after executing (U2):
\begin{enumerate}
	\item The allocation is \pone.
	\item $(s,t)$ remains a support pair.
	\item Every agent except $t$ only has exactly one unit bundle.
	\item No new resent is created from agents $t$ and $s$ to other agents.
    \item Agent $q$ \weakresents\ agent $p$.
\end{enumerate}
\end{lemma}
\begin{proof}
    Note that $a_{t,p}$ and $a_{p,q}$ are unallocated after that agents $p$ and $q$ drop their previous bundles because $\X$ is an orientation, so they can be allocated to agents $t$ and $p$, respectively. Moreover, $(B_t\setminus E_{t,p})$ can be allocated to agent $t$, because 
    if some agent $j\notin \{t,p\}$ possesses a unit bundle in $E_{t,j}$, then in the new allocation,
    $j$ gets $a_{j,t}$, and we have $(B_t\setminus E_{t,p})\cap a_{j,t}= b_{j,t} \cap a_{j,t}=\emptyset$. Hence, allocation $\xp$ is well-defined.

    Let $H=\{s,t,p\}$ and $h_0=q$. 
	We first show that conditions of \cref{lem:treebreaker2} hold for $\xp, H$ and  $h_0=q$.

    Since $(s,t)$ was a support pair in $\X$, $X_s \subseteq E_{s, t}$, 
    $X'_t  = a_{t,p} \sumsq (B_t\setminus E_{t,p})$, 
	and $B_t= \bigsumsq_{j\ne t} b_{j,t}$, any agent 
	$h \in H\setminus p =\{s,t\}$ does not hold any $a_{h,q}$ for $q\notin H$. Therefore, the first condition is satisfied. Moreover, we have $X'_p=a_{p,q}$ and by the definition of the update rule (U2), for any agent $j$, $X'_p = a_{p,q}\greatereqval{p} a_{p,j}$.
    Therefore, we get that agent $p$ does not resent anyone, and the second condition is satisfied.

    Note that $h_0=q$ is non-resented in $\xp$, since she is getting an unallocated unit bundle that no agent resented, as $\X$ was \pone. Next, we show that agents $s$ and $t$ are non-resented. Since $X'_s \subseteq E_{s,t}$, $t$ did not resent $s$ before, and $t$ now has a more valuable bundle compared to her previous one, agent $s$ is non-resented.
	Moreover, agent $p$ does not resent $t$ since $X'_t\cap E_p = a_{t,p} \lowereqval{p} a_{p,t} \lowerval{p} a_{p,q}$. Also, any agent $r\notin \{p,t\}$ does not resent $t$ since 
    $X'_t \cap E_r = (a_{t,p} \sumsq (B_t\setminus E_{t,p})) \cap E_r = b_{r,t} \lowereqval{r} X'_r$.
    Therefore, if a new resent relation $i\to j$ has been created right before the execution of \textsc{Reduce Trees}$(q)$, we have $i=q=h_0$ since the other agents are getting a more valuable bundle, meaning that the third condition is satisfied. Moreover, since $\X$ was a \pone\ allocation in which only $t$ possessed more than one unit bundle, all agents in $N \setminus H$ still have a single unit bundle in $\xp$. Thus, the fourth condition is also satisfied. 
    
    Next, we argue that $\xp$ is \pzero. Note that the allocation is clearly, by construction, still an orientation, unitary, and the resent graph is a forest. Moreover, if $i \to j$ in $\xp$, then, $X_j' =a_{j,i}$.  This completes our argument, i.e., $\xp$ is \pzero.

    Now, by applying \cref{lem:treebreaker2}, we show that all the properties in this lemma are satisfied.

    [1] By \cref{lem:treebreaker2}, the allocation will be \pone.

    [2] By \cref{lem:treebreaker2}, the bundles of agents in $H$ do not change, and every non-resented agent in $H$ remains non-resented.
	Therefore, $(s,t)$ is still a support pair.

    [3, 4] Since if an agent's bundle gets changed during \textsc{Reduce Trees}$(q)$, she gets a unit bundle, 
	we get that every agent except $t$ holds exactly one unit bundle, and no new resent relation is created from agents $s$ and $t$ toward other agents.

    [5] Since $p\in H$, agent $p$'s bundle does not change, so she possesses $a_{p,q}$,
    and since for any agent $j$, $a_{p,q}\greatereqval{p} a_{p,j}$, $p$ does not resent anyone.

    Moreover, By \cref{lem:treebreaker2}, $q=h_0$ remains non-resented.
    Since $\X$ was \pone\ at first, and $p$ resented agent $q$, agent $q$ did not resent anyone, so unit bundle $a_{q,p}$ was agent $q$'s favorite unit bundle. 
\end{proof}

\begin{lemma}\label{lem:added_strong_support}
    Given a \pone\ allocation $\X$, if there exists a support pair $(s,t)$ such that 
    agent $s$ does not resent any agent, and for every $p\to q$ with $p\ne t$, 
    we have $X_t \greatereqval{t} D_t(\X) \sumsq a_{t,p}$,
    we can construct a complete \efx\ allocation in polynomial time.
\end{lemma}
\begin{proof}
	As long as there exists $i \to j$ such that $b_{j,i} \sumsq A_j(\X) \greaterval{j} X_j = a_{j,i}$  
	update the allocation with update rule (U1), defined in \cref{up rule:u1}.
	By \cref{lem:cycle_support}, the allocation remains \pone, and no new resent is created.
	Hence, agent $s$ still does not resent anyone. 
    Note that since agent $s$ did not resent anyone, agent $s$'s bundle does not change, so $(s,t)$ remains a support pair.
	Moreover, every agent would have a bundle with a value not less than before.
	Therefore, we still have that for every $p\to q$, $X_t \greatereqval{t} D_t(\X) \sumsq a_{t,p}$.

    Note that at this point, update rule (U1) is not applicable. Moreover, since the allocation is \pone, without loss of generality we can assume that for every two resented vertices $u$ and $v$, the function $R(u, v)$ is well-defined by \cref{lem:strong_A}.
    
    Next, we enter the dumping phase to allocate all the remaining unit bundles:

	\paragraph{Dumping Phase.}
	We allocate every remaining unallocated unit bundle by the following rules.
	Here, a \emph{root} means a non-resented agent, and for a root $p$, the set $R_p$ denotes the agents resented by $p$. Every unallocated good gets allocated based on one of the following rules:

	\begin{enumerate}[label=(\roman*), leftmargin=2.2em]
		
		\item \textbf{Between $s$ and other roots.}
		Since $(s,t)$ is a support pair, we have $X_{s} \subseteq E_{s, t}$.
		We assign \dumpp{E_{s, t} \setminus X_s}{t}. For every other root $p\notin\{t,s\}$, 
		let \dumpu{a_{p,s}}{p} and \dumpp{b_{p,s}}{s}.

		\item \textbf{Root to root (except $s$).}
		Let $p$ and $r$ be two distinct roots such that $p,r \ne s$. 
		If at least one of the unit bundles is already allocated to one of the roots, allocate the other unit bundle to the other root. 		
		Otherwise, allocate one unit bundle in $E_{p,r}$ to $p$ and the other to $r$ arbitrarily.				

		\item \textbf{Root to its own children.}
		For every root $p$ and every $u\in R_p$, assign \dumpp{b_{p,u}}{p}. Also, we already have $X_u = a_{u,p}$.
		
		\item \textbf{Root to others' children.}
		Let $p$ and $r$ be two distinct roots, and let $u\in R_r$.
		If $p\in \{s,t\}$, let \dumpu{a_{p,u}}{p} and \dumpp{b_{p,u}}{r}.
		Also, if $p\notin\{s,t\}$, assign \dumpu{a_{p,u}}{p} and \dumpp{b_{p,u}}{s}.
		
		\item \textbf{Between two children of one tree.}
		Let $p$ be a root and let $u,v\in R_p$ be distinct. Then, assign \dumpp{a_{u,v}}{p} and \dumpp{b_{u,v}}{s}. Note that we choose the order of $u,v$ in order to get the unit bundles
		$(a_{u,v},b_{u,v})$ arbitrarily. Also, since $s$ does not resent anyone, we should have $p\ne s$.
		
		\item \textbf{Between two children of two trees.}
		Let $p$ and $r$ be two distinct roots, and let $u\in R_p$ and $v\in R_r$.
		Assign \dumpp{a_{u,v}}{s} and \dumpp{b_{u,v}}{R(u,v)}. Note that since $s$ does not resent any agent, we have $p, r\ne s$.		
	\end{enumerate}	
	This completes the allocation.     We have depicted the proposed allocation in \cref{fig:stageB}.

    \begin{figure}[h]
        \centering
\begin{center}
\begin{tikzpicture}[
    scale=1.3,
    every node/.style={
        circle,
        draw,
        inner sep=2pt,
        minimum size=10mm
    },
    main/.style={minimum size=11mm},
    envy/.style={->, thick},
    dashededge/.style={dashed, thick},
    pointydashed/.style={dashed, ->, thick}
]

\node[main] (i0) at (0,0) {$i_0$};

\node[main] (i1) at (1.8,2) {$i_1$};
\node[main] (i2) at (1.8,-2) {$i_2$};

\node[main] (j0) at (8,2) {$j_0$};
\node[main] (j1) at (8,-2) {$j_1$};

\node[main] (s) at (3.5,-4.8) {$s$};
\node[main] (t) at (5.8,-4.8) {$t$};

\node (a12) at (1.8,0.7) {$a_{i_1,i_2}$};
\node (b12) at (1.8,-0.7) {$b_{i_1,i_2}$};

\node (a1j0) at (4.1,2) {$b_{j_0,i_1}$};
\node (b1j0) at (6.3,2) {$a_{j_0,i_1}$};

\node (a2j1) at (4.3,-2) {$a_{i_2,j_1}$};
\node (b2j1) at (5.6,-2) {$b_{i_2,j_1}$};

\node (atj1) at (6.6,-3.9) {$a_{t,j_1}$};
\node (btj1) at (7.2,-3.1) {$b_{t,j_1}$};

\node (asj1) at (4.5,-4.1) {$a_{t,j_1}$};
\node (bsj1) at (5.6,-3.4) {$b_{t,j_1}$};

\draw[envy] (i0) -- (i1);
\draw[envy] (i0) -- (i2);
\draw[envy] (j0) -- (j1);


\draw[dashededge] (i1) -- (a12);
\draw[dashededge] (a12) -- (b12);
\draw[dashededge] (b12) -- (i2);

\draw[dashededge] (i1) -- (a1j0);
\draw[dashededge] (a1j0) -- (b1j0);
\draw[dashededge] (b1j0) -- (j0);

\draw[dashededge] (i2) -- (a2j1);
\draw[dashededge] (a2j1) -- (b2j1);
\draw[dashededge] (b2j1) -- (j1);

\draw[dashededge] (atj1) -- (btj1);
\draw[dashededge] (btj1) -- (j1);

\draw[dashededge] (asj1) -- (bsj1);
\draw[dashededge] (bsj1) -- (j1);


\draw[pointydashed] (atj1) -- (t);
\draw[pointydashed] (btj1) -- (j0);

\draw[pointydashed] (asj1) -- (s);
\draw[pointydashed] (bsj1) -- (j0);

\draw[pointydashed] (a12) -- (i0);
\draw[pointydashed] (b12) -- (s);

\draw[pointydashed] (a1j0) -- (s);
\draw[pointydashed] (b1j0) -- (j0);

\draw[pointydashed] (a2j1) -- (s);
\draw[pointydashed] (b2j1) -- (j0);

\end{tikzpicture}
\end{center}
            \caption{This figure illustrates an example of the allocation proposed in the dumping phase of \cref{lem:added_strong_support}. In this example, there are seven agents, namely $i_0,i_1,i_2,j_0,j_1,s,t$. Agent $i_0$ envies both $i_1$ and $i_2$, while agent $j_0$ envies $j_1$, and $(s,t)$ is a support pair.
            Envy relations are represented by solid directed edges.
            Dashed directed edges indicate the recipient of each unit bundle in the proposed allocation. Dashed non-directed lines indicate the agents to whom the corresponding unit bundles belong.\protect\\
            The pair $(i_2,j_1)$ illustrates the case of two children of two distinct trees.
            We select one of the partitions of $E_{i_2,j_1}$ arbitrarily, say $(a_{i_2,j_1}, b_{i_2,j_1})$, and assign $a_{i_2,j_1}$ to $s$ and $b_{i_2,j_1}$ to $R(i_2,j_1)\in 
            \{j_0,i_0\}$.
            Here we assumed that $R(i_2,j_1)=j_0$.\protect\\
            The pair $(i_1,i_2)$ illustrates the case of two children of one tree: $i_0$ envies $i_1$ and $i_2$. We select one of the partitions of $E_{i_1,i_2}$ arbitrarily, say $(a_{i_1,i_2}, b_{i_1,i_2})$, and assign $a_{i_1,i_2}$ share to $i_0$ and $b_{i_1,i_2}$ to $s$.
            }
        \label{fig:stageB}
    \end{figure}
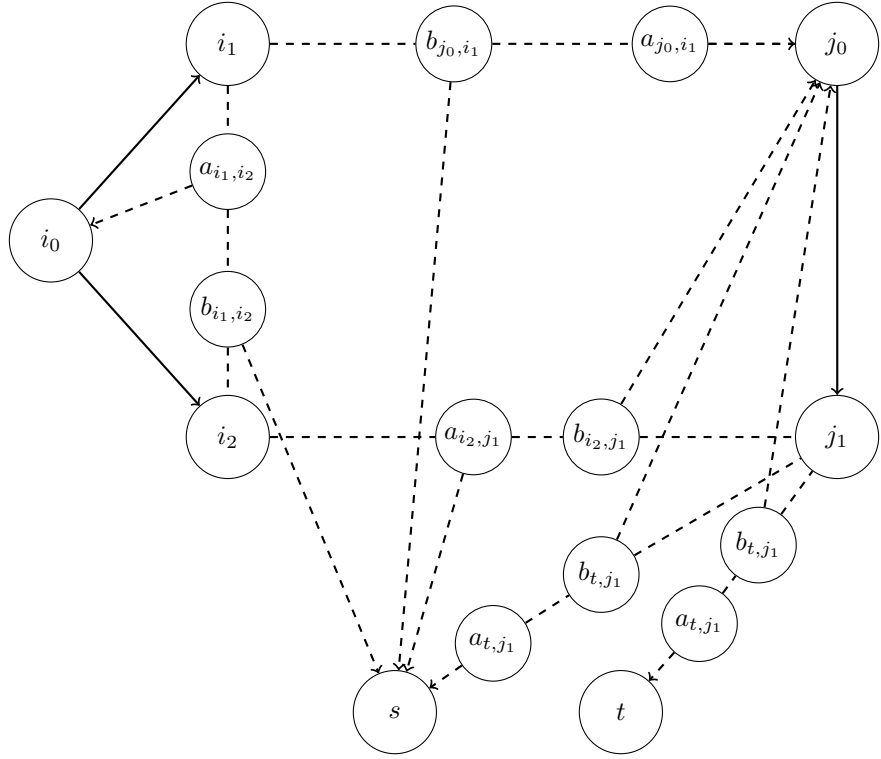

	\paragraph{Correctness of Dumping Phase.}
	Denote the allocation right before dumping phase  by $\X$, and denote the allocation at the end by $\xp$.
	Next, we prove that the final allocation, $\xp$, is \efx.
	
	First, note that dumping rules follow our general dumping rules properties, so by \cref{lem:dumprule},
\begin{enumerate}
	\item For every agent $p$, we have $X'_p \greatereqval{i} X_p$.
     \item If agent $q$ is resented in $\X$, then $q$ is not strongly envied by any agent in $\xp$.
	\item \label{efx3} For every resent $p\to q$ in $\X$, agent $q$ does not envy agent $p$ in $\xp$. 
	\item No one envies agent $s$, and agents $s$ does not strongly envy anyone.  
\end{enumerate}

    Therefore, it suffices to show that no agent envies an agent $p \neq s$ who was non-resented in $\X$. 
    
    First, we argue that agent $t$ does not envy $p$. Note that for any non-resented agent $r \neq \{t, p\}$, $X_p' \cap E_{r, t} = \emptyset$. Therefore, we have $$X'_p \cap E_t = (X'_p \cap E_{t, p}) \bigcup_{q \in R(\X)} (X_p' \cap E_{t, q}) \lowereqval{t} a_{t,p} \sumsq D_t(\X) \lowereqval{t} X_t \lowereqval{t} X'_t,$$
    where the first inequality comes from the fact that for a resented agent $q$, we have
	$X'_p \cap E_{t,q} \in \{b_{t,q}, \emptyset\}$ and $b_{t, q} \lowereqval{t} b_{q, t}$, and the second inequality comes from the fact that (U2) was not applicable before the dumping phase.

    Next, we show that any non-resented agent $p \ne s$ in $\X$, cannot be envied by a non-resented agent $r \notin \{s, t\}$. 
    Note that we know $X'_p\cap E_r$ is a single unit bundle from $E_{p , r}$, 
    so $X'_p\cap E_r \lowereqval{r} a_{r,p} \lowereqval{r} X_r \lowereqval{r} X'_r$,
    where the second inequality comes from the fact that $p$ was non-resented in $\X$.

    Finally, we show that an agent $v$ who was resented in $\X$ does not envy a non-resented agent $p \neq s$. For a resented agent $v$, by \cref{lem:dumprule}, we only need to consider the case where $r\to v$ for $r\ne p$.
	If there exists a $p\to u$ such that $R(u,v)=p$, then 
	$$X'_p \cap E_v \subseteq a_{p,v} \sumsq   \bigsumsq_{q \in R_p} (X'_p \cap E_{v,q})
			\lowereqval{v}  a_{v,p} \sumsq D_v(\X)  \lowereqval{v} X_v = X'_v,$$
	where the last inequality comes from the fact  $R(u,v)=p$.
	Also, if for every $p\to u$, we have $R(u,v)=r$, then 
	$$X'_p \cap E_v =a_{p,v} \lowereqval{v} X'_v.$$

    Hence, the final allocation is a complete \efx\ allocation.
\end{proof}

\begin{lemma}\label{lem:strong_support}
    Given a \pone\ allocation $\X$, if there exists a support pair $(s,t)$ such that every agent except $t$ only has exactly one unit bundle,
    and agent $s$ does not resent any agent, we can construct a complete \efx\ allocation in polynomial time.
\end{lemma}
\begin{proof}
	We first update the allocation with the following update rule.
	While there exists $p\to q$ such that $a_{t,p} \sumsq B_t\setminus E_{p,t} \greaterval{t} X_t$,
	update the allocation with update rule (U2), defined in \cref{up rule:u2}.
	Then, by \cref{lem:t update},
	after applying this update rule as much as possible, allocation remains \pone, every agent except $t$ has exactly one unit bundle, $(s,t)$ remains a support pair,
	and $s$ does not resent any agent. Clearly, $v_t(X_t)$ is strictly increasing after each round of applying (U2),
    and since $v_t(a_{t,p} \sumsq B_t\setminus E_{p,t})$ is independent of the allocation and is a function of agents $t$ and $p$, and since agent $t$ is fixed, 
    this process terminates after at most $n$ steps. 
    Hence, at the end for every $p\to q$, 
    we have $X_t \greatereqval{t} (B_t \setminus E_{t,p}) \cup a_{t,p} \greatereqval{t} D_t(\X) \sumsq a_{t,p}$. Hence, by \cref{lem:added_strong_support}, we get a complete $\efx$ allocation.
\end{proof}

    Next, we are ready to get a complete $\efx$ allocation given an allocation in \stage\ $B$.

\begin{lemma}\label{lem: stageB}
	If an allocation is in \stage\ B, we can construct a complete \efx\ allocation in polynomial time.
\end{lemma}
\begin{proof}
    Denote such an allocation by $\X$.
	Since allocation $\X$ is in \stage\ B, it is \pone, and there exists $j\to i$ such that
        $$X_j \sumsq A_j(\X) \greatereqval{j} a_{j,i} .$$ 
	Then, update the allocation as follows: (let $\xp$ be the allocation after the update)
    \begin{align*}
        X_j' &\takes X_j \sumsq A_{j}(\X)
    \end{align*}
	Therefore, agent $j$ does not resent agent $i$ anymore and agent $i$ will become non-resented.  
	Also, we have $X'_i=a_{i,j}$ because $\X$ was \pzero.
    Hence, $(i,j)$ is a support pair. Also, note that allocation $\X$ was \pthree, so now, except for agent $j$, every agent possesses exactly one unit bundle.
	Moreover, since the length of any path in the resent graph was at most one in $\X$, agent $i$ did not resent anyone. Hence, agent $i$ does not resent anyone in $\xp$ as $X_i = X_i'$. Thus, the conditions of \cref{lem:strong_support} are satisfied, and we can construct a complete \efx\ allocation in polynomial time.
\end{proof}

\section{\Stage\ H, Final \Stage}\label{sec:final}

In this section, we show that if a \pthree\ allocation is not in any stage A-G, we can construct a complete $\efx$ allocation. Let $\X$ be the initial allocation that is not in any stage A-G.

	 Since $\X$ is not in \stage\ G, we can assume that for every four distinct agents $k,i,\ell,$ and $j$ 
	such that $k\to i$, $\ell \to j$, at least one of the following holds:
        \begin{enumerate}
            \item $[C_j(\X)\cap E_{j,R_k}] \sumsq [D_j(\X)\setminus E_{j,R_k}] \sumsq a_{j,k} \lowereqval{j} X_j,$ 
            \item $[C_i(\X)\cap E_{i,R_\ell}] \sumsq [D_i(\X)\setminus E_{i,R_\ell}] \sumsq a_{i,\ell} \lowereqval{i} X_i$.
        \end{enumerate}

Therefore, we are now able to define the following for every such four agents:

\begin{definition}\label{lem:uij}
	Given an allocation that is not in any \stage\ A-G, for every four distinct agents $k,i,\ell,$ and $j$ such that $k\to i$, and $\ell \to j$, we define $U(i,j)\in \{k,\ell\}$ 
        \begin{enumerate}
            \item $U(i,j)=k$ if $[C_j(\X)\cap E_{j,R_k}] \sumsq [D_j(\X)\setminus E_{j,R_k}] \sumsq a_{j,k} \lowereqval{j} X_j,$
            \item $U(i,j)=\ell$ otherwise, i.e., $$[C_j(\X)\cap E_{j,R_k}] \sumsq [D_j(\X)\setminus E_{j,R_k}] \sumsq a_{j,k} \greaterval{j} X_j,$$
					and $$[C_i(\X)\cap E_{i,R_\ell}] \sumsq [D_i(\X)\cap E_{i,R_\ell}] \sumsq a_{i,\ell} \lowereqval{i} X_i.$$
        \end{enumerate}
\end{definition}

\begin{lemma}\label{lem: stageH}
	Given an allocation $\X$ in \stage\ H, we can construct a complete \efx\ allocation in 
    polynomial time.
\end{lemma}

\begin{proof}
	Let $\xp$ denote the allocation obtained from $\X$ after the dumping rules below.
	
	\paragraph{Dumping Phase.}
    We allocate every remaining unallocated unit bundle by the following rules.
	Here, a \emph{root} means a non-resented agent, and for a root $p$, the set $R_p$ denotes the agents resented by $p$ in $\X$. Every unallocated good gets allocated based on one of the following rules:
	
	\begin{enumerate}[label=(\roman*), leftmargin=2.2em]
		
		\item \textbf{Root to its own children.}
		For every root $p$ and every $u\in R_p$, assign \dumpp{b_{p,u}}{p}.
		
		\item \textbf{Between two children of one tree.}
		Let $p$ be a root and let $u,v\in R_p$ be distinct.
		Let $r$ be the root of another resent tree with non-zero height, which exists since $\X$ is not in \stage\ D.
		Assign \dumpp{a_{u,v}}{p} and \dumpp{b_{u,v}}{r}.
		
		\item \textbf{Between two children of two trees.}
		Let $p$ and $r$ be two distinct roots and let $u\in R_p$ and $v\in R_r$.
		Assign \dumpp{a_{u,v}}{U(u,v)}. Let $i \in \{p, r\}$ be such that $i \neq U(u, v)$. Then, assign \dumpp{b_{u,v}}{i}.
		
		\item \textbf{Root to others' children.}
		Let $p$ and $r$ be two distinct roots and let $v\in R_r$.
		Assign \dumpu{a_{p,v}}{p} and \dumpp{b_{p,v}}{r}.
		
		\item \textbf{Root to root.} \label{H:roottoroot}
		Let $p$ and $r$ be two distinct roots. If at least one of the unit bundles is already allocated to one of the roots, allocate the other unit bundle to the other root. 		
		Otherwise, allocate one unit bundle in $E_{p,r}$ to $p$ and the other to $r$ arbitrarily.	
		
	\end{enumerate}
	
	This completes the allocation.
	
	\paragraph{Correctness of Dumping Phase.}
	We prove that $\X'$ is \efx. Since $\X$ is a \pthree\ allocation and \stage\ E does not apply,
	the global properties before dumping (\Cref{pbd}) hold for $\X$.
	Moreover, our dumping rules satisfy the general dumping-structure requirements \ref{pd1}--\ref{pd8}.
	Therefore, all conclusions of \Cref{lem:dumprule} apply:
	\begin{enumerate}
		\item For every agent $p$, $X'_p \greatereqval{p} X_p$.
		\item Nobody strongly envies any agent who was resented in $\X$.
		\item For every resent edge $p\to q$ in $\X$, agent $q$ does not envy agent $p$ in $\X'$.
	\end{enumerate}

    By the second statement of \cref{lem:dumprule}, in order to prove that the final allocation is $\efx$, it suffices to show that any agent $p$ who was non-resented in $\X$, cannot be envied in $\xp$.

    Let $r \neq p$ be a non-resented agent in $\X$. Note that if $X_p' \cap E_r \subseteq E_{p, r}$, then $X_p' \cap E_r$ must be a single unit bundle from $E_{p, r}$,
    so $X_p' \lowereqval{r} a_{r,p} \lowereqval{r} X_r \lowereqval{r} X'_r$, where the second inequality comes from the fact that $p$ was non-resented in $\X$. 
    Otherwise, if there exists a unit bundle from $E_{r, j}$ for $j \notin \{p,r\}$ that has been dumped to $p$, it must have been done by Rule (iv), so $j\in R_p$. 
    Since $\X$ was not at \stage\ F, we get that either $r$ possesses $a_{r, p}$ in $\X$ or 
    $a_{r, p} \cup D_r(\X) \lowereqval{r} X_r$. 
    If $r$ possesses $a_{r, p}$ in $\X$, then by rule \ref{H:roottoroot}, we get that 
    $a_{r,p}\subseteq X'_r$, so
    $$X'_p \cap E_r \lowereqval{r} b_{r, p} \sumsq \bigsumsq_{i \in R_p} b_{r, i} \lowereqval{r} a_{r, p} \sumsq \bigsumsq_{i \in R_p} a_{r, i} \lowereqval{r} X'_r,$$
    where the last inequality comes from the fact for every resented agent $i$, agent $r$ gets
    $a_{p,i}$ by rule (iv).
    Hence, $r$ does not envy $p$. 
    Otherwise, if $a_{r, p} \cup D_r(\X) \lowereqval{r} X_r$, we have 
    $$X'_p \cap E_r \lowereqval{r} a_{p, r} \cup \bigsumsq_{i \in R_p} b_{r, i}
    \lowereqval{r} a_{p, r} \cup \bigsumsq_{i \in R_p} b_{i,r}    \lowereqval{r} a_{r, p} \cup D_r(\X) \lowereqval{r} X_r \lowereqval{r} X'_r.$$
   Hence, $r$ does not envy $p$.

    Now, let $q$ be a resented agent in $\X$. We show that $q$ does not envy $p$. By the third statement of \cref{lem:dumprule}, we only need to consider the case where $r \to q$ and $r \neq p$.
    Hence, by rule (iv), we have $a_{p, q} \subseteq X'_p$.
    
    If $p$ did not resent anyone in $\X$, we get $X'_p \cap E_q = a_{p, q} \lowereqval{q} X_q=X'_q$, so $q$ does not resent $p$. Thus, we can assume that $p$ resents at least on other agent.
    If there exists an agent 
    $u \notin \{p,q\}$ such that a unit bundle from $E_{q, u}$ has been dumped to $p$, it must come from Rules (ii) or (iii), so $u$ is resented in $\X$. Hence,
       $$X_p' \cap E_q =  (X_p' \cap E_{q,p}) \cup  \bigcup_{u \in R(\X)} (X_p' \cap E_{q,u}).$$
    
    We distinguish between two cases on whether there existed an agent $i$ resented by $p$ in $\X$
    such that $U(q, i) = p$. 

    First, consider the case where for every $p\to i$, we have $U(q, i) \ne p$. Then,
    for every $U\in R(\X)$,  we have $(X_p' \cap E_{q,u})\in \{b_{q,u},b_{u,q}, \emptyset\}$, so 
    $(X_p' \cap E_{q,u})\lowereqval{q} b_{u,q}$.
    Hence,
    $$X_p' \cap E_q =  (X_p' \cap E_{q,p}) \cup  \bigcup_{u \in R(\X)} (X_p' \cap E_{q,u})
    \lowereqval{q} a_{p, q} \cup  \bigcup_{u \in R(\X)} b_{u,q}
    \lowereqval{q} a_{q, p} \sumsq D_q(\X)  \lowereqval{q} X_q = X'_q,$$ 
    where the first and the second inequality follows from \cref{obs: ub-ineq} and the third inequality comes from the fact that $\X$ was not in \stage\ C, and the fact that $p$ resents at least one agent in $\X$.

   Now, consider the case where there existed an agent $i$ resented by $p$ in $\X$ 
   such that $U(q, i) = p$. 
   Then, by \cref{lem:uij} and the fact that $\X$ is not in \stage\ G, we have 
   $[C_q(\X)\cap E_{q,R_p}] \sumsq [D_q(\X)\setminus E_{q,R_p}] \sumsq a_{q,p} \lowereqval{q} X_q$. 
   As a result,
   \begin{align*}
        X_p' \cap E_q &=  
        (X_p' \cap E_{q,p}) \cup \bigcup_{u \in R_p} (X_p' \cap E_{q,u}) 
        \cup \bigcup_{u \in R(\X)\setminus R_p} (X_p' \cap E_{q,u}) \\
        & \lowereqval{q} a_{p,q} \cup \bigcup_{u \in R_p} a_{q,u} \cup
        \bigcup_{u \in R(\X)\setminus R_p} b_{u,q} \\
        &\lowereqval{q} a_{q,p} \cup [C_q(\X)\cap E_{q,R_p}] \sumsq [D_q(\X)\setminus E_{q,R_p}]\\ &\lowereqval{q} X_q = X_q',   
   \end{align*}
    and hence, $q$ does not envy $p$.
    
    Hence, the final allocation is a complete \efx\ allocation.
\end{proof}

\section{Conclusion}
    We studied the existence of $\efx$ allocations in the multi-graph valuation model, where each good is valued by at most two agents, but pairs of agents may share multiple goods. This model generalizes graphical valuations introduced by \citet{CFKS23} and captures a broad class of fair division instances with local interactions. While exact $\efx$ existence was previously known only for simple graphs or under additional structural or approximation assumptions, extending these guarantees to multi-graphs has remained a central open problem.
    
    In this work, we resolved this question for cancelable valuations by proving that $\efx$ allocations always exist in multi-graph instances. Our result applies to an arbitrary number of agents and does not impose restrictions on the structure of the underlying graph. This places multi-graph valuations among the limited classes for which exact $\efx$ existence is known and advances the understanding of fairness in settings with overlapping agent interests. Moreover, our proof is algorithmic and computes such allocations in polynomial time when the valuation functions are cancelable.
    
    Several directions remain open. A natural question is whether $\efx$ allocations exist for multi-graph valuations under broader classes of valuations, such as monotone or subadditive valuations. Finally, it would be interesting to study how $\efx$ fairness in multi-graph valuations interacts with efficiency objectives, such as maximizing social welfare or Nash social welfare.

    We believe that our results provide a useful step toward a general understanding of $\efx$ in structured valuation models and may inform future work on both the existence and computation of fair allocations beyond the simple graph setting.

\newpage

\bibliographystyle{plainnat}
\bibliography{mybibliography}

\newpage
\appendix

\section{Computing Unit Bundles in Polynomial Time}
\label{sec:computing-unit-bundles}

	In this section, we provide an algorithm for computing unit bundles, and we show that it runs in polynomial time with respect to 
	the number of goods. In order to do that we first provide the PR algorithm \cite{PR20} and show some of its properties.
	In this section, given an agent $i$ and a bundle $S$, we use $\ell^i_r(S)$ to denote the $r$'th least valued good in $S$
	with respect to agent $i$'s valuation function. Moreover, when the valuation function is fixed, we 
    drop the upper index and write $\ell_r(S)$.

\paragraph{The Plaut-Roughgarden (PR) algorithm.}
    This local search algorithm, dubbed the \emph{PR algorithm} by \citet{ACGMM23}, takes as input a set of bundles $\Y=(Y_1, Y_2,..., Y_k)$ 	and a monotone valuation function $v$, and returns a set of bundles $\yp=(Y'_1, Y'_2,..., Y'_k)$ that is \efxf\ for valuation $v$ in $\yp$ (i.e., all of 
	its bundles are \efxf\ for $v$ in $\yp$). If $Y_j$ is the least valuable bundle in $\Y$ w.r.t.\ valuation $v$, the algorithm checks whether $Y_j$ is
	\efxf\ w.r.t.\ $v$ in $\Y$. If it is, the algorithm terminates and returns $\Y$. On the other hand, if $Y_j$ is not \efxf\ in $\Y$, then there 
	must exist another bundle $Y_i$ such that $v(Y_j)< v(Y_i \setminus g)$ for some good $g\in Y_i$. The PR algorithm then removes $g$ 
	from $Y_i$, it adds it to $Y_j$, and it repeats the same sequence of steps until it reaches an \efxf\ set of bundles.
    Each time the PR algorithm moves a good, the minimum value over all the bundles, i.e., $\min_{Y\in \Y}v(Y)$, weakly increases.
\begin{observation}\label{lem:PR increase}
\label{PR}
    Let $\yp=(Y'_1,...,Y'_k)$ be the set of bundles returned by the PR algorithm with input $\Y=(Y_1,...,Y_n)$ and valuation function $v$. Then:
    \begin{equation*}
    \min(v(Y_1),...,v(Y_k)) \leq \min(v(Y'_1),...,v(Y'_k)).
    \end{equation*}
\end{observation}

	\citet{AGS25} slightly modified this algorithm, and showed that, given a cancelable valuation function, the
	modified PR algorithm runs in polynomial time if the number of bundles is two.
    
	\paragraph{Modified PR  \citep{AGS25}:} The only difference is that  the good removed from $Y_i$ 
	and added to $Y_j$ is: $$g\in \arg\max_{x\in Y_i} v(Y_i \setminus x).$$

    When the valuation function $v$ is fixed, for two subsets of goods $S,T$
    we denote by $S \prec T$ whenever $v(S)<v(T)$ and we denote by $S\preceq T$ whenever $v(S)\leq v(T)$.

\begin{claim}\label{claim: no transfer}
	\citep{AGS25} Let the valuation function be cancelable. 
    Suppose that during the execution of the modified PR algorithm, at some step, we have a set of bundles
    $(\ya,\yb)$, such that $\ya\setminus x \succ \yb$ and $Y_1\setminus x \prec \yb \cup x$, where $x = \ell_1(\ya)$. 
    Then the algorithm removes $x$ from $\ya$ and adds it to $\yb$.
After the transfer of $x$, only goods $y$ with $y \prec x$ will be transferred. 
\end{claim}
    
\begin{proof}
    Note that by Lemma 3.13 in \cite{AGS25}, for a cancelable valuation function and an arbitrary subset of goods,
    we have $\ell_1(S) = \arg\max_{g\in S} v(S \setminus g)$, so by modified PR algorithm, 
    at the starting step, the algorithm removes $x$ from $\ya$ and adds it to $\yb$.
    
    Next, suppose on the contrary that the second part of lemma does not hold, and let $y$ be the first good being transferred after $x$, such that $x \preceq y$. 
    Then, for any good $z$ being transferred after $x$ and before $y$, if any, we have $z \prec x$. 
    Suppose that the bundles exactly before the transfer of $y$ are $(\yap,\ybp)$.
    Since we had $x = \ell_1(\ya)$, for every $g \in \ya \setminus x$, 
    we have $g \succeq x$, so these goods have not been transferred after $x$ and before $y$,
    hence, $\ya \setminus x \subseteq \yap$.
    Also since $\yap \sqcup \ybp = (\ya\setminus x) \sqcup (\yb \cup x)$, we get that  
    $\ybp \subseteq \yb \cup x$. 
    In other words, all goods that are transferred between $x$ and $y$ are transferred from $\yb\cup x$ to $\ya \setminus x$.
    
    We consider two cases based on which set, $\yap$ or $\ybp$, $y$ belongs to. 
     
    \vspace{5pt}
    \noindent $\mathbf{\bullet}$ \textbf{\caseA.} $\mathbf{y\in \yap}$:
        Since $y$ is being transferred from $\yap$, by the definition of PR algorithm, we have 
        that $y= \ell(\yap)$, so for every $g\in\yap$ we have that $g \succeq y \succeq x$,
        so $g \succeq x$. Since for every $z$ being transferred after $x$ and before $y$
        we have that $z \prec x$, we get that there is no good in $\yap$ that has been transferred from
        the other bundle to $\yap$ after
        the transfer of $x$, so $\yap = \ya \setminus x$, and $\ybp = \yb \cup x$.
        So we get $\yap = \ya \setminus x \prec \yb \cup x = \ybp$,
        which is a contradiction to the PR algorithm, since the algorithm transfers a good from the most valuable set to the least valuable one. 

    \vspace{5pt}
    \noindent $\mathbf{\bullet}$ \textbf{\caseB.} $\mathbf{y\in \ybp}$:
    Since for every good $z$ that has been transferred after $x$, we have that $z \prec x$,
    we get that $x \in \ybp$. Recall that $x \preceq y$. 
    If $x \prec y$, then $\ell(\yap) \neq y$. Therefore, $x$ has equal value to $y$, 
    which means that $\ybp \setminus y$ and $\ybp \setminus x$ 
    have the same value \footnote{This holds because if it was w.l.o.g. $\ybp \setminus y \prec \ybp \setminus x$, then by subtracting $\ybp \setminus \{y,x\}$ from both sides, cancelability would give $x\prec y$.}. 
    Moreover, since
    $\ybp \subseteq \yb\cup x$, we get $\ybp\setminus x\subseteq \yb$, 
    which implies $\ybp \setminus x \preceq \yb$, 
    and therefore, $\ybp \setminus y \preceq \yb$. 
    Thus,
    \begin{equation}
        \ybp \setminus y \preceq \yb \prec \ya \setminus x \preceq \yap\,, \notag
     \end{equation}
    where the last inequality is due to $\ya \setminus x \subseteq \yap$.
    So we have $\ybp \setminus y \prec \yap$, which is a contradiction, 
    since in the PR algorithm, when the algorithm removes some good $y$ from
    $\ybp$ and adds it to $\yap$, it should be that $\ybp \setminus y \succ \yap$.
	Hence, the proof is complete.
    \end{proof}

\begin{lemma}\label{lem: pr polynomilal}
	\citep{AGS25}
    The modified PR algorithm runs in polynomial time 
    when the input is a set of two disjoint bundles 
    and a cancelable valuation function.
\end{lemma}
\begin{proof}
	Suppose we execute the modified PR algorithm given $(\ya,\yb)$ as input. 
    During the execution of the modified PR algorithm, there can be
    at most $m$ consecutive transfers of a good from $\ya$ to $\yb$ in a row, since there are $m$ goods.
    Also, for the last good $x$ transferred from $\ya$ to $\yb$ in a row of transfers, we must have  that
    $\ya \setminus x \succ \yb$, and if algorithm does not terminate it also holds that
    $\ya \setminus x \prec \yb \cup x$. Hence for every at most $m$ removal of goods from $\ya$ in a row,
    there exists at least one good (the last one) that according to \cref{claim: no transfer}, will not come back to $\ya$
    ever again. Hence there could be at most $m^2$ removals of goods from $\ya$.
    Similarly, this holds for $\yb$, too, therefore the modified PR algorithm runs in polynomial time.
\end{proof}

\begin{lemma}\label{cancelablevalprops}
    If $v$ is a cancelable valuation function, and $(A_1,A_2)$ and $(B_1,B_2)$ are two partitions of a subset of goods $S$, then:
    \begin{enumerate}
        \item If $\min (v(B_1),v(B_2)) < \min (v(A_1),v(A_2))$, then $\max(v(A_1),v(A_2))\leq \max (v(B_1),v(B_2))$.
        \item            If $|v(B_1)-v(B_2)| \leq |v(A_1)-v(A_2)|$, then 
        $\min(v(A_1),v(A_2))\leq \min (v(B_1),v(B_2))$, 
        and
        
        $\max(v(A_1),v(A_2))\geq \max (v(B_1),v(B_2))$.
    \end{enumerate}
\end{lemma}
\begin{proof}

\begin{enumerate}
    \item First, assume that $\min (v(B_1),v(B_2)) < \min (v(A_1),v(A_2))$.
    Without loss of generality assume $v(B_1)=\min (v(B_1),v(B_2))$, 
    so $B_1 \prec A_1$ and $B_1 \prec A_2$.

    Let $i,j \in\{1,2\}$ and $i\ne j$. We show  $B_1 \cap A_i \prec B_2 \cap A_j$. 
    Towards a contradiction, suppose $B_1 \cap A_i \succeq B_2 \cap A_j$.
    Then, since $(B_1 \cap A_j) \cap (B_2 \cap A_j) = \emptyset$, 
    $(B_1 \cap A_j)  \cap (B_1 \cap A_i) = \emptyset$, and since the valuation function is cancelable, 
    by \cref{cancp}, we get 
    $$A_j = (B_1 \cap A_j) \cup (B_2 \cap A_j) \preceq (B_1 \cap A_j)  \cup (B_1 \cap A_i) = B_1
    \to  A_j \preceq B_1,$$
    which is a contradiction. Hence,  $B_1 \cap A_i \prec B_2 \cap A_j$.
    
    Then, since $(B_1 \cap A_i) \cap (B_2 \cap A_i)  = \emptyset$, 
    $(B_2 \cap A_j)  \cap (B_2 \cap A_i)  = \emptyset$, and since the valuation function is cancelable, 
    we get 
    $$A_i = (B_1 \cap A_i) \cup (B_2 \cap A_i) \preceq (B_2 \cap A_j)  \cup (B_2 \cap A_i) = B_2
    \to  A_i \preceq B_2,$$
    so since $i$ was chosen arbitrarily, we get $v(B_2) \geq \max (v(A_1),v(A_2))$.

    \item  Suppose $|v(B_1)-v(B_2)| \leq |v(A_1)-v(A_2)|$.
    Assume $\min (v(B_1),v(B_2)) < \min (v(A_1),v(A_2))$ for the sake of getting a contradiction.
    Then, by the first part of the lemma, we get $\max(v(A_1),v(A_2))\leq \max (v(B_1),v(B_2))$, which results in $|v(B_1)-v(B_2)| > |v(A_1)-v(A_2)|$, which is a contradiction.
    Hence, $\min(v(A_1),v(A_2))\leq \min (v(B_1),v(B_2))$. 
    
    If $\min(v(A_1),v(A_2))= \min (v(B_1),v(B_2))$ holds, then by $|v(B_1)-v(B_2)| \leq |v(A_1)-v(A_2)|$, we conclude
    $\max(v(A_1),v(A_2))\geq \max (v(B_1),v(B_2))$.
    Also, if $\min(v(A_1),v(A_2)) < \min (v(B_1),v(B_2))$, we get 
     $\max(v(A_1),v(A_2))\geq \max (v(B_1),v(B_2))$ by the first part of the lemma.    \qedhere
\end{enumerate}    
\end{proof}

    we use \texttt{PR}$(S_1, S_2,v_i)$ to denote the output of running the Modified PR algorithm on the initial partition $(S_1, S_2)$ with valuation $v_i$. 

\begin{lemma} \label{canceldif}
    Given a cancelable valuation function $v$, 
    if we execute $\texttt{PR}(B_1,B_2,v)$ and get $(A_1,A_2)$, then:
    \begin{enumerate}
        \item $\max(v(A_1),v(A_2))\leq \max (v(B_1),v(B_2))$.
        \item $|v(A_1)-v(A_2)| \leq |v(B_1)-v(B_2)|$.
    \end{enumerate}
\end{lemma}
\begin{proof}
\begin{enumerate}
    \item By $(A_1,A_2) = \texttt{PR}(B_1,B_2,v)$,
    we get $\min (v(B_1),v(B_2)) \leq \min (v(A_1),v(A_2))$.
    By \cref{cancelablevalprops}, if $\min (v(B_1),v(B_2)) < \min (v(A_1),v(A_2))$, then we get $\max(v(A_1),v(A_2))\leq \max (v(B_1),v(B_2))$.
    So assume $\min (v(B_1),v(B_2)) = \min (v(A_1),v(A_2))$.
    Without loss of generality assume $B_2 \preceq B_1$ and $A_2 \preceq A_1$.
    Therefore, $v(B_2) = v(A_2)$.

    Note that during the PR algorithm a good $g$ is transferred from some bundle $S$ to $T$ 
    only if $S\setminus g \succ T$, so since the minimum valued bundle's value never decreases at any step, the bundle that is losing a good would have a value strictly higher than the value of $B_2$.
    Hence, since the minimum value does not increase, during the PR algorithm, goods only transfer from $B_1$ to $B_2$ during the transfer.
    Therefore $A_1 = B_1 \setminus S \succ B_2$, and $A_2 = B_2 \cup S$ for some $S\subset B_1$. Hence,
    \begin{align*}
        v(B_1) &\geq  v(B_1 \setminus S) = v(A_1)\\
        v(B_1) &\geq v(B_2) = v(A_2),
    \end{align*}
    So $\max(v(A_1),v(A_2))\leq \max (v(B_1),v(B_2))$.

    \item  By combining $\max(v(A_1),v(A_2))\leq \max (v(B_1),v(B_2))$ and
    $\min (v(B_1),v(B_2)) \leq \min (v(A_1),v(A_2))$, we get 
    $|v(A_1)-v(A_2)| \leq |v(B_1)-v(B_2)|$. \qedhere
\end{enumerate}    
\end{proof}

	Next, we provide the algorithm that we use to compute unit bundles  in polynomial time.

\begin{algorithm}[H]
	\caption{Computing Unit Bundles of $E_{i,j}$}
	\label{alg:computing-unit-bundles}
	
	$t\gets1$\;
	$(Z^{(t)}_1,Z^{(t)}_2) \gets \texttt{PR}(E_{i,j}, \emptyset, v_i)$ \;
	$(Y^{(t)}_1,Y^{(t)}_2) \gets \texttt{PR}(Z^{(t)}_1,Z^{(t)}_2, v_j)$ \;

	\While{$\bigl|v_i(Y^{(t)}_1)-v_i(Y^{(t)}_2)\bigr|<\bigl|v_i(Z^{(t)}_1)-v_i(Z^{(t)}_2)\bigr|$}{
			$t\gets t+1$\;
			$(Z^{(t)}_1,Z^{(t)}_2) \gets \texttt{PR}(Y^{(t-1)}_1,Y^{(t-1)}_2,v_i)$\;
	\If{$\bigl|v_j(Y^{(t-1)}_1)-v_j(Y^{(t-1)}_2)\bigr|\leq \bigl|v_j(Z^{(t)}_1)-v_j(Z^{(t)}_2)\bigr|$}{
		$(Y^{(t)}_1,Y^{(t)}_2) \gets (Y^{(t-1)}_1,Y^{(t-1)})$\;
		break
		}
			$(Y^{(t)}_1,Y^{(t)}_2) \gets \texttt{PR}(Z^{(t)}_1,Z^{(t)}_2,v_j)$\;
	}
	$\{a_{i,j}, b_{i,j}\} \gets \{Y^{(t)}_1,Y^{(t)}_2\}$  such that $a_{i,j} \greatereqval{i} b_{i,j}$.\\
	$\{a_{j,i}, b_{j,i}\} \gets \{Z^{(t)}_1,Z^{(t)}_2\}$ such that $a_{j,i} \greatereqval{j} b_{j,i}$.
\end{algorithm}

\begin{observation}\label{ob: inv unit}
	Suppose we execute Algorithm~\ref{alg:computing-unit-bundles} with an arbitrary pair of agents $i,j$ with cancelable valuation function as input.
	Then, at every step $t$ of the while loop:
	\[
	E_{i,j}=Y^{(t)}_1\sqcup Y^{(t)}_2
	\qquad\text{and}\qquad
	E_{i,j}=Z^{(t)}_1\sqcup Z^{(t)}_2
	\]
	satisfy the following three properties:
	\begin{enumerate}
		\item
		The pair $(Y^{(t)}_1,Y^{(t)}_2)$ is \emph{EFX-feasible} for agent $j$.
		\item
		The pair $(Z^{(t)}_1,Z^{(t)}_2)$ is \emph{EFX-feasible} for agent $i$.
		\item $\bigl|v_j(Y^{(t)}_1)-v_j(Y^{(t)}_2)\bigr|\le \bigl|v_j(Z^{(t)}_1)-v_j(Z^{(t)}_2)\bigr|$.
	\end{enumerate}
\end{observation}
\begin{proof}
	The first two are because the bundles are the output of PR algorithm.
	For the last one, we either have  $\bigl|v_j(Y^{(t)}_1)-v_j(Y^{(t)}_2)\bigr|\le \bigl|v_j(Z^{(t)}_1)-v_j(Z^{(t)}_2)\bigr|$ or:
	$$(Y^{(t)}_1,Y^{(t)}_2) \gets \texttt{PR}(Z^{(t)}_1,Z^{(t)}_2,v_j),$$
	so by \cref{canceldif},
	we get $\bigl|v_j(Y^{(t)}_1)-v_j(Y^{(t)}_2)\bigr|\le \bigl|v_j(Z^{(t)}_1)-v_j(Z^{(t)}_2)\bigr|$.
\end{proof}

\begin{corollary}\label{cor:unitbundleoutput}
	Given any pair of agents $i$ and $j$ with cancelable valuation function as input to Algorithm~\ref{alg:computing-unit-bundles}, when Algorithm~\ref{alg:computing-unit-bundles} terminates:
	\begin{enumerate}
		\item
		The pair 	$\{a_{j,i}, b_{j,i}\}$ is an \emph{EFX-feasible} cut for agent $i$.
		\item
		The pair $\{a_{i,j}, b_{i,j}\}$ is an \emph{EFX-feasible} cut for agent $j$.
		\item      $v_i(a_{ij}) \ge \max(v_i(a_{ji}),v_i(b_{ji})) \ge \min(v_i(a_{ji}),v_i(b_{ji})) \ge v_i(b_{ij})$, i.e.,
        the $(a_{j,i},b_{j,i})$-partition is at least as balanced for $i$ as the $(a_{i,j},b_{i,j})$-partition.
		\item  $ v_j(a_{ji}) \ge \max(v_j(a_{ij}),v_j(b_{ij})) \ge \min(v_j(a_{ij}),v_j(b_{ij})) \ge v_j(b_{ji}),$, i.e.,
        the $(a_{i,j},b_{i,j})$-partition is at least as balanced for $j$ as the $(a_{j,i},b_{j,i})$-partition.
	\end{enumerate}
\end{corollary}
\begin{proof}
    When the while loop ends, by the condition of the while loop and by \cref{ob: inv unit},
    we get
        $\bigl|v_j(a_{i,j})-v_j(b_{i,j})\bigr|\le \bigl|v_j(a_{j,i})-v_j(b_{j,i})\bigr|$
        and
        $\bigl|v_i (a_{j,i})-v_i(b_{j,i})\bigr|\hspace{0.8 mm}\le \hspace{0.8 mm} \bigl|v_i(a_{i,j})-v_i(b_{i,j})\bigr|$.
    Hence, using \cref{cancelablevalprops}, we get
    \begin{align*}
     v_i(a_{ij}) &= \max(v_i(a_{ij}),v_i(b_{ij}))\geq \max(v_i(a_{ji}),v_i(b_{ji})) \\
     &\ge \min(v_i(a_{ji}),v_i(b_{ji})) \geq  \min(v_i(a_{ij}),v_i(b_{ij})) = v_i(b_{ij}).        
    \end{align*}
    Hence, the third property, and similarly fourth property hold as well.
\end{proof}

    Next, we show that Algorithm~\ref{alg:computing-unit-bundles} terminates in polynomial time in $m$ (the number of goods) when the valuation functions are cancelable.
	Let $k$ be the number of iterations of the while loop. We show  $k\leq m+1$, where $m$ is the number of goods in $|E_{i,j}|$.
	Without loss of generality, we can assume $Z^{(t)}_1 \greatereqval{i} Z^{(t)}_2$ and $Y^{(t)}_1 \greatereqval{i} Y^{(t)}_2$.

\begin{lemma}\label{lem:transfer more}
	Suppose that the while loops has repeated $k$ times, and $t+1<k$, then
	during execution of
			$$(Z^{(t+1)}_1,Z^{(t+1)}_2) \gets \texttt{PR}(Y^{(t)}_1,Y^{(t)}_2,v_i),$$
	some goods have been transferred from $\ybt$ to $\yat$, and some goods have been transferred from $\yat$ to $\ybt$.
\end{lemma}
\begin{proof}
    Note that if there is no transfer, then 
    $(Z^{(t+1)}_1,Z^{(t+1)}_2) =(Y^{(t)}_1,Y^{(t)}_2), $ so 
    $\bigl|v_j(Y^{(t)}_1)-v_j(Y^{(t)}_2)\bigr|=\bigl|v_j(Z^{(t+1)}_1)-v_j(Z^{(t+1)}_2)\bigr|$.
    Hence, the if condition in the while loop would be true, and algorithm terminates in the $t$'th iteration of the while loop, which contracts with $t+1<k$, so there should be at least a transfer.

	Note that $\yat \greatereqval{i} \ybt$, so if there is a first transfer, it is from $\yat$ to $\ybt$.
	Hence, we need to show that there is a transfer from $\ybt$ to $\yat$.
	Suppose on the contrary that no good is transferred from $\ybt$ to $\yat$.
	Then, we get $\{Z^{(t+1)}_1,Z^{(t+1)}_2\} = \{\yat \setminus S, \ybt \cup S\}$ for some $S\subset \yat$ such that $S\ne \emptyset$. 

By \cref{ob: inv unit}, $(Y^{(t)}_1,Y^{(t)}_2)$ is an \efx\ cut for agent $j$, so using the fact $S\ne \emptyset$, 
	we get $\yat \setminus S \lowereqval{j} \ybt \lowereqval{j} \ybt \cup S$. 
	Hence, we have 
\begin{align*}
	&\ybt \cup S \greatereqval{j} \ybt,\\
	&\ybt \greatereqval{j} \yat \setminus S \to
	\ybt \cup S \greatereqval{j} \yat,
\end{align*}
	where in the last line we used the fact that valuation function is cancelable. Hence, we get:
	$$ \max(v_j(Z^{(t+1)}_1),v_j(Z^{(t+1)}_2)) \geq \max(v_j(\yat),v_j(\ybt)).$$

    Moreover, we have 
    \begin{align*}
    \min(v_j(Z^{(t+1)}_1),v_j(Z^{(t+1)}_2)) = v_j(\yat \setminus S) \leq v_j(\yat), \\
    \min(v_j(Z^{(t+1)}_1),v_j(Z^{(t+1)}_2)) = v_j(\yat \setminus S) \leq v_j(\ybt) , \\
    \min(v_j(Z^{(t+1)}_1),v_j(Z^{(t+1)}_2)) \leq \min(v_j(\yat),v_j(\ybt)).  
    \end{align*}

	Therefore, we get 
	$$\bigl|v_j(Y^{(t)}_1)-v_j(Y^{(t)}_2)\bigr| \leq \bigl|v_j(Z^{(t+1)}_1)-v_j(Z^{(t+1)}_2)\bigr|,$$
	which means that algorithm should have terminated in step $t$, which contradicts with $t+1<k$. 
	Hence, the proof is complete.  \qedhere
\end{proof}

    We will need another auxiliary lemma for cancelable valuation functions.

    \begin{lemma}\label{lem:cancelable-swap}
	Let $v$ be cancelable, and let $(S_1,S_2)$ and $(T_1,T_2)$ be two partitions of the same set of goods.
	Suppose
	\[
		T_1 \succeq S_1 \succeq S_2 \succeq T_2
	\]
	and
	\[
		|v(S_1)-v(S_2)| < |v(T_1)-v(T_2)|.
	\]
	If $x\in S_1$, $y\in T_1$, and
	\[
		S_1\setminus x \succ T_1\setminus y,
	\]
	then $x\prec y$.
\end{lemma}
\begin{proof}
	Let
	\[
		C=T_1\setminus S_1,\qquad D=S_1\setminus T_1,\qquad I=S_1\cap T_1,\qquad K=S_2\cap T_2.
	\]
	Then
	\[
		S_1=D\sqcup I,\qquad T_1=C\sqcup I,\qquad S_2=C\sqcup K,\qquad T_2=D\sqcup K.
	\]
	Suppose, for contradiction, that $y\preceq x$. We show that $C\preceq D$. We divide to four cases based on in which subset $x$ and $y$ are.

	\textbf{Case 1:} $\mathbf{x\in D}$ and $\mathbf{y\in C}$.
    Hence, by $S_1\setminus x \succ T_1\setminus y$, 
    we get $I \sqcup (D\setminus x) \succ I \sqcup (C\setminus y)$.
    By canceling the common bundle $I$ from both sides, we get
	\[
		D\setminus x \succ C\setminus y.
	\]
	Since $y\preceq x$, by \cref{cancp},
	\[
		C=(C\setminus y)\sqcup y \preceq (D\setminus x)\sqcup x=D.
	\]

	\textbf{Case 2:} $\mathbf{x\in I}$ and $\mathbf{y\in C}$. 
    Hence, by $S_1\setminus x \succ T_1\setminus y$, 
    we get $(I \setminus x) \sqcup D \succ I \sqcup (C\setminus y)$.
    By canceling the common bundle $I\setminus x$, we get
	\[
		D \succ (C\setminus y)\sqcup x.
	\]
	Since $y\preceq x$, by \cref{cancp},
	\[
		C=(C\setminus y)\sqcup y \preceq (C\setminus y)\sqcup x \prec  D.
	\]

	\textbf{Case 3:} $\mathbf{x\in D}$ and $\mathbf{y\in I}$. 
    Hence, by $S_1\setminus x \succ T_1\setminus y$, 
    we get $I  \sqcup (D \setminus x) \succ (I\setminus y) \sqcup C$.
    By canceling the common bundle $I\setminus y$, we get
	\[
		(D\setminus x)\sqcup y \succ C.
	\]
	Since $y\preceq x$, by \cref{cancp},
	\[
		C \prec  (D\setminus x)\sqcup y \preceq  (D\setminus x)\sqcup x=D.
	\]

	\textbf{Case 4:} $\mathbf{x,y\in I}$. 
    Hence, by $S_1\setminus x \succ T_1\setminus y$, 
    we get $(I\setminus x)  \sqcup D \succ (I\setminus y) \sqcup C$.
    If $x=y$, then by canceling the common bundle $I\setminus x$, we get $D\succ C$.
	If $x\neq y$, then by canceling the common bundle $I\setminus\{x,y\}$, we get
	\[
		D\sqcup y \succ C\sqcup x.
	\]
	Since $y\preceq x$, by \cref{cancp}, we have $D\sqcup y\preceq D\sqcup x$.
	Thus,
	\[
		C\sqcup x \prec D\sqcup x.
	\]
	By cancelability, $C\prec D$.

	Therefore, in all cases, $C\preceq D$. By \cref{cancp}, this implies
	\[
		T_1=C\sqcup I \preceq  D\sqcup I=S_1
	\]
	and
	\[
		S_2=C\sqcup K \preceq  D\sqcup K=T_2.
	\]
	Combining this with the assumptions $S_1\preceq T_1$ and $T_2\preceq S_2$, we get
	\[
		v(S_1)=v(T_1)
		\qquad\text{and}\qquad
		v(S_2)=v(T_2),
	\]
	contradicting
	\[
		|v(S_1)-v(S_2)| < |v(T_1)-v(T_2)|.
	\]
	Hence, $y\not\preceq x$, and therefore $x\prec y$.
\end{proof}

	When agent $i$ strongly envies bundle $T$ w.r.t bundle $S$, we write $S \efxfenvyll{i} T$, and
	when agent $i$ does not strongly envy bundle $T$ w.r.t bundle $S$, we write $S \doesnotefxenvyll{i} T$.
\begin{lemma}\label{lem:decrease least} 
	If $k>t+1$, then
	$\ell_1^i(\zatt) \lowerval{i} \ell_1^i(\zat)$.
\end{lemma}
\begin{proof}
	By \cref{lem:transfer more}, we have that 
			$\texttt{PR}(\yat,Y^{(t)}_2,v_i)$ changes the bundles $(\yat,Y^{(t)}_2)$, so
	there is a strong envy at first, so $\ybt \efxfenvyll{i} \yat$. 
    For the rest of the proof, we write $\ell_z(S)$ instead of $\ell_z^i(S)$.
  
	By \cref{lem:transfer more}, during execution of $\texttt{PR}(\yat,Y^{(t)}_2,v_i)$, 
	at some point, there exist a transfer from $\ybt$ to  $\yat$.
	So by the modified PR algorithm, there is an integer $r$ such that 
	$$ \ybt \cup \{\ell_1(\yat),\ldots, \ell_r(\yat)\}  \greaterval{i} 
		\yat \setminus \{\ell_1(\yat),\ldots, \ell_r(\yat)\} \greaterval{i} \ybt \cup \{\ell_1(\yat),\ldots, \ell_{r-1}(\yat)\}. $$
	According to \cref{claim: no transfer}, after the transfer of $\ell_r(\yat)$,
	if $y$ is transferred, we have that: 
\begin{align}
	y \lowerval{i} \ell_r(\yat).   \label{< ryat}
\end{align}
	Moreover, by the condition of the while loop, 
    we have $|v_i(\yat)-v_i(\ybt)| < |v_i(\zat)-v_i(\zbt)|$, so since the valuation function is cancelable, by the second statement of 
    \cref{cancelablevalprops}, we get 
    $v_i(\zbt)= \min(v_i(\zat),v_i(\zbt)) \leq \min(v_i(\yat),v_i(\ybt))=v_i(\ybt)$. 
    Hence,
		$$\yat \setminus \{\ell_r(\yat)\}  \greatereqval{i} \yat \setminus \{\ell_1(\yat),\ldots, \ell_r(\yat)\}  \greaterval{i} \ybt  \greatereqval{i} 
		\zbt \greatereqval{i} \zat \setminus \ell_1(\zat),$$
    where the last inequality comes from the fact that $\zbt \doesnotefxenvyll{i} \zat$.
	
	Combining the facts $|v_i(\yat)-v_i(\ybt)| < |v_i(\zat)-v_i(\zbt)|$, $\ybt \greatereqval{i} \zbt$, $\yat \lowereqval{i} \zat$, $\yat \setminus \ell_r(\yat) \greaterval{i} \zat \setminus \ell_1(\zat)$, the fact that valuation functions are cancelable, and \cref{lem:cancelable-swap},
	we get that:  
\begin{align}
    \ell_r(\yat) \lowerval{i} \ell_1(\zat).  \label{ryat<azat}
\end{align}
	We have that $(\zatt,\zbtt) = \texttt{PR}(\yat,\ybt,v_i)$.  Let $(S^{h}, T^{h})$ denote the bundles before step $h$ of 
	execution $\texttt{PR}(\yat,\ybt,v_i)$, and let $\zatt = S^{h'}$ for some $h'$. 
	
	\paragraph{For some step $p$, we have
	$S^{p} \lowerval{i} T^{p}$ and $S^{p+1} \greatereqval{i} T^{p+1}$:}
	Suppose not, then since $\zatt \greatereqval{i} \zbtt$, 
	we should have $S^{h} \greatereqval{i} T^{h}$ for any step $h$, which means transfer is only from $S$ to $T$,
	which contradicts with \cref{lem:transfer more}. Hence, such a $p$ exists. 

	Thus, according to PR algorithm, $S^{p} \lowerval{i} T^{p}\setminus \ell_1(T^p) \lowereqval{i} S^{p} \cup \ell_1(T^p)$.
	By \cref{claim: no transfer}, after this point if $y$ is transferred, we have that $y \lowerval{i} \ell_1(T^p)$.
	Hence, we get $\ell_1(T^p) \in \zatt$. Combining this with \cref{< ryat} and \cref{ryat<azat}, we get:
	$$\ell_1(\zatt) \lowereqval{i}\ell_1(T^p)  \lowereqval{i} \ell_r(\yat)  \lowerval{i} \ell_1(\zat).$$
	Hence, the proof is complete.	
\end{proof}

\begin{corollary}\label{lem:polyunitbundle}
	Algorithm~\ref{alg:computing-unit-bundles} terminates in polynomial time in $m$ (the number of goods) in when the valuation functions are cancelable.
\end{corollary}
\begin{proof}
	By \cref{lem: pr polynomilal}, we have that every iteration of the while loop runs in polynomial time.
	Moreover, by \cref{lem:decrease least},
	if the while loop executes $k$ times, then for every  $t<k-1$, we have
	$\ell_1(\zatt) \lowerval{i} \ell_1(\zat)$, which means $t<m$ since there are $m$ goods.
	Hence, $k<m+1$. Thus, algorithm terminates in polynomial time.
\end{proof}

\section{Missing Proof of \cref{sec:p3-maker}}\label{sec:tree_appendix}

\lemtreebreakersupport*
\begin{proof}
	Let $U=H, p=p$, and $j=h_0$ at first. 
	Then, during the execution of \textsc{Reduce Trees}$(\X,h_0)$, 
	for every iteration of the while loop, let $i_1 \to \ldots  \to i_k$ be the critical path of that iteration, and
	set $U \gets U \cup \{i_1,i_2,\ldots, i_{k-1}\}$, $p=i_{k-1}$, and $j = i_{k}$.
	First, we show that the following invariants hold after every iteration:
\begin{enumerate}
	\item $j \notin U$,
    \item $\forall r \in N\setminus U$, for $i= \arg\max_{\ell \in U \setminus p} v_r(a_{r,\ell})$,
        unit bundle $a_{i,r}$ is not allocated to agent $i$, i.e., $X_i \cap a_{i,r}=\emptyset$. 
	\item Agent $p$'s bundle is a unit bundle in $E_{p,j}$, and agent $p$ does not resent anyone.
	\item Every agent $t\in N\setminus U$ possesses exactly one unit bundle.
	\item Allocation remains \pzero.
\end{enumerate}

	By the assumptions of the lemma, all invariants hold before the first iteration of the while loop.
	Assume inductively that the invariants hold at the beginning of some iteration, 
	and we show that they also hold at the beginning of the next iteration.

	Let $U^0$, $p^0$, and $j^0$ denote the values of $U$, $p$, and $j$ at the start of the iteration, and let $U$, $p$, and $j$ denote their 
	updated values after executing of on iteration of the while loop.
	Let
\[
i_1 \to i_2 \to \cdots \to i_k
\]
	be the critical path selected in this iteration. Note that by definition of critical path $k\geq 2$.
    Moreover, since $j^0$ was the last agent in the previous critical path, we get $i_1=j^0$.
    First, we show that for every $\ell\in [k]$, we have $i_\ell \notin U^0$.
	By invariant~1 at the start of the iteration, we have $j^0 = i_1 \notin U^0$.
	We now show by induction on $\ell$ that $i_\ell \notin U^0$ for all $\ell \ge 1$.
	The base case $\ell = 1$ holds.
	Assume $i_\ell \notin U^0$.
	By invariant~2, 
    for $s= \arg\max_{t \in U \setminus p^0} v_{i_\ell}(a_{i_\ell,t})$,
    unit bundle $a_{s,i_\ell}$ is not allocated to agent $s$, 
    so $a_{i_\ell,s}$ is either allocated to $s$ or is unallocated.
    In either case, since allocation is unitary, 
    we get $X_{i_\ell} \greatereqval{i_\ell} a_{i_\ell,s}$,
    so for every $q\in U^0\setminus p^0$, 
    we get $X_{i_\ell}  \greatereqval{i_\ell} a_{i_\ell,s} \greatereqval{i_\ell} a_{i_\ell,q}$.
    As a result, agent $i_\ell$ does not resent any agent in $U^0 \setminus p^0$, so 
    $i_{\ell+1}\notin U^0\setminus p^0$.

	In addition, by invariant~3, since $i_1=j^0$, we get agent $p^0$'s bundle is a unit bundle in $E_{p^0,i_1}$, so only agent $i_1$ may resent agent $p^0$.
    By invariant~3, agent $p^0$ does not resent anyone, and
	since the critical path has length at least two, we get that $p^0$ is not in the critical path.
	It follows that $i_{\ell+1} \notin U^0$.

	\paragraph{Invariant 1.}  
    It follows from the fact that $j=i_k$, $i_k \notin U^0$, $i_k \notin \{i_1,\ldots,i_{k-1}\}$,
    and $U=U^0 \cup \{i_1,\ldots, i_{k-1}\}$.

	\paragraph{Invariant 2.}
    First, note that the bundles of agents in $U^0$ do not change, since $i_\ell\notin U^0$ for every 
    $i_\ell$ in the critical path.
    We have $U=U^0\cup \{i_1,\ldots,i_{k-1}\}.$
    Hence, for every agent $r \in N\setminus U$, we have $r \in N\setminus U^0$.
    Thus, for $i= \arg\max_{\ell \in U^0 \setminus p^0} v_r(a_{r,\ell})$, we have
    $X_i \cap a_{i,r}=\emptyset$.
    Agent $p^0$ holds a unit bundle from agent $j^0 = i_1$, and since $i_1$ is added to $U$,
    we have that $X_{p^0} \cap a_{p^0,r}=\emptyset$.
    Moreover, every agent $i_\ell \neq p = i_{k-1}$ that is added to $U$ receives a unit bundle from $i_{\ell+1}$, who is also added to $U$.
    Therefore, for $i'= \arg\max_{\ell \in U \setminus p} v_r(a_{r,\ell})$, we have
    $X_{i'} \cap a_{i',r}=\emptyset$.
	Hence, invariant~2 holds at the start of the next iteration.


	\paragraph{Invariant 3.}
	Agent $p = i_{k-1}$ receives $a_{i_{k-1},i_k}$, and $j = i_k$.
	Moreover, since the chosen path is the critical path, we have $k \ge 2$, and agent $p = i_{k-1}$ resented $i_k$ the most before the reassignment.
	Therefore, after receiving this unit bundle, agent $p$ does not resent any agent, and invariant~3 holds.

\paragraph{Invariant 4.}
	Let $t\in N\setminus U$, then $t\in N\setminus U^0$.
	If $t\ne j$, then since $i_1,\ldots, i_{k-1}\in U$, agent $t$'s bundle does not change, so by invariant $4$ in the last iteration, agent $t$
	still possesses exactly one unit bundle.
	Moreover, since agent $j$ chooses an unallocated unit bundle, she possesses exactly one unit bundle.

\paragraph{Invariant 5.}
    Note that by invariant 5 at the start of the iteration, allocation was \pzero,
    and in one iteration, we execute \textsc{BreakTree$(\X, T)$}.
    Hence, by \cref{lem:breaktree}, allocation remain an unitary orientation, and the resent tree remains a forest.

    If for two arbitrary agent $r$ and $t$, $r$ does not resent agent $t$, then $r$ does not envy $t$.
    So suppose $r \to t$. By \cref{lem:breaktree}, $t\notin \{i_1,\ldots,i_{k-2},i_k\}$.
    
    First, assume $r\ne j=i_k$. By \cref{lem:breaktree}, we get that $r$ resented agent $t$ before this iteration as well, so since allocation was \pzero, we get $X_t= a_{t,s}$.
    Moreover, because $X_{i_{k-1}}=a_{i_{k-1},i_k}$, we get $t\ne i_{k-1}$.
    Hence, agent $t$'s bundle is not changed during the iteration. Hence,
    agent $r$ does not strongly envy agent $t$.

	It remains to consider the case $r=j$. By invariant 1, $j \in N\setminus U$, so
	by invariant 2 and the fact that agent $j$ got the most valuable unallocated unit bundle, 
    $\forall q\in U\setminus p$, we have $X_j \greatereqval{j} a_{j,q}$.
	Therefore, if agent $j$ resents some agent $t$, then $t\in \{p\} \cup (N\setminus U)$, 
	so by invariants $3$ and $4$, we get $X_t = a_{t,j}$.
	By \cref{lem:minimum value}, $X_j\greatereqval{j} b_{t,j}$, so agent $j$ does not strongly envy any other agent.
	We conclude that no strong envy exists.
	Thus, the allocation remains \pzero.

\paragraph{Termination.}
	By invariant $3$, $j^0\notin U^0$, but since $j^0=i_1$, we get $j^0\in U$. Hence, $|U|$ strictly increases in every iteration, 
	and since there are $n$ agents in total, the while loop terminates after at most $n$ iterations.

	Moreover, since at every iteration $i_\ell \notin U^0$, we have $i_\ell \notin H$, and therefore the bundles of agents in $H$ do not change 
	throughout the execution.
	Moreover, every agent in $N \setminus H$ initially holds exactly one unit bundle, and whenever an agent's bundle changes, she receives a unit bundle.
	Hence, at termination, every agent in $N \setminus H$ holds exactly one unit bundle.

    Finally, suppose $h_0$ was resenting $i$ in the first critical path, so they are added to $U$ at
    the first step, so their bundle will not change after first step, so at the end
    $X_{h_0}=a_{h_0,i}$, and since $i$ did not resent $h_0$ after first step, and her bundle does not change, $i$ does not resent $h_0$ at the end. Hence, $h_0$ remains non-envied.  
\end{proof}

\localtreebreakerends*
\begin{proof}
Let $i_1 \to i_2 \to \cdots \to i_k$
be the critical path selected in the first iteration of \textsc{Reduce Trees}$(\X,r(T))$.
 Execute the first iteration of \textsc{Reduce Trees}$(\X,r(T))$.
After this iteration, agent $i_{k-1}$ receives a unit bundle from agent $i_k$ and, by the definition of the critical path, does not resent any agent.

We now define $H = \{i_1, \ldots, i_{k-1}\}$, $p = i_{k-1}$, $h_0 = i_k$, and $G = N \setminus (H \cup \{h_0\})$. At this point, all conditions of \cref{lem:treebreaker support} are satisfied. Therefore, by \cref{lem:treebreaker support}, the execution of \textsc{Reduce Trees}$(\X,T)$ terminates after at most $n$ iterations.
Moreover, throughout the execution the allocation remains \pzero.

Finally, observe that at every iteration, any agent whose bundle changes receives a unit bundle. Since the initial allocation is \ptwo, the allocation remains \ptwo\ throughout the execution.
\end{proof}

\lemtreei*
\begin{proof}
    In every iteration of the while loop in \cref{alg:treebreaker}, algorithm executes a
    \textsc{Reduce Trees}($\X$,$r(T)$), and
    by \cref{local tree breaker ends}, allocation remains \ptwo\ after every execution.
    Hence, if the algorithm terminates, the resulting allocation is \ptwo.
    Moreover, if the algorithm terminates, by the condition of the while loop, we get that there is no path with a length greater than one in the resent graph of the resulting allocation.
    Therefore, if and when algorithm terminates, the resulting allocation is \pthree.
    
	By \cref{local tree breaker ends}, every execution of the while loop in \cref{alg:treebreaker} terminates. 	
	We show that after the execution of \textsc{Reduce Trees}($\X$,$r(T)$) in an iteration of the while loop, 
	the potential function $\phi(\X)$ decreases by at least one. Since $\phi(\X) \le n^2$, the while loop terminates after at most $n^2$ executions.

	Consider a single iteration of \cref{alg:treebreaker}, and suppose that during this iteration the while loop in
	\textsc{Reduce Trees}($\X$,$r(T)$)  executes $t$ times. For each $\ell \in [t]$,
    let $\X^\ell$ denote the allocation right before the $\ell$-th execution of the
	while loop in \textsc{Reduce Trees}($\X$,$r(T)$).
    Thus, $\X^1$ is the initial allocation right before execution 
	of \textsc{Reduce Trees}($\X$,$r(T)$), and let $\X^{t+1}$ be the resulting allocation after the execution of 
	\textsc{Reduce Trees}($\X$,$r(T)$). Our goal is to show that
\[
\phi(\X^{t+1}) < \phi(\X^1).
\]

    For each $\ell \in [t]$,
    let $T_\ell$ denote
	the tree processed in the $\ell$-th execution of \textsc{BreakTree} (in $(\X^{\ell})$), and let $j_\ell$ be its root.
	According to the algorithm, for every $\ell \ge 2$, the root $j_\ell$ is the last agent
	on the critical path of $T_{\ell-1}$. Let $j_{t+1}$ be the last agent on the critical
	path of $T_t$, and denote by $T_{t+1}$ the tree rooted at $j_{t+1}$ at the end of the
	execution of \Cref{alg:treereduce}.

\paragraph{Claim.}
For every $\ell \in \{1,\ldots,t\}$, we have
\begin{align*}
\phi(\X^{\ell+1})
&\le \phi(\X^\ell) \\
&\quad + \bigl(\text{number of agents of depth at least $2$ in } T_{\ell+1}\bigr) \\
&\quad - \bigl(\text{number of agents of depth at least $2$ in } T_{\ell}\bigr).
\end{align*}

\paragraph{Proof of the Claim.}
	Fix $\ell \in \{1,\ldots,t\}$ and consider the execution of \textsc{BreakTree}$(T_\ell)$ in the while loop of 
	\textsc{Reduce Trees}($\X$,$r(T)$) that transforms $\X^\ell$ into $\X^{\ell+1}$.
    Suppose $j^\ell=i_1 \to \ldots \to i_{k_1} \to i_k=j^{\ell+1}$ is the critical path of $T_\ell$.
    By \cref{lem:breaktree}, any resent in $\X^{\ell+1}$ that does not exist in $\X^\ell$ is directed from agent $j^{\ell+1}$ towards some other agent. Moreover, agent $j^{\ell+1}$ is non-envied in $\X^{\ell+1}$. Therefore, for any agent $i$ the length of the longest path ending at agent $i$ in $G_r(\X^{\ell+1})$ is at most one greater than the length of the longest path ending at $i$ in $G_r(\X^{\ell})$, i.e., $D_i(\X^{\ell+1})\leq D_i(\X^\ell)+1$. Hence, 
    \begin{align}
    \forall i\in N:   \phi_i(\X^{\ell+1}) \leq \phi_i(\X^\ell)+1. \label{eq:foranyi}        
    \end{align}

    Moreover, if $i\notin T_{\ell+1}$, then agent $j^{\ell+1}$ is not included in the longest path
    ending at agent $i$ in $G_r(\X^{\ell+1})$, so this path existed in $G_r(\X^\ell)$ as well; thus,
    $D_i(\X^{\ell+1})\leq D_i(\X^\ell)$.
    As a result, we get:
    \begin{align}
        \forall i\notin T_{\ell+1}: \phi_i(\X^{\ell+1}) \leq \phi_i(\X^\ell). \label{eq:inotinTell}
    \end{align}

    Next, suppose $i\in T_\ell$ and $D_i(\X^\ell)\leq 1$. Then, either $i=j^\ell$ or agent $j^\ell$
    resents agent $i$ in $\X^\ell$. By \cref{lem:breaktree}, $j^\ell$ is non-resented in $\X^{\ell+1}$,
    so $\phi_{j^\ell}(\X^{\ell+1})=0 \leq \phi_{j^\ell}(\X^{\ell})$.
    So assume that agent $j^\ell$ resents agent $i$ in $\X^\ell$; therefore,
    since $\X^\ell$ is \ptwo, we get $X_i = a_{i,j^\ell}$.
    If $i$ is not in the critical path, then its bundle does not change, so the only agent who may envy $i$ in $\X^{\ell+1}$ is agent $j^\ell$. Hence, since $j^\ell$ is non-resented in $\X^{\ell+1}$, we get $D_i(\X^{\ell+1})\leq 1$, so 
    $\phi_i(\X^{\ell+1}) = 0 = \phi_i(\X^\ell)$.
    Also, if $i$ is in the critical path, by \cref{lem:breaktree}, in $\X^{\ell+1}$,
    either $i$ is non-resented or is resented by $j^{\ell+1}$, who is non-resented.
    Hence, 
    \begin{align}
    \forall i\in T_\ell \text{ with } D_i(\X^\ell)\leq 1:  \phi_i(\X^{\ell+1}) \leq \phi_i(\X^\ell).   \label{eq:notinTelfl}     
    \end{align}

    Next, suppose $i\in T_\ell$ and $D_i(\X^\ell)\geq 2$, which means $\phi_i(\X^\ell)\geq 1$.
    First, we show that there does not exist a path with length greater than one from $j^{\ell+1}$
    to $i$. Note that agent $j$ may only resent agents who hold some incident unit bundle to her.
    Hence, agent $j^{\ell+1}$ may only resent agent $i_{k-1}$ and agents who were non-resented in $\X^\ell$. Suppose there exists a path with length greater than one from $j^{\ell+1}$ to $i$, i,e.,
    $j^{\ell+1}\to r_1 \to \ldots, r_s \to  i$. Note that if $r_1\in T_\ell$, then $r_1=i_{k-1}$, who does not resent anyone in $\X^{\ell+1}$, which is contradiction; thus, $r_1 \notin T_\ell$.
    Moreover, since any resent in $G_r(\X^{\ell+1})$ that does not exist in $G_r(\X^\ell)$ is directed from agent $j^{\ell+1}$ towards some other agent, we get that $r_1 \to \ldots, r_s \to  i$ is a path in $G_r(\X^\ell)$ as well, so we get that $i$ and $r_1$ are in the same tree in $G_r(\X^\ell)$, which is a contradiction since $i\in T_\ell$ and $r_1 \notin T_\ell$.
    Therefore, there does not exist a path with length greater than one from $j^{\ell+1}$
    to $i$ in $G_r(\X^{\ell+1})$. 
    
    Moreover, since $i\in T_\ell$, the longest path to $i$ in $G_r(\X^\ell)$ starts from $j^\ell$.
    Since agent $j^\ell$ receives her most preferred unit bundle  among
    all agents whom she resents and that lie on a path of length greater than one,
    and since $D_i(\X)\geq 2$, it follows that agent $j^\ell$ is not in the longest path to $i$ in $G_r(\X^{\ell+1})$. Hence, we get that $D_i(\X^{\ell+1})\leq D_i(\X^\ell)-1$, so
    \begin{align}
        \forall i\in T_\ell \text{ with } D_i(\X^\ell) \geq 2: 
        \phi_i(\X^{\ell+1})\leq \phi_i(\X^\ell)-1.  \label{eq:inTell>1}
    \end{align}
    Moreover, for every agent $i$ with $D_i(\X^{\ell+1})\leq 1$, we have $\phi_i(\X^{\ell+1})=0$, so:
    \begin{align}
        \forall i \in T_{\ell+1} \setminus T_\ell \text{ with } D_i(\X^{\ell+1}) \leq 1: \phi_i(\X^{\ell+1}) \leq \phi_i(\X^\ell). \label{eq:tell+1<1}
    \end{align}

    As a result, by \cref{eq:inotinTell}, \cref{eq:notinTelfl}, \cref{eq:tell+1<1}, \cref{eq:inTell>1},
    and \cref{eq:foranyi}, we get:
    \begin{align*}
        \phi(\X^{\ell+1}) &= \sum_{i\in [n]} \phi_i(\X^{\ell+1}) \\
        &= \sum_{i\notin T_\ell \cup T_{\ell+1}} \phi_i(\X^{\ell+1})
        +  \sum_{i\in T_\ell  \text{ with } D_i(\X^{\ell})\leq 1} \phi_i(\X^{\ell+1})
        + \sum_{i\in T_{\ell+1} \setminus T_\ell \text{ with } D_i(\X^{\ell+1})\leq 1} \phi_i(\X^{\ell+1})\\
        &+ \sum_{i\in T_\ell \text{ with } D_i(\X^{\ell})\geq 2} \phi_i(\X^{\ell+1})
        + \sum_{i\in T_{\ell+1} \setminus T_\ell \text{ with } D_i(\X^{\ell+1})\geq 2} \phi_i(\X^{\ell+1})\\
        &\leq \sum_{i\in [n]} \phi_i(\X^{\ell}) -  
        \sum_{i\in T_\ell \text{ with } D_i(\X^{\ell})\geq 2} 1
        + \sum_{i\in T_{\ell+1} \setminus T_\ell \text{ with } D_i(\X^{\ell+1})\geq 2} 1 \\
        & \leq  \phi(\X^\ell) \\
        &\quad + \bigl(\text{number of agents of depth at least $2$ in } T_{\ell+1}\bigr) \\
        &\quad - \bigl(\text{number of agents of depth at least $2$ in } T_{\ell}\bigr).
    \end{align*}
    Thus, the proof of the claim is complete. \medskip

    Since the inner while loop terminates after $t$ executions, the final tree $T_{t+1}$
    contains no agent of depth at least $2$. Moreover, since $t \ge 1$, the initial tree
    $T_1$ contains at least one agent of depth at least $2$. Therefore, summing the
    inequalities in the claim over $\ell = 1,\ldots,t$ yields
\begin{align*}
    \phi(\X^{t+1})
    &\le \phi(\X^1)
    + \bigl(\text{agents of depth $\ge 2$ in } T_{t+1}\bigr)
    - \bigl(\text{agents of depth $\ge 2$ in } T_1\bigr) \\
    &\le \phi(\X^1) - 1 < \phi(\X^1).
\end{align*}
    Hence, the proof is complete.
\end{proof}

\lemtreeiii*
\begin{proof}
	Let $H'= H \cup \{t\in G: X_t \subseteq E_{t,h} \text{ for some } h\in H\}$.
	By the first condition of the lemma, and definition of $H'$,
    we get that for every $j \in N\setminus H'$, for $i= \arg\max_{\ell \in H' \setminus p} v_j(a_{j,\ell})$, unit bundle $a_{i,j}$ is not allocated to agent $i$.
    
    
	Hence, all conditions of \cref{lem:treebreaker support} are satisfied for $(H',G'=N\setminus (H'\cup h_0), h_0)$.
	It follows that throughout the execution of \textsc{Reduce Trees}$(\X,h_0)$, allocation remains \pzero. Moreover, $h_0$ remains non-resented.
	Thus, to show that the allocation is \pone, it suffices to show that, at termination, every resent path has length at most one.

	During the execution of \textsc{Reduce Trees}$(\X,h_0)$
	let $j_\ell$ denote the $j$ of the $\ell$'th iteration of the while loop of \textsc{Reduce Trees}$(\X,h_0)$, 
	and denote the resent tree rooted at $j_\ell$ after $\ell$'th iteration by $T_j$. Also, let $j_0=h_0$.
    Note that for every $\ell$, $j_{\ell+1}$ is the last agent on the critical path of $T_{j_\ell}$.
    Hence, by \cref{lem:breaktree}, right after the execution of $\ell$'th iteration of the while loop, 
	agent $j_\ell$ is non-resented. Also, $j_0=h_0$ is non-resented at first by lemma's assumption.
    
	We prove by induction on the iterations that the following invariant holds:
	right after the execution of $\ell$'th iteration of the while loop, 
	every resent path in the resent graph that does not include $j_\ell$ has length exactly one. 
	The base case holds since every resent path not including $h_0$ has length one.

	Now consider an iteration in which the tree $T_{j_\ell}$ is processed.	
	Note that $j_{\ell+1}$ is the last agent on the critical path of $T_{j_\ell}$.
	By the induction hypothesis, the height of $T_{j_\ell}$ is exactly two, so the critical path has
	the form $j_\ell \to i_1 \to j_{\ell+1}$ for some agent $i_1$.

	After this iteration, by \cref{lem:breaktree}, agents $j_{\ell}$ and $j_{\ell+1}$ are non-resented.
	Moreover, if any new resent relation $r\to q$ is created, it must satisfy $r=j_{\ell+1}$.

	Consider any path of the form $j_\ell \to i'_1 \to i'_2$ at the start of the iteration.
	By the definition of the critical path, we have $a_{j_\ell,i_1} \greatereqval{j_\ell} a_{j_\ell,i'_1}$.
	Therefore, after receiving $a_{j_\ell, i_1}$, agent $j_\ell$ cannot be part of any resent path of
	length greater than one.
	Thus, after the iteration, any resent path of length greater than one must start at $j_\ell$,
	and the induction invariant is preserved. Consequently, when \textsc{Reduce Trees}$(\X,h_0)$ terminates, every resent tree has height at most one,
	and hence the resulting allocation is \pone.

    By \cref{lem:treebreaker support}, bundles of agents in $H'$ do not change.
	By the first condition of the lemma,
    we get that for every $j \in N \setminus H$, for $i= \arg\max_{\ell \in H \setminus p} v_j(a_{j,\ell})$, unit bundle $a_{i,j}$ is not allocated to agent $i$,
    so since the final allocation is \pone, we get that for every $h\in H\setminus p$, 
    we have $X_j \greatereqval{j} a_{j,i}  \greatereqval{j} a_{j,h}$.
    So no agent $j \in N \setminus H$, resents an agent in $H\setminus p$.

	Suppose agent $h\in H\setminus p$ was non-resented before the execution of \textsc{Reduce Trees}$(\X,h_0)$.
	By the above argument, 
	no agent $r\notin H$ resents $h$ at the end.
	Moreover, since $h$ was non-resented, any $h' \in H$ did not resent it.
	Hence, since the bundles in $H$ do not change, agent $h'$ still does not resent $h$, so $h$ remains non-resented.

	Finally, before executing \textsc{Reduce Trees}$(\X,h_0)$, for any agent $h\in H$, and any agent $t\notin \{h,h_0\}$ possessing a unit bundle in $E_{h,t}$, we had either $t\in H$, or 
    $X_t \subseteq E_{t,h}$, so in any case $t\in H'$.
	Hence, using the fact that the bundles of the agents in $H'$ do not change, we get after the execution of \textsc{Reduce Trees}$(\X,h_0)$,
	unit bundle $a_{t,h}$ is still allocated to $t$. 
	Consequently, For every $h\in H$, and for every $t\notin \{h,h_0\}$, no allocated unit bundle in $E_{u,t}$ gets unallocated.
	 This completes the proof.
\end{proof}

\section{Missing Proof of \cref{sec:dump}}\label{sec:dump_appendix}

\lemdumprule*

\begin{proof}
    Let $\X$ and $\xp$ denote the allocation before and after the dumping phase. We prove each statement separately:

	(1) If $p$ was resented before, we have $X'_p = X_p$ by property \ref{pd1}. If $p$ was not resented before, properties \ref{pd3}, \ref{pd4} guarantee that 
    $X'_p \cap E_{p,q} \greatereqval{p} X_p \cap E_{p,q}$ for every $q$.
    This is because $a_{p,q}\greatereqval{p} b_{q,p}$ and \cref{cancp}. 
    Therefore, $X'_p \greatereqval{p} X_p$. 

	(2) Let $p$ be any agent, and let $q$ be any  agent that is resented in $\X$. 
    Since by \cref{pbd}, $\X$ is \pone, by \cref{obs:basicEFX}, $p$ does not strongly envy $q$ in $\X$. Also, $X'_p \greatereqval{p} X_p$ and $X'_q=X_q=a_{q,p}$. Thus $p$ does not strongly envy $q$ in $X'$. 
	
    (3) By \cref{pbd}, we have $X_q \greatereqval{q} b_{q,p} \cup A_q(\X)$. We show $b_{q,p} \cup A_q(\X) \greatereqval{q} X'_p \cap E_q$. 
    
    By rule \ref{pd5}, we have that \dumpp{b_{q,p}}{p} and hence $X'_p \cap E_{p,q} \lowereqval{q} b_{q,p}$. 
    
    Let $j \notin \{q,p\}$ be an arbitrary agent. By property \ref{pd2}, at most one unit bundle in $E_{q,j}$ is allocated to $p$ in $X'$. By \cref{pbd}, the allocation was an orientation before the dumping phase, and hence $X_p\cap E_{q,j} =\emptyset$.

    If $a_{q,j}$ was unallocated in $\X$, 
    then since $a_{q,j}$ is the most valuable unit bundle in $E_{q,j}$ with respect to agent $q$, we get $A_q(\X) \cap E_{q,j} = a_{q,j} \greatereqval{q} X'_p \cap E_{q,j}$.
    
    So, suppose $a_{q,j}$ is not unallocated in $\X$. 
    Since $\X$ is \pone, $X_q=a_{q,p}$, so $a_{q,j}$ is not allocated to $q$.
    In addition, since $\X$ is unitary, $X_j\cap E_{q,j} \in \{a_{j,q},b_{q,j}\}$. 
    Because the allocation is an orientation, if $X_j\cap E_{q,j} = b_{q,j}$, then $a_{q,j}$ would be unallocated. Hence, $X_j\cap E_{q,j} = a_{j,q}$, so 
    $A_q(\X) \cap E_{q,j} = b_{j,q}$.
    Since $\X$ is \pone, we get that $j$ is non-resented, 
    so by rule \ref{pd3}, we get that $a_{j,q}\subseteq X'_j$. Thus, if some unit bundle in $E_{j,q}$
    is dumped to $X'_p$, it is $b_{j,q}$.
    Therefore, $A_q(\X) \cap E_{q,j} = b_{j,q} \greatereqval{q} X'_p \cap E_{q,j}$.

    As a result, in every case, we have 
    $A_q(\X) \cap E_{q,j} \greatereqval{q} X'_p \cap E_{q,j}$. 
    So, by \cref{cancp}, we get
	$$X'_p \cap E_q = (X'_p \cap E_{p,q}) \sumsq \bigsumsq_{j\notin\{p,q\}} (X'_p \cap E_{q,j}) 
	 \lowereqval{q} b_{q,p} \sumsq \bigsumsq_{j\notin\{p,q\}} (A_q(\X) \cap E_{q,j}) \lowereqval{q} b_{q,p} \sumsq A_q(\X) \lowereqval{q} X_q.$$

	(4) First, we show that $s$ does not strongly envy anyone. By part (2), $s$ does not strongly envy any agent who was resented before. Let $p$ be a non-resented agent in $\X$.  
	Then by \cref{lem:minimum value}, we get $X_s\cap E_{s,t}= X_s \greatereqval{s} a_{s,p}$ as $s$ does not resent $p$ in $\X$.
	Moreover, if for some $q\ne p$, a unit bundle in $E_{s,q}$ is allocated to $p$, then by rule \ref{pd7}, $q$ is resented. Note that such an agent $q$ cannot be resented by $s$, since otherwise, we would have $X'_q = a_{q, s}$ and $b_{q, s} \subseteq X'_s$, meaning that $E_{s, q}$ was completely allocated to agents $s$ and $q$, and thus, $p$ cannot receive a unit bundle from $E_{s,q}$. Therefore, for any such $q$ that a unit bundle from $E_{s, q}$ has been allocated to $p$, by rule \ref{pd6}, we get that $a_{s,q}$ is allocated to $s$, and $b_{s,q}$ is allocated to $p$. 
    Let $R'(\X)$ be the set of resented agents who are not resented by $s$ in $\X$.
    Since $t$ is non-resented, we have $t \notin R'(\X)$, so by \cref{cancp}, we get:
    $$X'_s \greatereqval{s} (X_s\cap E_{s,t}) \sumsq \bigsumsq_{q \in R'(\X)} a_{s,q} 
	\greatereqval{s} a_{s,p} \sumsq \bigsumsq_{q \in R'(\X)} b_{s,q} \greatereqval{s} X'_p,$$
    i.e., $s$ does not envy $p$. 

	Second, we show that nobody envies $s$. 
    Let $q$ be any agent resented in $\X$. If $s$ resents $q$ in $\X$, $q$ does not envy $s$ in $\xp$ by part (3). Thus, assume $q$ is resented by $i\neq s$.
    We have $X_q \greatereqval{q} A_q(\X)$, so 
    \begin{align*}
    X'_s \cap E_q &= \bigsumsq_{j\in R(\X)\cup s} (X'_s \cap E_{q,j}) \sumsq 
    \bigsumsq_{j\notin R(\X)\cup \{s,q\}} (X'_s \cap E_{q,j})\\ 
	 &\lowereqval{q} \bigsumsq_{j\in R(\X)\cup s} a_{q,j} \sumsq 
    \bigsumsq_{j\notin R(\X)\cup \{s,q\}} b_{j,q}  \\
    &\lowereqval{q}  \bigsumsq_{j\in R(\X)\cup s} (A_q(\X) \cap E_{q,j}) \sumsq 
    \bigsumsq_{j\notin R(\X)\cup \{s,q\}} (A_q(\X) \cap E_{q,j}) \\
    &= A_q(\X) \lowereqval{q} X_q = X'_q,
    \end{align*}
	where the first inequality is due to the fact during the dumping phase, by rule \ref{pd2} agent $s$ may get at most one unallocated unit bundle from every $E_{q,j}$, and
    for every non-resented agent $j$ who does not resent $q$, agent $j$ 
    receives $a_{j,q}$ by rule \ref{pd6}.
    Moreover, second inequality comes from the fact for $j\in R(\X)\cup s$, 
    we have that $E_{q,j}$ is unallocated in $\X$. 
    Hence, agent $s$ is not envied by any resented agent $q$ in $\xp$.

	For a non-resented agent $p\notin \{s,t\}$, by rule \ref{pd8}, agent $p$ gets $a_{p,s}$.
	Also, if for some $q\notin \{p,s\}$, a unit bundle in $E_{p,q}$ is allocated to $s$, then by rule \ref{pd7}, $q$ is resented. With an identical argument to our previous case, $q$ cannot be resented by $p$, and thus by rule \ref{pd6},
	$a_{p,q}$ is allocated to $p$, and $b_{p,q}$ is allocated to $s$.	Hence,
    $$X'_p \greatereqval{p} a_{p,s} \sumsq \bigsumsq_{q \in R(\X)} a_{p,q} 
	\greatereqval{p} b_{p,s} \sumsq \bigsumsq_{q \in R(\X)} b_{p,q} \greatereqval{p} X'_s,$$
    which means that $p$ does not envy $s$. In addition, agent $t$ does not envy agent $s$ since agent $s$ is not receiving any new edge from $E_t$ by rule \ref{pd8} and agent $t$ did not envy agent $s$ before the dumping.
\end{proof}

\lemcyclesupport*

\begin{proof}
    We first observe that $i$ and $j$ receive more valuable bundles by the application of (U1). Indeed, $i$ gets a more valuable bundle since she resented $j$, i.e., $a_{i,j}\greaterval{i} X_i$.
Also, $j$ receives a more valuable bundle since 
	 $$X_j \lowerval{j} A_j(\X) \sumsq b_{ji} \lowereqval{j} A_j(\X) \sumsq b_{i,j}.$$

Furthermore, $j$ does not resent $i$ since $a_{i,j} \lowereqval{j} a_{j,i} \lowerval{j} A_j(\X) \sumsq b_{j,i}  \lowereqval{j} A_j(\X) \sumsq b_{i,j}$.
	Moreover, since $a_{i,j} \subseteq E_{i,j}$, no other agent resents $i$. Additionally, agent $j$ becomes non-resented by the following observation: $A_j(\X)$ was unallocated, and the allocation was unitary. Therefore, $(i,j)$ becomes a support pair.

	Also, since the allocation is height-one before the update, and $j$ was resented, she did not resent anyone. Since she is getting a more valuable bundle, she does not resent anyone. Moreover, agent $i$ is getting better-off and will not resent any agent that she did not previously resent. Therefore, no new resent relation is created, and hence the resent graph remains a forest and height-one. Since only the bundles of agents $i$ and $j$ change in the update and they both become non-resented, it is clear that the allocation remains $\efx$.
    
	Moreover, the allocation remains unitary, since if a unit bundle gets unallocated, it was previously allocated to $i$. Since $i$ was not resented (recall that $i$ resented $j$, and the allocation was height-one), no one envies any newly unallocated unit bundle.
	Thus,  the allocation remains height-one.
	Moreover, the process of executing (U1) as long as possible terminates because each update decreases the number of  resent edges. 
\end{proof}

\lemtwosupport*
\begin{proof} 
	Let $(s_1,t_1)$ and $(s_2,t_2)$ be two disjoint support pairs. 
	First, as long as there exists $i \to j$ such that $b_{j,i} \sumsq A_j(\X) \greaterval{j} X_j = a_{j,i}$, we update the allocation with (U1).
	By \cref{lem:cycle_support}, the allocation remains \pone, $(i,j)$ is a new support pair, and no new resent is created.
	Note that since our previous support pairs were disjoint, we may assume $i\notin \{s_1,t_1\}$. 
	Moreover, since $i$ resents $j$, we have $j \notin \{s_1,t_1\}$. Hence, $(s_1,t_1)$ remains a support pair.
	Therefore, there still exist two disjoint support pairs after the application of (U1), namely $(s_1,t_1)$ and $(i,j)$. 
    The process terminates by \cref{lem:cycle_support}. Let $(s_1, t_1)$ and $(s_2, t_2)$ be two disjoint support pairs when (U1) is no longer applicable. We now enter a dumping phase. We use $\X$ and $\xp$ to denote the allocations before and after dumping. Since (U1) is no longer applicable, the dumping properties (\cref{pbd}) hold.

	\paragraph{Dumping Phase.}
	We allocate the remaining unallocated unit bundles by the following rules. Here, a \emph{root} means a non-resented agent, and for a root $p$, $R_p$ denotes the set of agents resented by $p$. The six cases cover all possible relationships between two agents.

	\begin{enumerate}[label=(\roman*), leftmargin=2.2em]
		
		\item \textbf{Between $s_i$ and other roots for $i\in \{1,2\}$.}
		Let \dumpp{a_{s_1,s_2}}{s_1} and \dumpp{b_{s_1,s_2}}{s_2}. 
		Since $(s_i,t_i)$ is a support pair, we have $X_{s_i} \subseteq E_{s_i, t_i}$.
		Assign \dumpp{E_{s_i, t_i} \setminus X_{s_i}}{t_i}.	For every other root $p\notin\{t_i,s_1,s_2\}$, 
		let \dumpp{b_{p,s_i}}{s_i} and \dumpu{a_{p,s_i}}{p}.
       
		\item \textbf{Root to its own children.}
		For every root $p$ and every $u\in R_p$, assign \dumpp{b_{p,u}}{p}. Note that we already have $X_u = a_{u,p}$.
		
		\item \textbf{Root to others' children.}
		Let $p$ and $r$ be two distinct roots, let $u\in R_r$, and let $i \in \{1, 2\}$ be such that $p\notin\{s_i,t_i\}$. Note that such an agent $p$ exists since there are two disjoint support pairs.
		Assign \dumpu{a_{p,u}}{p} and \dumpp{b_{p,u}}{s_i}.
		
		\item \textbf{Root to root.}
		Let $p$ and $r$ be two distinct roots such that $p,r \notin \{s_1,s_2\}$. 
		If at least one of the unit bundles is already allocated to one of the roots, allocate the other unit bundle to the other root. 		
		Otherwise, allocate one unit bundle in $E_{p,r}$ to $p$ and the other to $r$ arbitrarily.
		
		\item \textbf{Between two children of the same root.}
		Let $p$ be a root and let $u,v\in R_p$ be two distinct agents resented by $p$.
		Let $i \in \{1,2\}$ be such that $p\neq s_i$. Then, 
		assign \dumpp{a_{u,v}}{p} and \dumpp{b_{u,v}}{s_i}. 
		
		\item \textbf{Between two children of different roots.}
		Let $p$ and $r$ be two distinct roots, and let $u\in R_p$ and $v\in R_r$.
		Assign \dumpp{a_{u,v}}{s_1} and \dumpp{b_{u,v}}{s_2}.		
	\end{enumerate}	
	This completes the allocation. For a root $p$, there may be several assignments of the form \dumpu{a_{pz}}{p}. These assignments are for distinct $z$, and hence their order of execution is irrelevant.

	\paragraph{Correctness of Dumping Phase.}
    We prove that the final allocation $\xp$ is \efx. Note first that the rules above follow our general dumping rules properties, so by \cref{lem:dumprule}:
\begin{enumerate}
	\item For every agent $p$, we have $X'_p \greatereqval{i} X_p$.
	\item Nobody strongly envies any agent who was resented at the start of the dumping phase.
	\item For every resent $p\to q$ before the dumping phase, agent $q$ does not envy agent $p$ after the dumping phase.
	\item No one envies agents $s_1$ and $s_2$, and agents $s_1$ and $s_2$ do not strongly envy anyone.  
\end{enumerate}

	For any other non-resented agent $p\notin \{s_1,s_2\}$, and for any other agent $q$ not resented by $p$, 
	we have that $X'_p\cap E_q$ is a unit bundle, which was either possessed by agent $p$ or was unallocated before the dumping phase.
	In the first case, $q$ does not envy this unit bundle because she did not envy $p$ before the dumping, and in the second case,
	by the fact that allocation was \pone, and therefore, unitary, we have that $q$ did not envy any unallocated unit bundle. Hence, we conclude that the complete allocation is $\efx$.
\end{proof}

\section{\Stage\ C}\label{sec:c}

In this section, we show that given an allocation in \stage\ C, we can construct a complete \efx\ allocation.
Recall that an allocation is in \stage\ C if allocation is \pthree, and there exist four agents $i, j, k, \ell$ such that $k\to i$, $\ell \to j$, and
\[
D_j(\X) \sumsq a_{j,k} \greaterval{j} X_j.
\]
We first prove the structural lemma \cref{lem:weak_support}, and then use it to prove the main lemma of \stage\ C, \cref{lem: stageCprop}.

	
		
		
		

\begin{lemma}\label{lem:weak_support}
	
	Suppose $\X$ is a \pone\ allocation, and $(s,t)$ is a support pair such that:
	\begin{itemize}[leftmargin=2.2em]
		\item All agents except $t$ hold exactly one unit bundle.
		
		\item Agent $t$ does not resent anyone.
		
		\item Among all unit bundles, $a_{t,s}$ is agent $t$'s most valuable unit bundle. 
		
		\item There exist two agents $p^\star, q^\star \notin\{s,t\}$ 
                such that $p^\star$ \weakresents\ $q^\star$.
	\end{itemize}
	Then, we can compute a complete \efx\ allocation in polynomial time.
\end{lemma}

\begin{proof}
    Our algorithm has two main phases: the Update Phase and the Dumping Phase. 
    After the termination of the update phase, if we have not already found a complete $\efx$ allocation, we find one in the dumping phase.

	\noindent\textbf{Update Phase.} In this phase, we apply four update rules, denoted by (C0)–(C3). We proceed iteratively: as long as at least one rule is applicable, we update the allocation using any applicable rule.
    We show that after executing (C0), we get a complete \efx\ allocation.
    Moreover, we show that after executing any of (C1), (C2), and (C3), allocation still satisfies Lemma's conditions, and agent $s$ remains fixed.
    In addition, we show that after executing (C1) or (C3), $|R_s(\X)|$ (the number of agents whom are resented by $s$) strictly decreases.
    Moreover, we show that executing (C2), $|R_s(\X)|$ does not increase, and
    there may be at most $n$ iterations of update rule (C2) in a row, so the process of 
    executing update rules as long as any of them is applicable ends after at most $n^2$ iterations.
    We exit the update phase when no update rule is applicable.


	\noindent\textbf{(C0)} If there exists a resent edge $p\to q$ with $p\neq s$ such that
	$X_q \lowerval{q} A_q(\X) \sumsq b_{q,p}$:
	Then, we compute a complete $\efx$ allocation by constructing two disjoint support pairs.
	First, note that since agent $t$ does not resent anyone, we get that $p\ne t$.
	We apply update rule (U1) (\cref{up rule:u1}) to $(p,q)$. Then, by \cref{lem:cycle_support}, $(p,q)$ would be a support pair.
	Since $(s,t)$ and $(p,q)$ are disjoint, 
	we obtain two disjoint support pairs, and \cref{lem:two_support} yields a complete \efx\ allocation.
	Hence, for the remainder of the proof, we assume that for every resent edge $p\to q$ with $p\neq s$,
	\[
	X_q \greatereqval{q} A_q(\X) \sumsq b_{q,p}.
	\]
		
	\noindent\textbf{(C1)} If there exists a resent edge $s\to q$ such that $\X_q \lowerval{q} A_q(\X) \sumsq b_{q,s}$:
	Let $\X$ denote the allocation before the update and let $\X'$ denote the allocation after the following update:
	\begin{align*}
		X'_q &\takes A_q(\X) \cup b_{s,q}\\
		X'_s &\takes a_{s,q}\\
		X'_t &\takes a_{t,s}
	\end{align*}	
	Since $a_{t,s}$ is $t$'s most valuable unit bundle, $t$ does not resent anyone in $\xp$. 
    Moreover, agents $s$ and $q$ are getting a more valuable bundle, so they do not resent anyone in $\xp$ who they did not resent in $\X$. Other agents' bundles does not change, too,
    so no new resent is created.
	Note that the allocation remains a unitary orientation and since no new resent relation is created, the allocation remains \pone.
	By \cref{lem:cycle_support}, $(s,q)$ becomes a support pair.
	Also, agent $q$ does not resent anyone and is the only agent who has more than one unit bundle. Also, note that $a_{q, s}$ is agent $q$'s most valuable unit bundle since she was not resenting any agent in $\X$. 
    Moreover, $p^\star$ still \weakresents\ $q^\star$ since $p^\star,q^\star \notin \{t,s,q\}$.
    Hence, by setting agent $q$ to be the new $t$, we get our desired properties.
	We relabel agent $q$ as the new $t$ for the remainder of the update phase. 
    Since agent $s$ does not resent $q$ anymore, and no new resent is created,
    $s$ resents fewer agents, so $|R_s(\xp)|<|R_s(\X)|$.

	\noindent\textbf{(C2)} If  for some $p,q\notin \{s,t\}$, agent $p$ \weakresents\ some agent $q$, and
    \begin{align*}
	a_{t,p} \sumsq (B_t\setminus E_{t,p}) \greaterval{t} X_t,     
    \end{align*}
	update the allocation with update rule (U2), defined in \cref{up rule:u2}.
	Then, by \cref{lem:t update},
	allocation remains \pone, every agent except $t$ only has exactly one unit bundle,
	and no new resent from agents $s$ and $t$ to other agents is created, so agent $t$ still does not resent anyone, and $|R_s|$ is not increased.
	Also, agent $q$ \weakresents\ agent $p$, and we have $p,q\notin \{s,t\}$.
	Moreover, $a_{t,s}$ is agent $t$'s most valuable unit bundle, since the agents $s$ and $t$ are still the same agents.
	Hence, all the properties still hold. Also, $v_t(X_t)$ is strictly increases while $|R_s|$ does not increase.
    Since $v_t(a_{t,p} \sumsq B_t\setminus E_{p,t})$ is independent of the allocation and is a function of agents $t$ and $p$, and since agent $t$ does not change with update (C2), 
    update rule (C2) can only be repeated at most $n$ times in a row. 
	Hence, at the end (when (C2) is not applicable) for every $p,q \notin \{s,t\}$ such that $p$ \weakresents\ $q$, we have $X_t \greatereqval{t} a_{t,p} \sumsq (B_t\setminus E_{t,p})$.
	
	\noindent\textbf{(C3)} If there exist resent edges $s\to r$ and some agent $p\ne s$ who 
    \weakresents\ at least one agent such that
	\[
	a_{r,p} \sumsq (B_r \setminus b_{p,r}) \greatereqval{r} X_r.
	\]
	Then, let $q$ be an agent such that $v_p(a_{p,q})$ is maximized,
    and $p$ \weakresents\ $q$.
	Let $\X$ denote the allocation before the update and let 
    $\xp$ denote the allocation right before executing \textsc{Reduce Trees}.
	We update as follows(Use $\xp$ to denote the allocation right before executing \textsc{Reduce Trees}.):
	\begin{align*}
		&\text{Agents } t,s,r,p, \text{and } q \text{ drop their previous bundles}\\
		&X'_t \takes a_{t,s}\\
		&X'_s \takes a_{s,r}\\
		&X'_r \takes a_{r,p} \sumsq (B_r \setminus b_{p,r})\\
		&X'_p \takes a_{p,q}\\
		&X'_q \takes \textsc{choose}(q)\\
		&Run \textsc{ Reduce Trees}(\xp,q)
	\end{align*}
	
	Let $H = \{s,r,p\}$ and $h_0=q$. Note that $\xp$ is \pzero\ and satisfies the conditions of \cref{lem:treebreaker2}. Therefore, by \cref{lem:treebreaker2}, we get a \pone\ allocation,
	$(s,r)$ is a support pair, $r$ is the unique agent that may hold more than one unit bundle,  $r$ does not resent anyone, and $h_0=q$ is non-resented in $\xp$.
    Also, since agent $q$ was \weakresented\ by $p$,
    agent $q$'s most valuable unit bundle is $a_{q,p}$. Therefore, agent $q$ \weakresents\ agent $p$ after the update, so there are agents
    $p^\star, q^\star \notin\{s,t\}$ such that $p^\star$ \weakresents\ $q^\star$.
    Hence, by relabeling $r$ as the new $t$ the properties of lemma are maintained, and
    since $v(X_s)$ has increased, no new resent is created from agent $s$, and $s$ no longer 
    resents $r$, so $|R_s|$ has strictly decreased.

    Next, when we get an allocation $\X$ such that no update rule is applicable, we start the dumping phase.


	\paragraph{Dumping Phase.}
	We allocate every remaining unallocated unit bundle by the following rules.
	Here, a \emph{root} means a non-resented agent, and for a root $p$, the set $R_p$ denotes the agents resented by $p$ in $\X$. Every unallocated good gets allocated based on one of the following rules:

	\begin{enumerate}[label=(\roman*), leftmargin=2.2em]
		\item \textbf{Between $s$ and other roots.}
		Since $(s,t)$ is a support pair, we have $X_{s} \subseteq E_{s,t}$.
		Allocate the other unit bundle between $s$ and $t$ to $t$. For every other root $p\notin\{t,s\}$, 
		let \dumpu{a_{p,s}}{p} and \dumpp{b_{p,s}}{s}.
		
		\item \textbf{Root to root (except $s$).}
		Let $p$ and $r$ be two distinct roots such that $p,r \ne s$. 
		If at least one of the unit bundles is already allocated to one of the roots, allocate the other unit bundle to the other root. 		
		Otherwise, allocate one unit bundle in $E_{p,r}$ to $p$ and the other to $r$. 
		
		\item \textbf{Root to its own children.}
		For every root $p$ and every $u\in R_p$, assign \dumpp{b_{p,u}}{p}. Also, we have already have $X_u = a_{u,p}$.		
		
		\item \textbf{Root to other children.}
		Let $p$ and $r$ be two distinct roots and let $v\in R_r$.
		Assign \dumpu{a_{p,v}}{p}.
		If $p\notin \{s,t\}$, assign \dumpp{b_{p,v}}{s}.
		If $p=s$, assign \dumpp{b_{p,v}}{r}.
		If $p=t$ and $r=s$, assign \dumpp{b_{p,v}}{p^\star}.
		If $p=t$ and $r\ne s$, \dumpp{b_{p,v}}{r}.
		
		\item \textbf{Between two children of the same tree.}
		For every root $p$ and every two distinct children $u,v\in R_p$, assign \dumpp{a_{u,v}}{p}.
		If $p\neq s$, assign \dumpp{b_{u,v}}{s}.
		If $p=s$, 
		assign \dumpp{b_{u,v}}{p^\star}.

		\item \textbf{Between two children of two different trees.}
		Let $p$ and $r$ be two distinct roots with $r\neq s$, let $u\in R_p$, and let $v\in R_r$.
		Assign \dumpp{a_{u,v}}{s}.
		If $p\neq s$, assign \dumpp{b_{u,v}}{R(u,v)}.
		If $p=s$, assign \dumpp{b_{u,v}}{r}.
	\end{enumerate}
	
	This completes the allocation.

	\paragraph{Correctness of Dumping Phase.}
	Denote the allocation right before dumping phase  by $\X$, and denote the allocation at the end by $\xp$. 
	Next, we prove that the final allocation, $\xp$, is \efx.
	
	First, note that $\X$ is \pone, and by the fact that none of (C0) and (C1) were applicable, we get that $\X$ follows global properties before the dumping phase (\cref{pbd}), and our dumping rules follow our general dumping structure. Thus, by \cref{lem:dumprule},
	\begin{enumerate}
		\item For every agent $p$, we have $X'_p \greatereqval{p} X_p$.
		\item Nobody strongly envies any agent who was resented at the start of the dumping phase.
		\item For every resent $p\to q$ before the dumping phase, agent $q$ does not envy agent $p$ after the dumping phase.
		\item No one envies agent $s$, and agent $s$ does not strongly envy anyone.  
	\end{enumerate}

    To prove that the final allocation is $\efx$, we show that any agent $p \ne s$ who was non-resented in $\X$, cannot be strongly envied in $\xp$. To do so, we prove the following:

    \begin{itemize}
        \item \textbf{No one envies $t$, and $t$ envies no one.}
        For every agent $i$, $X_t\cap E_i$ is a unit bundle in $E_{t,u}$; hence,
        $X'_t \lowereqval{i} a_{i,t} \lowereqval{i} X_i \lowereqval{i} X'_i$, where the second
        inequality comes from the fact that $t$ was non-resented in $\X$.

	
	
	We next show that agent $t$ does not envy any agent $p$ who was non-resented in $\X$. Fix a root $p\neq s$. 
    If $p$ did not \weakresent\ any agent in $\X$, then we get that $p\ne p^\star$ and $p$ did not 
    resent any agent in $\X$, so $X'_p\cap E_t$ is a unit bundle in $E_{t,p}$, so       
    $X'_p \lowereqval{t} a_{t,p} \lowereqval{t} X_t \lowereqval{t} X'_t$, where the second
    inequality comes from the fact that $p$ was non-resented in $\X$.
    
    Next, assume that $p$ \weakresented\ at least one agent in $\X$.
    Note that for any non-resented agent $r \notin \{t, p\}$, $X_p' \cap E_{r, t} = \emptyset$. Therefore, we have 
    $$X'_p \cap E_t = (X'_p \cap E_{t, p}) \bigcup_{q \in R(\X)} (X_p' \cap E_{t, q}) \lowereqval{t} a_{t,p} \sumsq (B_t \setminus E_{t,p}) \lowereqval{t} X_t \lowereqval{t} X'_t,$$
    where the first inequality comes from the fact that for a resented agent $q$, we have
	$X'_p \cap E_{t,q} \in \{b_{t,q}, \emptyset\}$ and $b_{t, q} \lowereqval{t} b_{q, t}$, and the second inequality comes from the fact that $p$ \weakresented\ at least one agent in $\X$, and (C2) was not applicable in $\X$.

    \item \textbf{Two roots different from $s$ and $t$ do not envy each other.}
	Fix two distinct roots $i$ and $j$ with $i,j\notin \{s,t\}$.
	The only unit bundle incident to $i$ that can be allocated to $j$ during the dumping phase
	is the unit bundle from $E_{i,j}$ assigned by Rule~(ii),
    so $X'_j \lowereqval{i} X'_j \cap E_{i,j} \lowereqval{i} a_{i,j} \lowereqval{i} X_i \lowereqval{i} X'_i$, where the third inequality comes from the fact that $j$ was non-resented in $\X$.

    \item \textbf{No resented agent envies an agent $p$ who was non-resented in $\X$.}
    Note that because no one envies $s$ and $t$ in $\xp$, we can assume $p\notin \{s,t\}$.
	For a resented agent $v$, by \cref{lem:dumprule}, we only need to consider the case where $r\to v$ for $r\ne p$.
    If $p$ does not \weakresent\ anyone, then since $p\ne s$, we get 
    $X'_p\cap E_v =a_{p,v}\lowereqval{v} X_v = X'_v$,
    where the inequality comes from the fact $p$ was non-resented in $\X$.
    Next, assume $p$ \weakresents\ at least one agent in $\X$.

	First, assume $r\ne s$, then if there exists a resent edge $p\to u$ such that $R(u,v)=p$, then 
	$$X'_p \cap E_v \subseteq a_{p,v} \sumsq   \bigsumsq_{q \in R_p} (X'_p \cap E_{v,q})
	\lowereqval{v}  a_{v,p} \sumsq D_v(\X)  \lowereqval{v} X_v = X'_v,$$
	where the last inequality comes from the fact  $R(u,v)=p$.
	Also, if for every $p\to u$, we have $R(u,v)=r$, then 
	$$X'_p \cap E_v =a_{p,v} \lowereqval{v} X'_v.$$

    Next, assume that $r=s$.
    $$X'_p \cap E_v \subseteq a_{p,v} \sumsq   \bigsumsq_{q \ne p} (X'_p \cap E_{v,q})
	\lowereqval{v}  a_{v,p} \sumsq (B_v\setminus E_{v,p})  \lowereqval{v} X_v = X'_v,$$
    where the last inequality comes from the fact (C3) was not applicable for $\X$
    and the fact that $p$ \weakresents\ at least one agent in $\X$.	
    \end{itemize}
	Hence, the final allocation is a complete \efx\ allocation.
\end{proof}

\begin{lemma}\label{lem: stageCprop}
	Suppose $\X$ is a \pone\ allocation, and suppose there exist four agents $i, j, k, \ell$ such that  $\ell \to j$  and $k$ \weakresents\ $i$, 
	only agent $\ell$ may have more than one unit bundle, and
	$$D_j(\X) \sumsq a_{j,k} \greaterval{j} X_j.$$
	Then, we can construct a complete \efx\ allocation in polynomial time.
\end{lemma}
\begin{proof}
	We construct a support pair $(\ell,j)$ satisfying the assumptions of \Cref{lem:weak_support}.
	Choose $i$ to be the agent whom $k$ \weakresents\ and $a_{k,i}$ is a most preferred unit bundle for agent $k$, which is possible since the inequality in the lemma's condition is independent of the choice of agent $i$. We update the allocation as follows.
	Let $\X$ denote the allocation before the update and let $\X'$ denote the allocation after the update.
	We update the allocation as follows: (Use $\xp$ to denote the allocation right before executing \textsc{Reduce Trees}.)
	\begin{align*}
		&\text{Agents } \ell,j,k, \text{and } i \text{ drop their previous bundles}\\
		&X'_\ell \takes a_{\ell,j}\\
		&X'_j \takes a_{j,k} \sumsq D_j(\X) \\
		&X'_k \takes a_{k,i}\\
		&X'_i \takes \textsc{choose}(i)\\
		&Run \textsc{ Reduce Trees}(\xp,i).
	\end{align*}
	Let $H =\{\ell,j,k\}$ and let $h_0=i$. Note that the allocation before executing the \textsc{ Reduce Trees}(i) satisfies the conditions of \cref{lem:treebreaker2}.
	Then, by \cref{lem:treebreaker2}, we get a \pone\ allocation,
	$(\ell,j)$ is a support pair, $j$ is the unique agent that may hold more than one unit bundle,
    and $h_0=i$ is non-resented in $\xp$.
	Moreover, agent $j$ does not resent anyone because $v_j(X_j)$ increases and $j$ did not resent anyone before the update.
	Also, since agent $i$ had $a_{i,k}$, and she did not resent anyone in $\X$, agent $i$'s most valuable unit bundle is $a_{i,k}$. Therefore, agent $i$ \weakresents\ agent $k$ after the update, so there is some $p^* \to q^*$ for $p^*\ne \ell$.
	Therefore, the assumptions of \Cref{lem:weak_support} hold, and the claim follows.
\end{proof}

\begin{corollary}\label{lem: stageC}
	If an allocation is in \stage\ C, we can construct a complete \efx\ allocation in polynomial time.
\end{corollary}
\begin{proof}
	Every allocation in \stage\ C, satisfies the conditions of \cref{lem: stageCprop}, so we get a complete \efx\ allocation.
\end{proof}

\section{\Stage\ D}\label{sec:d}
    In this section, we show that given an allocation in \stage\ D, we can construct a complete \efx\ allocation in polynomial time.	
    Throughout this section, we have an allocation $\X$ in \stage\ D.
    We suppose the root of the single non-trivial resent tree is $p$
    who resents agents $q_1,\dots,q_k$, sorted in increasing order by $v_p(a_{p,q_i})$.
    We let the rest of the agents be $r_1,\dots,r_\ell$, sorted in increasing order by $v_{q_k}(a_{q_k,r_i})$.
    Also, $r_i$'s are non-resented and do not resent anyone.
    We also define $q=q_k, r_{\ell+1} =p$, and $d=k-1$. Note that since allocation $\X$ is \pthree, and since $p$ resents $q_i$, we
    have $X_{q_i} = a_{q_i,p}$.

\begin{lemma}\label{lem:i>1}
	Given an allocation in \stage\ D, either we can construct a complete \efx\ allocation in polynomial time, or:
	\begin{itemize}
		\item For all $1<i\leq k$, we have
		$ a_{q_i,p} \greatereqval{q_i}  b_{p,q_i} \sumsq A_{q_i}(\X) .$
	\end{itemize}
\end{lemma}
\begin{proof}
	We suppose there exists some $i > 1$ such that
	$$ a_{q_i,p} \lowerval{q_i} b_{p,q_i} \sumsq A_{q_i}(\X),$$
	and we construct a complete \efx\ allocation.
	First, update the allocation with Update 1, as follows:
	
	\begin{align*}
		X'_{q_i} &\takes b_{p,q_i} \sumsq A_{q_i}(\X) \\
		X'_p &\takes a_{p,q_i}.
	\end{align*}
	Agent $p$ no longer resents agents $q_1,\ldots,q_i$. 
	By \cref{lem:cycle_support}, allocation remains \pone, and no new resent is created.
	Hence, agents $q_1$ and $p$ are non-resented, and $q_1$ does not resent anyone.
	Also, since $i\ne 1$, we have $X_{q_1}=a_{q_1,p}$, so $(q_1,p)$ is a support pair. 
	Thus, using the fact that there is no resent edge $v\to u$ with $v\ne p$, by \cref{lem:added_strong_support},
	we get a complete \efx\ allocation.
\end{proof}

\begin{definition}[Update D] \label{updateD}
	Given a \pthree\ allocation $\X$, with root $p$ resenting $q_1,\ldots,q_k$, and with non-resented agents $r_1\ldots,r_\ell$, where
	$q_i$'s are sorted by an increasing order of $v_p(a_{p,q_i})$, define the following update as Update D:
	\begin{align*}
		&\text{Agents } p \text{ and } q_k \text{ drop their previous bundles}\\
		&X_p \takes a_{p,q_k}\\
		&X_{q_k} \takes \textsc{Choose}(q_k).
	\end{align*}
\end{definition}
\begin{observation}\label{lem:updateD}
	After execution of Update D:
	\begin{enumerate}
		\item The resulting allocation $\hat{\X}$ is \pthree.
		\item For any agent $t\ne q_k$, agent $t$ does not resent anyone.
		\item $q_k$ \weakresents\ agent $p$.
		\item Agents $q_1,\ldots,q_k$ are non-resented.
	\end{enumerate}
\end{observation}
\begin{proof}
	Let $\hat{\X}$ be the allocation after Update D.
	Since $p$ resented $q_k$ in $\X$, we have $X_{q_k}=a_{q_k,p}$.
	Moreover, $q_k$ did not resent anyone in $\X$, and therefore $a_{q_k,p}\greatereqval{q_k}a_{q_k,j}$ for every agent $j$.

	Agent $p$ does not resent anyone in $\hat{\X}$ since
	by the ordering of the $q_i$'s, we have $a_{p,q_k} \greatereqval{p} a_{p,j}$ for any other agent $j$.

	Every agent $t\notin\{p,q_k\}$ keeps the same bundle and did not resent anyone in $\X$.
	Thus, by the definition of resent, $t$ does not resent anyone in $\hat{\X}$.
	Hence, only $q_k$ may resent some agents.

	The allocation $\hat{\X}$ is still an orientation and unitary.
	If $q_k$ resents an agent $u$ in $\hat{\X}$, then $a_{q_k,u}\greaterval{q_k}\hat{X}_{q_k}$.
	Since $\hat{X}_{q_k}=\textsc{Choose}(q_k)$, the bundle $a_{q_k,u}$ was not available to $q_k$ after $p$ and $q_k$ dropped their previous bundles.
	Hence, $\hat{X}_u=a_{u,q_k}$.
	Therefore, $\hat{\X}$ is \pzero.

	Since only $q_k$ may have outgoing resent edges, every resent tree has height at most one.
	Also, every agent has at most one unit bundle.
	Therefore, $\hat{\X}$ is \pthree.

	Finally, $q_k$ is non-resented because no agent other than $q_k$ resents anyone.
	Also, $p$ does not resent anyone, $\hat{X}_p=a_{p,q_k}$, and $a_{q_k,p}\greatereqval{q_k}a_{q_k,j}$ for every agent $j$.
	Hence, $q_k$ \weakresents agent $p$.

	For every $i\ne k$, if $q_k$ resented $q_i$, the \pzero\ property would imply $\hat{X}_{q_i}=a_{q_i,q_k}$, contradicting $\hat{X}_{q_i}=a_{q_i,p}$.
	Since no other agent resents anyone, each $q_i$ with $i\ne k$ is non-resented.
\end{proof}

First, we prove a weaker version of \stageD,
which is when there exists no agent outside of the single resent tree.

\begin{lemma}\label{lem:lone_tree}
	Suppose $\X$ is a \pthree\ allocation, and suppose there exists only one resent tree, 
	and no vertex outside the resent tree exists, i.e., $\ell=0$.
	In this case, we can construct a complete \efx\ allocation.
\end{lemma}
\begin{proof}
	We divide to two cases.
	
	\paragraph{Case 1:}
	If $a_{q_1,p} \sumsq b_{q_k,q_1} \greatereqval{q_1} b_{q_1,p} \sumsq a_{q_k,q_1}$:
	
	In this case, we update the allocation by update D, defined in \cref{updateD}.
	Denote the resulting allocation by $\hat{\X}$. 
	By \cref{lem:updateD}, $\hat{\X}$ is \pthree, agents $q_1,\ldots,q_k$ are non-resented in $\hat{\X}$, and $q_k$ \weakresents\ agent $p$.
	Since $\X$ is not in \stage\ B and $p\to q_k$, we have
	$$b_{p,q_k}\sumsq A_p(\hat{\X})=   \bigsumsq_{i\in [k]} b_{q_i,p} = X_p\sumsq A_p(\X)\lowerval{p}a_{p,q_k}=\hat{X}_p.$$
	
	\paragraph{Dumping Phase of Case 1.}
	We allocate every remaining unallocated unit bundle by the following rules.
	
	\begin{enumerate}[label=(\roman*), leftmargin=2.2em]
		\item \textbf{Between $q_k$ and $p$.}
		We have \dumpp{a_{p,q_k}}{p}, and we assign \dumpp{b_{p,q_k}}{q_k}.
		
		\item \textbf{Between non resented $q_i$ and $p$.}
		For $i<k$, we assign \dumpu{a_{q_i,p}}{q_i} and \dumpp{b_{q_i,p}}{q_k}.
		
		\item \textbf{Between $q_i$ and $q_j$.}
		For $i < j \le k$, we assign \dumpp{a_{q_j,q_i}}{q_j}, and \dumpp{b_{q_j,q_i}}{q_i}.
	\end{enumerate}
	\noindent After applying Rules~(i) to (iii), every unit bundle in the instance is allocated.
	
	\paragraph{Correctness of Dumping Phase of Case 1.}
	Denote the resulting allocation by $\xp$.
	We prove that $\xp$ is $\efx$.
	First, note that $\hat{\X}$ satisfies the global dumping properties: by \cref{lem:updateD}, it is \pone, and the displayed inequality verifies the second dumping property. Moreover, our dumping rules follow the general dumping structure,
	so by \cref{lem:dumprule},
	\begin{enumerate}
		\item For every agent $t$, we have $X'_t \greatereqval{t} \hat{X_t}$.
		\item Nobody strongly envies agent $p$.
		\item Agent $p$ does not envy agent $q$ after the dumping phase.
	\end{enumerate}
	
	For every $i<k$, $q_k$ does not envy  $q_i$ since $X'_i\cap E_{q_k} = b_{q_k,q_i} \lowereqval{q_k} a_{q_k,q_i} \subseteq X'_{q_k}$.
	Agent $q_k$ is not envied by $q_1$ because, by case 1 condition, we have:
	$$X'_{q_1} \greatereqval{q_1} a_{q_1,p} \sumsq b_{q_k,q_1} 
	\greatereqval{q_1} b_{q_1,p} \sumsq a_{q_k,q_1} = X'_{q_k}\cap E_{q_1}.$$

	By \cref{lem:i>1}, since $a_{q_i,q_k}\subseteq A_{q_i}(\X)$, for every $1 <i <k$, we have
	$$a_{q_i,p} \greatereqval{q_i} b_{p,q_i}\sumsq A_{q_i}(\X)\greatereqval{q_i} a_{q_i,q_k}\sumsq b_{p,q_i}.$$

	Hence, for $i >1$, using \cref{lem:i>1}, agent $q_k$ is not envied by $q_i$ because
	$$X'_{q_i} \greatereqval{q_i} a_{q_i,p}  \greatereqval{q_i} b_{p,q_i} \sumsq A_{q_i}(\X) \greatereqval{q_i}
	a_{q_k,q_i} \sumsq b_{q_i, p} = X'_{q_k} \cap E_{q_i}.$$
	
	For $i,j <k$, $q_i$ does not envy  $q_j$ since 
	$X'_j\cap E_{q_i} \lowereqval{q_i} a_{q_i,q_j} \lowereqval{q_i} a_{q_i,p} \subseteq X'_{q_i}$.
	Hence, $\xp$ is a complete \efx\ allocation.

	\paragraph{Case 2:} Suppose we are not in case 1, so $a_{q_1,p} \sumsq b_{q_k,q_1} \lowereqval{q_1} b_{q_1,p} \sumsq a_{q_k,q_1}$.
	We update the allocation as follows to get allocation $\xp$:
	\begin{align*}
		X'_p &\takes a_{p,q_1} \sumsq \bigsumsq_{i>1} b_{q_i,p} \sumsq 
		\bigsumsq_{j>i>1} a_{q_j,q_i} \sumsq \bigsumsq_{j>1} b_{q_1,q_j} \\
		X'_{q_1} &\takes b_{p,q_1} \sumsq \bigsumsq_{j>1} a_{q_1,q_j} \sumsq \bigsumsq_{j>i>1} b_{q_j,q_i}\\
		X'_{q_i} &\takes a_{q_i,p}  \hspace{1cm} :\forall 2\leq i \leq k.
	\end{align*}
	
	\noindent Every unit bundle in $\xp$ is allocated. Next, we show that it is \efx.
	
	\noindent First, note that for $i\geq 2$, no one strongly envies $q_i$, since $X'_{q_i} \cap E_{q_1} = \emptyset$, and 
	$b_{q_i,p} \subseteq X'_{p}$, so $p$ does not strongly envy $a_{p,q_i}$. 
	
	\noindent\textbf{No one envies agent $p$.}
	For  $i\geq 2$, by \cref{lem:i>1}, we get:
	$$X'_p \cap E_{q_i} =  b_{q_i,p} \sumsq  b_{q_1,q_i}  \sumsq \bigsumsq_{j: j>i} a_{q_j,q_i} \sumsq \bigsumsq_{j: i>j>1} a_{q_i,q_j}
	\lowereqval{q_i} b_{q_i,p} \sumsq A_{q_i}(\X) \lowereqval{q_i} a_{q_i,p} =X'_{q_i},$$
	Hence, $q_i$ does not envy agent $p$. 
	
	Moreover, since we are not in case 1, we get:
	$$a_{q_1,p} \sumsq b_{q_k,q_1} \lowereqval{q_1} b_{q_1,p} \sumsq a_{q_k,q_1} 
	\lowereqval{q_1} b_{p,q_1} \sumsq a_{q_1,q_k},$$
	so agent $q_1$ does not envy agent $p$ because:
	$$X'_p \cap E_{q_1} =  a_{p,q_1} \sumsq  \bigsumsq_{j>1} b_{q_1,q_j}
	\lowereqval{q_1} a_{q_1,p} \sumsq b_{q_k,q_1} \sumsq  \bigsumsq_{k>j>1} b_{q_1,q_j}
	\lowereqval{q_1} b_{p,q_1} \sumsq a_{q_1,q_k} \sumsq  \bigsumsq_{k>j>1} a_{q_1,q_j} \subseteq X'_{q_1}.$$
	
	\noindent\textbf{No one envies agent $q_1$.}
	Agent $p$ does not envy $q_1$ since $X'_{q_1} \cap E_{p} =  b_{p,q_1} \lowereqval{p} a_{p,q_1} \subseteq X'_p$.
	For  $i\geq 2$, by \cref{lem:i>1}, we get:
	$$X'_{q_1} \cap E_{q_i} =  a_{q_1,q_i} \sumsq  \bigsumsq_{j: j>i} b_{q_j,q_i} \sumsq \bigsumsq_{j: i>j>1} b_{q_i,q_j}
	\lowereqval{q_i} A_{q_i}(\X) \lowereqval{q_i} a_{q_i,p} =X'_{q_i}.$$
	Hence, $q_i$ does not envy agent $q_1$. Therefore, $\xp$ is a complete \efx\ allocation, and the proof is complete. 
\end{proof}

From now on, we may suppose that along with our single tree,
there exists at least one non-resented agent, i.e., $\ell \geq 1$.

\begin{lemma}\label{lem:D1}
	Given an allocation in \stage\ D, either we can construct a complete \efx\ allocation, or
	for every non-resented agent $r\neq p$, we have $X_{r} = a_{r,q_k}$.
	Moreover, if we execute update D, agent $q$ would resent all the agents $r_1,\ldots,r_\ell,p$ in the resulting allocation.
\end{lemma}
\begin{proof}
	We suppose that the conclusion does not hold, and then we construct a complete \efx\ allocation.
	
	First, suppose there exists a non-resented agent $r\neq p$ such that $X_r\ne a_{r,q_k}$.
	If $X_r \subseteq E_{r,v}$ for a non-resented agent $v$, then $(r,v)$ is a support pair.
	Moreover, agent $r$ does not resent anyone.
	Therefore, by \cref{lem:strong_support}, we can construct a complete \efx\ allocation in this case.
	
	So assume that $X_r \subseteq E_{r,q_i}$ for some $q_i$.
	Since $X_{q_i}=a_{q_i,p}$, we get $X_r= a_{r,q_i}$.
	If $q_i\neq q_k$, then we execute Update D as defined in \cref{updateD}.
	By \cref{lem:updateD}, the resulting allocation is \pthree.
	Moreover, agent $q_i$ is non-resented, agent $r$ is non-resented, agent $r$ does not resent anyone, and $X_r=a_{r,q_i}$.
	Consequently, $(r,q_i)$ is a support pair, and by \cref{lem:strong_support}, we get a complete \efx\ allocation.
	
	Hence, we may assume that for every non-resented agent $r_i$, we have $X_{r_i}=a_{r_i,q_k}$.
	Now execute Update D, and denote the resulting allocation by $\hat{\X}$.
	By \cref{lem:updateD}, agent $q$ \weakresents\ agent $p$, and every agent other than $q$ does not resent anyone.
	If agent $q$ does not resent some agent $r_i$, then $r_i$ is non-resented in $\hat{\X}$.
	Moreover, $q$ is non-resented, agent $r_i$ does not resent anyone, and $\hat{X}_{r_i}=a_{r_i,q}$.
	Thus, $(r_i,q)$ is a support pair.
	Therefore, by \cref{lem:strong_support}, we can construct a complete \efx\ allocation.
	
	Thus, unless we are already done, after Update D agent $q$ resents every agent $r_1,\ldots,r_\ell$.
    Also, note that $a_{q,p} \greatereqval{q} a_{q,r_i}$ for every $i\in [\ell]$, since agent $q$ 
    did not resent anybody while possessing $a_{q,p}$.
    Hence, $q$ resents agent $p$ in $\hat{\X}$..
\end{proof}

    Recall that $q:= q_k$, $r_{\ell+1}:=p$, and $d = k-1$.
    Denote the original allocation by $\X$ and the updated allocation using Update D by $\hat{\X}$. 

\noindent From now on, by \cref{lem:D1}, we may suppose that for every $i\in [\ell]$, we have $X_{r_i} = a_{r_i,q_k}$, and
in $\hat{\X}$, agent $q$ resents agents $r_1,\ldots,r_\ell,r_{\ell+1}=p$.

Now, it is important to notice a certain duality in the structure of our graph.
Agent $p$ (in $\X$) is analogous to $q$ (in $\hat{\X}$). 
Agents $(q_1,\dots, q_d)$ (in $\X$)  are analogous  to $r_1,\dots r_\ell$ (in $\hat{\X}$).

\begin{itemize}
	
	\item In $\X$,
	$p=r_{\ell+1}$ is non resented and resents $q_1,\ldots,q_d,q_{d+1}=q$.
	Agent $p$ most-resents $q$. Also, for every $i\in [\ell]$, we have $X_{r_i} = a_{r_i,q}$.
		
	\item In $\hat{\X}$,
	$q=q_k$ is non resented and resents $r_1,\ldots,r_\ell,r_{\ell+1}=p$.
	Agent $q$ most-resents $p$. Also, for every $i\in [d]$, we have $X_{q_i} = a_{q_i,p}$.
\end{itemize}

    If $\ell = 0$, then using \cref{lem:lone_tree} for allocation $\X$, we get a complete \efx\ allocation.
    If $d = 0$, then using \cref{lem:lone_tree} for allocation $\hat{\X}$, we get a complete \efx\ allocation.
    Hence, we may assume $\ell\ne 0$ and $d\ne 0$.

\begin{corollary}\label{claim:D2}
	Given an allocation in \stage\ D, either we can construct a complete \efx\ allocation, or:
	\begin{itemize}
		\item For all $1<i\leq d$, we have
		$ a_{q_i,p} \greatereqval{q_i}  b_{p,q_i} \sumsq A_{q_i}(\X) =  b_{p,q_i} \sumsq \bigsumsq_{u\neq q_i,p} a_{q_i,u}.$
		\item For $q=q_{d+1}$, we have
		$a_{q,p} \greatereqval{q} b_{p,q} \sumsq A_{q}(\X) = 
		b_{p,q} \sumsq \bigsumsq_{i\in [d]} a_{q,q_i} \sumsq \bigsumsq_{j\in [\ell]} b_{r_j, q}.$
		\item For all $1<j\leq \ell$, we have
		$ a_{r_j,q} \greatereqval{r_j} b_{q,r_j} \sumsq A_{r_j}(\hat{\X}) = b_{q,r_j} \sumsq \bigsumsq_{u\neq r_j,q} a_{r_j,u}.$
		\item For $p=r_{\ell+1}$, we have
		$ a_{p,q} \greatereqval{p} b_{q,p} \sumsq A_{p}(\hat{\X}) =
		b_{q,p} \sumsq \bigsumsq_{j\in [\ell]} a_{p,r_j} \sumsq \bigsumsq_{i\in [d]} b_{q_i,p}.$
	\end{itemize}
\end{corollary}
\begin{proof}
	The first two properties are obtained by \cref{lem:i>1}, and the fact that for every $r_i$, we have $X_{r_i}=a_{r_i,q}$.
	By duality of $\X$ and $\hat{\X}$, we should have the same properties for allocation $\hat{\X}$, so the third and the fourth property are 
	obtained.
\end{proof}

\begin{lemma}\label{lem:s gurantee}
	Suppose allocation $\X$ is \pone, and in the corresponding resent graph there exists only one non-trivial tree.
	Denote the root of this non-trivial tree by $t$, and let $h_1,\ldots,h_z$ be the agents whom agent $t$ resents.
	In addition, suppose there exists a non-resented agent $s\ne t$ such that:
	\begin{enumerate}
		\item For every resented agent $h$, 
		we have $X_{h} \greaterval{h} b_{h,t} \sumsq A_{h}(\X)$ 
		and $X_{h} \greaterval{h} a_{h,s} \sumsq (A_{h}(\X)\setminus E_{h,s})$.
		\item $X_s \greatereqval{s} a_{s,t} \sumsq D_s(\X)$.
		\item For every non-resented agent $i\notin \{s,t\}$, 
		either $X_i\cap E_{i,t}=a_{i,t}$ or $X_i\cap E_{i,s}=a_{i,s}$.    
	\end{enumerate}
	Then, there exists a complete \efx\ allocation.
\end{lemma}
\begin{proof}
	We allocate the unallocated unit bundles in the dumping phase.
	\begin{enumerate}[label=(\roman*), leftmargin=2.2em]
		\item \textbf{Root to root.}
		Let $p$ and $r$ be two distinct roots. 
		If at least one of the unit bundles is already allocated to one of the roots, allocate the other unit bundle to the other root. 		
		Otherwise, allocate one unit bundle in $E_{p,r}$ to $p$ and the other to $r$ arbitrarily.

		\item \textbf{Root to its own children.}
		For root $t$ and every resented $h\in R_t$, assign \dumpp{b_{h,t}}{t}.
		Also, we already have $X_{h} = a_{h,t}$.	
		
		\item \textbf{Non resented $s$ to others children.}
		Let $h\in R_t$. Assign \dumpu{a_{s,h}}{s} and assign \dumpp{b_{s,h}}{t}. 
		
		\item \textbf{Non resented $i\notin \{s,t\}$ to others children.}
		Let $i\notin \{s,t\}$ be a non-resented and let $h\in R_t$.
		Assign \dumpu{a_{i,h}}{i}.
		If $X_i\cap E_{i,t}=a_{i,t}$, assign \dumpp{b_{i,h}}{t}.
		Otherwise, assign \dumpp{b_{i,h}}{s}. (Note that in the second case, by lemma's conditions, we have $X_i\cap E_{i,s}=a_{i,s}$.)

		\item \textbf{Between two children of the same tree.}
		For every two distinct children $h,u\in R_t$, assign \dumpp{a_{u,h}}{t},
		and assign \dumpp{b_{u,h}}{s}.  (Note that $s\ne t$ is the non-resented agent in the lemma's conditions.)
	\end{enumerate}
	\noindent After applying Rules~(i) to (v), every unit bundle in the instance is allocated.
	
	\paragraph{Correctness of Dumping Phase.}
	Denote the allocation right before dumping phase  by $\X$, and denote the allocation at the end by $\xp$.
	Next, we prove that the final allocation, $\xp$, is \efx.
	
	First, note that by the assumptions of the lemma, $\X$ satisfies the global dumping properties of \cref{pbd}: it is \pone, and for every resent edge $t\to h$, we have $X_h=a_{h,t}\greaterval{h} b_{h,t}\sumsq A_h(\X)$. Moreover, our dumping rules follow the general dumping structure,
	so by \cref{lem:dumprule},
	\begin{enumerate}
		\item For every agent $p$, we have $X'_p \greatereqval{p} X_p$. 
		\item Nobody strongly envies any agent who was resented at the start of the dumping phase.
		\item For every agent $h$ who was resented in $\X$, agent $h$ does not envy agent $t$ after the dumping phase.
	\end{enumerate}
	
	\noindent\textbf{A resented agent $h$ does not envy anyone.}
	By the previous argument, we know $h$ does not envy agent $t$ or any other resented agent.
	So let $p\ne t$ be an agent who was non-resented in $\X$.
	
	If $p\ne s$, then we get $X'_p\cap E_h = a_{p,h} \lowereqval{h} a_{h,p} \lowereqval{h} a_{h,t}= X'_h$, 
	where the first inequality follows from the definition of unit bundles, and the second one follows from the fact that $h$ does not resent $p$ in $\X$.
	Moreover, if $p=s$, agent $h$ does not envy $s$ since:
	\begin{align*}
		X'_s \cap E_h &= (X'_s \cap E_{h,s}) \sumsq \bigsumsq_{u\ne s,t} (X'_s \cap E_{h,u}) \\
		&\lowereqval{h} a_{s,h} \sumsq   (A_{h}(\X) \setminus E_{h,s})  \lowereqval{h}  X_h=X'_h,
	\end{align*}
	where the first inequality comes from the fact $X'_s \cap E_{h,u}$ is either empty or an unallocated unit bundle in $\X$.
	Also, the second inequality comes from lemma's conditions.
	
	\noindent\textbf{No one envies an agent $r$ who was non-resented in $\X$.}
	By our previous argument, we need to show that if $p\ne r$ is an agent who was non-resented in $\X$, 
	then $p$ does not envy $r$ in $\xp$.
	
	\textbf{If} $\mathbf{r\notin \{s,t\}:}$ then $X'_r\cap E_p$ is a unit bundle, so
    $X'_r\cap E_p \lowereqval{p} a_{p,r} \lowereqval{p} X_p \lowereqval{p} X'_p$, where the 
    second inequality comes from the fact $r$ was non-resented in $\X$.
	
	\textbf{If} $\mathbf{p,r\in \{s,t\}:}$ 
	If $r=s$, and if $p=t$, then $t$ does not envy $s$ since $X'_s\cap E_t\lowereqval{t} a_{t,s}\lowereqval{t} X_t\lowereqval{t} X'_t$.
	Also, if $r=t$ and $p=s$, then 
	\begin{align*}
		X'_t \cap E_s &= (X'_t \cap E_{t,s}) \sumsq \bigsumsq_{h \in R(\X)} (X'_t\cap E_{s,h})\\
		&\lowereqval{s} a_{s,t} \sumsq 	\bigsumsq_{h \in R(\X)} b_{s,h} \\
		&= a_{s,t} \sumsq D_s(\X) \lowereqval{s} X_s \lowereqval{s} X'_s.
	\end{align*}
	\textbf{If} $\mathbf{r\in \{s,t\}}$ \textbf{and} $\mathbf{p\notin \{s,t\}:}$
	We divide to two cases.
    
	If $X_p\cap E_{p,t}=a_{p,t}$: then by rules (i) and (iv), 
    we get that $X'_s\cap E_{p}$ is a unit bundle that was in $E_{s,p}$, 
	so $X'_p \greatereqval{p} X'_p \greatereqval{p} a_{p,s} \greatereqval{p} X'_s$. 
	Moreover, by rule (i), the unit bundle of agent $p$ from $E_{p,t}$ does not change, 
    so $X'_p\cap E_{p,t}=a_{p,t}$.
	Hence,
	\begin{align*}
		X'_p &\greatereqval{p} (X'_p \cap E_{p,t}) \sumsq \bigsumsq_{h \in R(\X)} (X'_p\cap E_{p,h})\\
		&= a_{p,t} \sumsq \bigsumsq_{h \in R(\X)} a_{p,h}\\
		&\greatereqval{p} b_{p,t} \sumsq \bigsumsq_{h \in R(\X)} b_{p,h}\\
		&= X'_t \cap E_p. 
	\end{align*}
	If $X_p\cap E_{p,t}\ne a_{p,t}$, then by lemma's conditions, we have that $X_p\cap E_{p,s}=a_{p,s}$. So in a similar way, we get that
	$X'_t\cap E_{p}$ is a unit bundle that was in $E_{t,p}$, 
	so agent $p$ does not envy $t$ in $\xp$. Also,
	\begin{align*}
		X'_p &\greatereqval{p} a_{p,s} \sumsq \bigsumsq_{h \in R(\X)} a_{p,h}\\
		&\greatereqval{p} b_{p,s} \sumsq \bigsumsq_{h \in R(\X)} b_{p,h} = X'_s \cap E_p.
	\end{align*}
	Hence, $\xp$ is a complete \efx\ allocation.
\end{proof}

\begin{lemma}\label{lem:aq>ap}
	Given an allocation in \stage\ D, either we can construct a complete \efx\ allocation, or:
	$$ a_{r_1,q} \sumsq b_{p,r_1} \sumsq \bigsumsq_{i\in [d]} b_{q_i, r_1}  \greaterval{r_1}  
	b_{q,r_1} \sumsq a_{r_1,p} \sumsq  \bigsumsq_{i\in [d]} a_{r_1,q_i}.$$  
\end{lemma}
\begin{proof}
	Suppose the inequality does not hold, then we construct a complete \efx\ allocation.
	In order to do that, we update allocation $\X$ as follows:
	\begin{align*}
		&\text{Agents } p,q \text{ and } r_1 \text{ drop their previous bundles},\\
		&X'_{r_1} \takes b_{q,r_1} \sumsq a_{r_1,p} \sumsq  \bigsumsq_{i\in [d]} a_{r_1,q_i}  \sumsq  \bigsumsq_{1<i\leq \ell} a_{r_1,r_i} \\
		&X'_{q} \takes a_{q,r_1} \\
		&X'_{p} \takes a_{p,q}\\
		&X'_{q_i} \takes a_{q_i,p} \sumsq a_{q_i,q} \hspace{1cm} :\forall i \in [d].
	\end{align*}
	Denote the resulting allocation by $\xp$.
	We show that conditions of \cref{lem:s gurantee} are satisfied for $\xp$ with $t=q$ and $s=r_1$.
	
	Agent $q$ does not resent $r_1$ since $X'_{q} = a_{q,r_1}$.
	Moreover, no other agent resents $r_1$ since for every
	$u\notin \{q,r_1\}$, $X'_u\greatereqval{u} X_u \greatereqval{u} a_{u,r_1}$, where the last inequality comes from the fact that $r_1$ was non-resented in $\X$.

	Note that because $ a_{r_1,q} \sumsq b_{p,r_1} \sumsq \bigsumsq_{i\in [d]} b_{q_i, r_1}  \lowereqval{r_1}  
	b_{q,r_1} \sumsq a_{r_1,p} \sumsq  \bigsumsq_{i\in [d]} a_{r_1,q_i}$, we get:
	$$ X'_{r_1} = b_{q,r_1} \sumsq a_{r_1,p} \sumsq  \bigsumsq_{i\in [d]} a_{r_1,q_i}  \sumsq  \bigsumsq_{1<i\leq \ell} a_{r_1,r_i}
	\greatereqval{r_1} a_{r_1,q} \sumsq b_{p,r_1} \sumsq \bigsumsq_{1<i\leq \ell} b_{r_1,r_i}  
	\greaterval{r_1} a_{r_1,q} \sumsq D_{r_1}(\xp).$$
	
	Hence, $q$ is non-resented, and $X'_{s} \greaterval{a} a_{s,t} \sumsq D_{s}(\xp)$.
	In addition, every agent's value has weakly improved, 
	so no new resent is created, so allocation remains \pone.
	
	Note that for every resented agent $h$ in $\xp$, we have $h\in \{r_2,\ldots,r_\ell\}$, so by the third item of \cref{claim:D2},
	$$X'_{h} = a_{h,q} \greatereqval{h} b_{q,h} \sumsq A_{h}(\hat{\X})  \greatereqval{h}  b_{q,h} \sumsq A_{h}(\xp),$$ 
	where the last inequality comes from the fact that no incident unit bundle of agent $h$ gets unallocated in the transition of $\hat{\X}$ to $\xp$.
	Moreover, 
	$$X'_{h} \greatereqval{h} A_{h}(\hat{\X}) = a_{h,s} \sumsq (A_{h}(\xp)\setminus E_{h,s}).$$ 
	
	Additionally, for every non-resented agent $u \notin \{s,t\}$, we have $u=q_i$ for some $i\in [d]$, so we have 
	$X'_u\cap E_{u,t}=a_{u,t}$.
	Therefore, all the conditions of \cref{lem:s gurantee} are satisfied, so we get a complete \efx\ allocation.
\end{proof}

\begin{lemma}\label{lem:q1 aq>ap}
	Given an allocation in \stage\ D, if $d\geq2$, then either we can construct a complete \efx\ allocation, or :
	$$b_{p, q_1} \sumsq a_{q_1, q_d} \sumsq a_{q_1,q} \lowereqval{q_1} a_{q_1,p} \sumsq b_{q_d, q_1} \sumsq b_{q,q_1}.$$  
\end{lemma}
\begin{proof}
	We assume $b_{p, q_1} \sumsq a_{q_1, q_d} \sumsq a_{q_1,q} \greaterval{q_1} a_{q_1,p} \sumsq b_{q_d, q_1} \sumsq b_{q,q_1}$,
	and we show that a complete \efx\ allocation exists.
	In order to do that, we update $\X$ as such:
	\begin{align*}
		&\text{Agents } p \text{ and } q_1 \text{ drop their previous bundles},\\
		&X'_{q_1} \takes b_{p, q_1} \sumsq \bigsumsq_{2\leq i\leq k} a_{q_1, q_i},\\
		&X'_p \takes a_{p,q_1}, \\
		&X'_{q_i} \takes a_{q_i,p} \hspace{1cm} :\forall 2\leq i \leq d,\\
		&X'_{r_i} \takes a_{r_i,q} \sumsq a_{r_i,q_1} \hspace{1cm} :\forall i \in [\ell].
	\end{align*}
	Denote the resulting allocation by $\xp$.
	We show that conditions of \cref{lem:s gurantee} are satisfied for $\xp$ with $t=p$ and $s=q_1$.
	
	Agent $p$ no longer resents $q_1$, but may still resent agents in $\{q_2,\ldots,q_k\}$. Moreover, no other agent resents $q_1$ since for every
	$2\leq i \leq k$, unit bundle $a_{q_1,q_i}$ was unallocated in $\X$, and $\X$ was unitary.

	Note that because $b_{p, q_1} \sumsq a_{q_1, q_d} \sumsq a_{q_1,q} \greatereqval{q_1} 
	a_{q_1,p} \sumsq b_{q_d, q_1} \sumsq b_{q,q_1}$, we get:
	$$ X'_{q_1} = b_{p, q_1} \sumsq \bigsumsq_{2\leq i\leq k} a_{q_1, q_i}  \greatereqval{q_1} 
	a_{q_1,p} \sumsq \bigsumsq_{2\leq i\leq k} b_{q_1, q_i}  \greaterval{q_1} a_{q_1,p} \sumsq D_{q_1}(\xp).$$

	Hence, $p$ is still non-resented, and $X'_{s} \greaterval{s} a_{s,t} \sumsq D_{s}(\xp)$.
	In addition, the allocation remains an orientation and unitary. By the previous paragraph, only agent $p$ may resent agents, and every agent resented by $p$ in $\xp$ has bundle $a_{h,p}$. Hence, the allocation is \pzero. Since its resent graph has height at most one, the allocation is \pone.
	
	Note that for every resented agent $h$ in $\xp$, we have $h\in \{q_2,\ldots,q_k\}$, so by the first two items of \cref{claim:D2},
	$$X'_{h} = a_{h,p} \greatereqval{h}  b_{p,h} \sumsq A_{h}(\X) \greatereqval{h}  b_{p,h} \sumsq A_{h}(\xp),$$ 
	where the last inequality comes from the fact that no incident unit bundle of agent $h$ gets unallocated in the transition of $\X$ to $\xp$.
	Moreover, for every such $h$, no incident unit bundle of $h$ outside $E_{h,s}$ changes in the transition from $\X$ to $\xp$. Therefore,
	$$A_h(\xp)\setminus E_{h,s}=A_h(\X)\setminus E_{h,s},$$
	and hence
	$$X'_h\greatereqval{h}  A_{h}(\X) = a_{h,s}\sumsq(A_h(\xp)\setminus E_{h,s}).$$
	
	Additionally, for every non-resented agent $u \notin \{s,t\}$, either $u=r_i$ for some $i\in [\ell]$, in which case $X'_u\cap E_{u,s}=a_{u,s}$, or $u=q_i$ for some $2\leq i\leq k$, in which case $X'_u\cap E_{u,t}=a_{u,t}$.
	Therefore, all the conditions of \cref{lem:s gurantee} are satisfied, so we get a complete \efx\ allocation.
\end{proof}

    Next, we show if the number of agents is more than $4$, then given an allocation in \stage\ D, we can construct a complete \efx\ allocation.

\begin{lemma}\label{lem: stageD>4}
	If an allocation $\X$ is in \stage\ D, and if the number of agents is greater than $4$, we can construct a complete \efx\ allocation.
\end{lemma}
\begin{proof}
	Note that $4< n = d+\ell +2$, so $d+\ell >2$. Since there is a duality between the allocation $\X$, and allocation $\hat{\X}$
	obtained after Update D, without loss of generality, we can assume $d\geq 2$, so $q_d \ne q_1$.
	We update the allocation $\X$ as follows:
	\begin{align*}
		&\text{Agents } p \text{ and } q_d \text{ drop their previous bundles},\\
		&X'_{p} \takes a_{p,q_d}\\
		&X'_{q_d} \takes \textsc{Choose}(q_d).
	\end{align*}
	Denote the resulting allocation by $\xp$. Note that if $q_d$ resents an agent $u$ after this update, then $a_{q_d,u}$ was not available when \textsc{Choose}$(q_d)$ was executed. Since only agent $p$ has a unit bundle from $q_d$, we must have $u=p$. Hence, $q_d$ only resents $p$.
	Moreover $q_1,\dots, q_{d-1}$ are now non resented since $q_i$'s were sorted in an increasing order of $a_{p,q_i}$.
	To distinguish $q_d$ from the other agents, we rename her as $v$.
	
	Hence, the only resent edges in the resent graph are $v\to p \to q$.
	We also denote non resented vertices $\{q_1,\dots,q_{d-1}\}$ by $R_p$, and 
	non resented vertices $\{r_1,\dots,r_{\ell}\}$ by $R_q$. 
	Note that here our allocation is \pzero, but not \pone. Hence, in this case, we are not following our global properties before the dumping 
	phase.
	
	\paragraph{Dumping Phase.}
	We allocate every remaining unallocated unit bundle by the following rules.
	
	\begin{enumerate}[label=(\roman*), leftmargin=2.2em]
		\item \textbf{Among $v$, $p$ and $q$.}
		We have \dumpp{a_{p,v}}{p} and we assign \dumpp{b_{p,v}}{v}.
		We have \dumpp{a_{q,p}}{q} and we assign \dumpp{b_{q,p}}{v}.
		We assign \dumpp{a_{v,q}}{v} and we assign \dumpp{b_{v,q}}{q_1}.
		(Note that since $d\geq 2$, we have $v=q_d \ne q_1$.)
		
		\item \textbf{Between non-resented $u\ne v$ and $p$,$q$.}
		Let $u\in R_p \cup R_q$.
		We let \dumpp{a_{u,p}}{u}, and we assign \dumpp{b_{u,p}}{v}.
		We let \dumpp{a_{u,q}}{u} and we assign \dumpp{b_{u,q}}{v}. 
		
		\item \textbf{Between $v$ and non-resented $u \in R_p \cup R_q$.}
		We assign \dumpp{a_{v,u}}{v} and \dumpp{b_{v,u}}{u}.
		(Note that this is possible since for every $u\in R_p$, we had $X_u=a_{u,p}$, and for every $u\in R_q$, we had $X_u=a_{u,q}$.)
		
		\item \textbf{Between two non-resented $u, u' \in R_p \cup R_q$.}
		If at least one of the unit bundles is already allocated to one of the roots, allocate the other unit bundle to the other root. 		
		Otherwise, allocate one unit bundle in $E_{u,u'}$ to $u$ and the other to $u'$, arbitrarily.				
	\end{enumerate}
	
	\paragraph{Correctness of Dumping Phase.}
	Denote the allocation right before dumping phase  by $\xp$, and denote the allocation at the end by $\xz$.
	Next, we prove that the final allocation, $\xz$, is \efx.
	
	First, note that no agent loses any good that they previously possessed, so for every agent $u$, we have $X''_u \greatereqval{u} X'_u$.
	Since $\xp$ is \pzero, \cref{obs:basicEFX} implies that $\xp$ is \efx. Agents $p$ and $q$ do not get any new good, so since they were not strongly envied, they are not still strongly envied.

	\textbf{Agent $p$ does not strongly envy anyone.}
	We have that $q$ is not strongly envied by anyone, so we show that $p$ does not envy any other agent.
	For any $u \notin \{p,q,v\}$, we have $X''_{u} \cap E_p = a_{u,p}\lowereqval{p} a_{p,q_d} = X''_p$, which comes from the fact that
	after $a_{p,q}$, unit bundle $a_{p,q_d}$ is agent $p$'s favorite unit bundle.
	Next, we show that $p$ does not envy agent $v$.
	\begin{align*}
		X''_v \cap E_{p} &= (X''_v \cap E_{p,v}) \sumsq (X''_v \cap E_{p,q}) \sumsq \bigsumsq_{u\in R_p \cup R_q} (X''_v \cap E_{p,u}) \\
		&= b_{p,v} \sumsq b_{q,p} \sumsq \bigsumsq_{u \in R_p \cup R_q} b_{u,p} \lowereqval{p} X_p \sumsq A_p(\X)
		\lowereqval{p} a_{p,q_d} =X''_p,
	\end{align*}
	where the last inequality follows since $\X$ is not in \stage\ B and $p\to q_d$, so $X_p\sumsq A_p(\X)\lowerval{p}a_{p,q_d}$.
	
	\textbf{Agent $q$ does not strongly envy anyone.}
	We have that $p$ is not strongly envied by anyone, so we show that $q$ does not envy any other agent.
	For any $u \notin \{p,q,v,q_1\}$, we have $X''_{u} \cap E_q = a_{u,q}\lowereqval{q} a_{q,p} = X''_q$, which comes from the fact that
	$a_{q,p}$ is agent $q$'s favorite unit bundle.
	Next, we show that $q$ does not envy agents $v$ and $q_1$.
	\begin{align*}
		X''_{v} \cap E_{q} &= (X''_{v} \cap E_{q,v}) \sumsq (X''_{v} \cap E_{q,p}) \sumsq 
		\bigsumsq_{u'\in R_p \cup R_q} (X''_{v} \cap E_{q,u'}) \\
		&= a_{v,q} \sumsq b_{q,p}  \sumsq \bigsumsq_{u'\in (R_p \cup R_q)\setminus \{q_1\}} b_{u',q} 
		\lowereqval{q} b_{q,p} \sumsq A_q(\X) 
		\lowereqval{q} a_{q,p}=X''_q, \\
		X''_{q_1} \cap E_{q} &= (X''_{q_1} \cap E_{q,v}) \sumsq (X''_{q_1} \cap E_{q,p}) \sumsq 
		\bigsumsq_{u'\in R_p \cup R_q} (X''_{q_1} \cap E_{q,u'}) \\
		&= b_{v,q} \sumsq \emptyset \sumsq a_{q_1,q} \sumsq \bigsumsq_{u'\in (R_p \cup R_q)\setminus \{q_1\}}\emptyset 
		\lowereqval{q} A_q(\X) 
		\lowereqval{q} a_{q,p}=X''_q,
	\end{align*}
	where  $b_{q,p} \sumsq A_q(\X) \lowereqval{q} a_{q,p}$ comes from \cref{claim:D2}.

	\textbf{A non-resented $u$ does not envy another non-resented agent $u'$.}
	First, we show that $v$ does not envy any other non-resented $u'$. 
	For every $u \in R_p \cup R_q$, we have \dumpp{a_{v,u}}{v} and \dumpp{b_{v,u}}{u}. 
	Also, \dumpp{a_{p,v}}{p}, \dumpp{b_{p,v}}{v}, \dumpp{a_{v,q}}{v}, and \dumpp{b_{v,q}}{q_1}. Hence,
	\begin{align*}
		X''_{u'} \cap E_v 
		&=  (X''_{u'} \cap E_{v,p}) \sumsq (X''_{u'} \cap E_{v,q}) \sumsq \bigsumsq_{u''\in R_p \cup R_q} (X''_{u'} \cap E_{v,u''}) \\
		&\subseteq \emptyset \sumsq b_{v,q} \sumsq \bigsumsq_{u''\in R_p \cup R_q} b_{v,u''}  \lowereqval{v} 
		a_{v,q} \sumsq \bigsumsq_{u''\in R_p \cup R_q} a_{v,u''}  \\ 
		&=  (X''_v \cap E_{v,q}) \sumsq \bigsumsq_{u''\in R_p \cup R_q} (X''_v \cap E_{v,u''}) \lowereqval{v} X''_v.
	\end{align*}
	Therefore, $v$ does not envy any other agent. 
	Next, assume $u \in R_p \cup R_q$. 
	If $u' \in R_p \cup R_q\setminus\{u\}$, then $X''_{u'}\cap E_u$ is a unit bundle that was unallocated in $\X$, so $u$ does not envy that. 
	
	All that remains to show is that no non-resented agent $u$ envies $v = q_d$. Note that:
	$$X''_v \cap E_u = a_{v,u} \sumsq b_{u,p} \sumsq b_{u,q}.$$
	
	Agent $q_1$ does not envy $q_d = v$ because by \cref{lem:q1 aq>ap}, we get:
	$$b_{p, q_1} \sumsq a_{q_1, q_d} \sumsq b_{q,q_1} \lowereqval{q_1}
	b_{p, q_1} \sumsq a_{q_1, q_d} \sumsq a_{q_1,q} \lowereqval{q_1} a_{q_1,p} \sumsq b_{q_d, q_1} \sumsq b_{q,q_1}
	\lowereqval{q_1} a_{q_1,p} \sumsq b_{q_d, q_1} \sumsq a_{q_1, q}.$$  
	Hence,
	$$X''_{v} \lowereqval{q_1} X''_{v}\cap E_{q_1} \lowereqval{q_1} b_{p, q_1} \sumsq a_{q_1, q_d} \sumsq b_{q,q_1} 
	\lowereqval{q_1} a_{q_1,p} \sumsq b_{q_d, q_1} \sumsq a_{q_1, q}\lowereqval{q_1} X''_{q_1}.$$
	
	For every $1<i\le d$, agent $q_i$ does not envy $v=q_d$ because using \cref{claim:D2}, we have
	$$X''_{v} \lowereqval{q_i} X''_{v}\cap E_{q_i} \lowereqval{q_i} a_{v,q_i} \sumsq b_{q_i,p} \sumsq b_{q_i,q} \lowereqval{q_i} b_{p,q_i} \sumsq \bigsumsq_{u\neq q_i,p} a_{q_i,u} \lowereqval{q_i} a_{q_i,p} = X''_{q_i}.$$
	
	Agent $r_1$ does not envy $q_d$ because by \cref{lem:aq>ap}, we can assume:
	$$b_{q,r_1} \sumsq a_{r_1,p} \sumsq  \bigsumsq_{i\in [d]} a_{r_1,q_i}   \lowerval{r_1} 
	a_{r_1,q} \sumsq b_{p,r_1} \sumsq \bigsumsq_{i\in [d]} b_{q_i, r_1}.$$  
	Therefore:
	$$b_{p, r_1} \sumsq a_{r_1, q_d} \sumsq b_{q,r_1} \lowereqval{r_1}
	b_{q,r_1} \sumsq a_{r_1,p} \sumsq a_{r_1,q_d} \lowerval{r_1}  a_{r_1,q} \sumsq b_{p,r_1} \sumsq  b_{q_d, r_1}
	\lowereqval{r_1} a_{r_1,p} \sumsq b_{q_d, r_1} \sumsq a_{r_1, q}.$$  
	Hence,
	$$X''_{v} \lowereqval{r_1}X''_{v}\cap E_{r_1} \lowereqval{r_1}  b_{p, r_1} \sumsq a_{r_1, q_d} \sumsq b_{q,r_1} 
	\lowereqval{r_1} a_{r_1,p} \sumsq b_{q_d, r_1} \sumsq a_{r_1, q}\lowereqval{r_1} X''_{r_1}.$$

	For every $1<i\le\ell$, agent $r_i$ does not envy $q_d$ because using \Cref{claim:D2}, we have
	$$X''_{v} \lowereqval{r_i} X''_{v}\cap E_{r_i} \lowereqval{r_i} b_{p, r_i} \sumsq a_{r_i, q_d} \sumsq b_{q,r_i} 
	\lowereqval{r_i} b_{q,r_i} \sumsq \bigsumsq_{u\neq r_i,q} a_{r_i,u} \lowereqval{r_i} a_{r_i,q} = X''_{r_i}.$$
	This means we have reached a complete \efx\ allocation.	
\end{proof}

All that remains to show is the case where $n \leq 4 $.
Throughout the algorithm, this case is the only place where we need to let go of unit bundles.
For this, we define a \emph{Minimal Greater} set.

\begin{definition}
	Given an agent $i$ and two subsets of goods $S$ and $T$ with $S\greaterval{i} T$, we define the bundle 
	$\textsc{MinimalGreater}_i (S,T)$ as follows: As long as there exists $g \in S$ such that $S\setminus \{g\} \greaterval{i} T$, remove $g$ from $S$. Define 
    $\textsc{MinimalGreater}_i (S,T)$ as the final $S$.
\end{definition}

\begin{lemma}\label{lem:mg}
	Given an agent $i$ and two subsets of goods $S$ and $T$ with $S\greaterval{i} T$, agent $i$ does not strongly envy
	$\textsc{MinimalGreater}_i (S,T)$ with respect to bundle $T$. Also, $\textsc{MinimalGreater}_i (S,T)\greaterval{i} T$.
\end{lemma}
\begin{proof}
	By the definition of $Y=\textsc{MinimalGreater}_i (S,T)$, for every $g\in Y$, we have $Y\setminus \{g\} \lowereqval{i} T$.
	Moreover, in the process of computing $\textsc{MinimalGreater}_i (S,T)$,
	whenever we remove $g$ from $S$, we have that $S\setminus \{g\} \greaterval{i} T$, so at the end we should have $Y\greaterval{i} T$.
\end{proof}

\begin{lemma}\label{lem: stageD=4}
	If $n\leq 4$, and allocation $\X$ is in \stage\ D, then we can construct a complete \efx\ allocation.
\end{lemma}

\begin{proof}
	Note that by $d \geq 1$, $\ell \geq 1$, and $n\leq 4$, we get $n=4$. 
	Hence, we have four agents $p,q,v,r$ such that agent $p$ resents $q$ and $v$ in $\X$. Also, $X_r= a_{r,q}$.
	Here by $S\greaterval{i} \max(T,P)$, we mean both $S\greaterval{i} T$, and $S\greaterval{i} P$.
	
	\textbf{Case 1:} 
	$\mathbf{a_{q,p} \lowerval{q} E_{v,q}}$. 
	Then, define $Y= \textsc{MinimalGreater}_q(E_{v,q},a_{q,p})$. By \cref{lem:mg}, we get $Y \greaterval{q} a_{q,p}$.
	
	\textbf{Subcase 1.1:} $\mathbf{Y \lowereqval{v} a_{v,p} \sumsq (E_{v,q} \setminus Y) \sumsq a_{v,r}}$. In this case, set:
	\begin{align*}
		&X'_r \takes a_{r,q} \sumsq a_{r,p}\sumsq b_{v,r}\\
		&X'_p \takes a_{p,q} \sumsq b_{v,p} \sumsq b_{r,p} \\
		&X'_q \takes b_{p,q} \sumsq Y \sumsq b_{r,q} \\
		&X'_v \takes a_{v,p} \sumsq (E_{v,q} \setminus Y) \sumsq a_{v,r}.
	\end{align*}
    Note that since $Y\greatereqval{q} a_{q,p} \greatereqval{q} a_{q,v} \greatereqval{q} b_{q,v}$,
    and since $Y\sqcup (E_{q,v}\setminus Y)= a_{q,v} \sqcup b_{q,v}$,
    and the valuation function is cancelable, by \cref{cancelablevalprops},
    we get that $a_{q,v} \greatereqval{q} E_{q,v} \setminus Y $.
    Hence, $Y \greaterval{q} E_{q,v} \setminus Y$.
	
	In this case, all agents become non-resented.
	
	\noindent \textbf{Agent $r$ does not envy anyone.}
	This is because $v,p,q$ each hold a single unit bundle from $r$.
	Because $a_{r,q}$ was her most valuable unit bundle, she does not resent anyone.
	
	\noindent \textbf{Agent $q$ does not envy anyone.}
	Agent $q$ does not envy $p$ because by the definition of $Y$,
	$$ X'_q \greatereqval{q} Y \greatereqval{q} a_{q,p} = X'_p \cap E_q.$$
	Agent $q$ does not envy $v$ because
	$$ X'_q \greaterval{q} Y \greaterval{q} E_{q,v} \setminus Y = X'_v \cap E_q.$$
	Agent $q$ does not envy $r$ because
	$$ X'_q \greatereqval{q} a_{q,p} \greatereqval{q}
    a_{q,r} \greatereqval{q} a_{r,q}=X'_r \cap E_q. $$

	\noindent \textbf{Agent $p$ does not envy anyone.}
	This is once again because $v,q,r$ all hold a single unit bundle from $p$.
	Because $a_{p,q}$ was her most valuable unit bundle, she does not resent anyone.
	
	\noindent \textbf{Agent $v$ does not envy anyone.}
	We know that $v$ has $a_{v,p}$.
	Because this was her most valuable unit bundle,
	$a_{v,p} \greatereqval{v} b_{v,p}$ and $a_{v,p} \greatereqval{v} a_{r,v}$
	which means she does not envy $p$ and $r$.
	Finally, by the definition of our subcase, $X'_q\cap E_v=Y\lowereqval{v}X'_v$
	which means $v$ does not envy $q$.
	
	\textbf{Subcase 1.2:} $\mathbf{Y \greaterval{v} a_{v,p} \sumsq (E_{v,q} \setminus Y) \sumsq a_{v,r}}$. In this case, set:
	\begin{align*}
		&X'_r \takes a_{r,q} \sumsq b_{p,r}\sumsq a_{r,v} \sumsq b_{p,v}\\
		&X'_p \takes b_{q,p} \sumsq a_{p,v} \sumsq a_{p,r}\sumsq b_{r,q} \sumsq b_{r,v} \sumsq   (E_{q,v} \setminus Y) \\
		&X'_q \takes a_{q,p} \\
		&X'_v \takes Y.
	\end{align*}
	
	\noindent \textbf{No one strongly envies agent $v$.} We have that $X'_v = Y \subset E_{v,q}$, so agents $p$ and $r$ do not envy agent
	$v$. Moreover, by \cref{lem:mg}, we get agent $q$ does not strongly envy
	$Y=\textsc{MinimalGreater}_q (E_{v,q},a_{q,p})$ with respect to bundle $X'_q=a_{q,p}$.
	
	\noindent \textbf{No one strongly envies agent $q$.} We have that $X'_q = a_{q,p} \subseteq E_{q,p}$, 
	so agents $v$ and $r$ do not envy agent $q$. Moreover, agent $p$ does not strongly envy $X'_q = a_{q,p}$ while possessing $b_{q,p}$.
	
	\noindent \textbf{No one envies agent $r$.} We have:
	\begin{align*}
		X'_r \cap E_q &= a_{r,q} \lowereqval{q} a_{q,p} = X'_q, \\
		X'_r \cap E_p &= b_{p,r}\sumsq b_{p,v} \lowereqval{p}  a_{p,r}\sumsq a_{p,v}  \lowereqval{p} X'_p, \\
		X'_r \cap E_v &= a_{r,v} \sumsq b_{p,v} \lowereqval{v}  a_{v,p} \sumsq (E_{v,q} \setminus Y) \sumsq a_{v,r} \lowereqval{v} Y= X'_v,
	\end{align*}
	where the last inequality comes from subcase's condition.
	
	\noindent \textbf{No one envies agent $p$.} We have:
	\begin{align*}
		X'_p \cap E_q &=b_{q,p} \sumsq b_{r,q} \sumsq  (E_{q,v} \setminus Y) 
		\lowereqval{q} b_{q,p} \sumsq b_{r,q} \sumsq  a_{q,v} \lowereqval{q} a_{q,p} = X'_q, 
	\end{align*}
	where the first inequality comes from the fact that  	
	$Y \sqcup (E_{q,v} \setminus Y) = a_{q,v} \sqcup b_{q,v}$ and $a_{q,v}\lowereqval{q} a_{q,p} \lowereqval{q} Y$, so
	$(E_{q,v} \setminus Y) \lowereqval{q} b_{q,v} \lowereqval{q} a_{q,v}$. 
	Also, the last inequality comes from \cref{claim:D2}. Moreover,
	\begin{align*}
		X'_p \cap E_v &= a_{p,v} \sumsq b_{r,v} \sumsq   (E_{q,v} \setminus Y) 
		\lowereqval{v}  a_{v,p} \sumsq a_{v,r} \sumsq   (E_{q,v} \setminus Y) \lowereqval{v} Y= X'_v,
	\end{align*}
	where the last inequality comes from subcase's condition. In addition,
	\begin{align*}
		X'_p \cap E_r &=  a_{p,r} \sumsq b_{r,q} \sumsq b_{r,v}  \lowereqval{r} a_{r,p} \sumsq b_{q,r} \sumsq a_{r,v}  
		\lowereqval{r} b_{p,r} \sumsq a_{r,q} \sumsq b_{v,r}    \lowereqval{r} X'_r, 
	\end{align*}
	where the second inequality comes from \cref{lem:aq>ap}. Hence, $\xp$ is a complete \efx\ allocation in this case.
	
	\textbf{Case 2:} $\mathbf{a_{q,p} \greatereqval{q} E_{v,q}}$. We divide to three subcases.
	
	\textbf{Subcase 2.1:} $\mathbf{E_{v,q}\sumsq b_{p,v} \sumsq b_{r,v} \greaterval{v} a_{v,p}}$. In this case, set:
	\begin{align*}
		&X'_r \takes a_{r,q} \sumsq a_{r,v}\sumsq b_{p,r}\\
		&X'_p \takes b_{q,p} \sumsq a_{p,v} \sumsq a_{p,r} \sumsq b_{r,q} \\
		&X'_q \takes a_{q,p}  \\
		&X'_v \takes E_{v,q}\sumsq b_{p,v} \sumsq b_{r,v}.
	\end{align*}

	\noindent \textbf{No one envies agent $v$.} We have:
	\begin{align*}
		X'_v \cap E_r &=  b_{r,v} \lowereqval{r} a_{r,q} \lowereqval{r} X'_r, \\
		X'_v \cap E_p &= b_{p,v} \lowereqval{p}  a_{p,v} \lowereqval{p} X'_p, \\
		X'_v \cap E_q &= E_{v,q} \lowereqval{q}  a_{q,p} = X'_q,
	\end{align*}
	
	\noindent \textbf{No one strongly envies agent $q$.} We have that $X'_q = a_{q,p} \subseteq E_{q,p}$, 
	so agents $v$ and $r$ do not envy agent $q$. Moreover, agent $p$ does not strongly envy $X'_q = a_{q,p}$ while possessing $b_{q,p}$.

	\noindent \textbf{No one envies agent $r$.} For every agent $t\ne r$, we have $X'_r \cap E_t$ is a unit bundle that was unallocated
	in $\X$, and since every agent is getting a (weakly) more valuable bundle, nobody envies $r$ as they did not in $\X$.
	
	\noindent \textbf{No one envies agent $p$.} We have:
	\begin{align*}
		X'_p \cap E_q &= b_{q,p} \sumsq b_{r,q} \lowereqval{q} a_{q,p} = X'_q, 
	\end{align*}
	where the last inequality comes from \cref{claim:D2}. Moreover,
	\begin{align*}
		X'_p \cap E_v &=a_{p,v}
		\lowereqval{v}  a_{v,p} \lowereqval{v} E_{v,q}\sumsq b_{p,v} \sumsq b_{r,v}= X'_v,
	\end{align*}
	where the last inequality comes from subcase's condition. In addition,
	\begin{align*}
		X'_p \cap E_r &=  a_{p,r} \sumsq b_{r,q}
		\lowereqval{r} b_{p,r} \sumsq a_{r,q} \sumsq b_{v,r}    \lowereqval{r} X'_r, 
	\end{align*}
	where the first inequality comes from \cref{lem:aq>ap}. Hence, $\xp$ is a complete \efx\ allocation in this case.

	\textbf{Subcase 2.2:} $\mathbf{E_{v,q}\sumsq b_{p,v} \sumsq b_{r,v} \lowereqval{v} a_{v,p}}$
	and $\mathbf{E_{v,q}\sumsq b_{p,q} \sumsq b_{r,q} \greaterval{q} a_{q,p}}$. In this case, set:
	\begin{align*}
		&X'_q \takes E_{v,q}\sumsq b_{p,q} \sumsq b_{r,q}  \\
		&X'_p \takes a_{p,q}.
	\end{align*}
	Then, in the resulting (incomplete) allocation $\xp$, agent $q$ does not resent anyone by the second subcase condition, agent $v$ does not resent anyone by the first subcase condition, and agents $p$ and $r$ do not resent anyone by their allocated unit bundles.
	So by allocating unallocated unit bundles to one of their incident agents as follows, we get a complete \efx\ allocation. 
	\begin{align*}
		&X''_r \takes a_{r,q} \sumsq a_{r,v}\sumsq b_{p,r}\\
		&X''_p \takes a_{p,q} \sumsq b_{v,p} \sumsq a_{p,r} \\
		&X''_q \takes E_{v,q}\sumsq b_{p,q} \sumsq b_{r,q}  \\
		&X''_v \takes a_{v,p} \sumsq b_{r,v}.
	\end{align*}
	Agent $q$ does not envy $p$ by the second subcase condition, does not envy $v$ because $X''_v\cap E_q=E_{v,q}\lowereqval{q}X''_q$, and does not envy $r$ since $X''_r\cap E_q=a_{r,q}\lowereqval{q}a_{q,p}\lowereqval{q}X''_q$.
	Agent $v$ does not envy $q$ by the first subcase condition, and does not envy $p$ or $r$ since $a_{v,p}\subseteq X''_v$.
	The checks for agents $p$ and $r$ are the same as in Subcase 2.1.

	\textbf{Subcase 2.3:} $\mathbf{E_{v,q}\sumsq b_{p,v} \sumsq b_{r,v} \lowereqval{v} a_{v,p}}$
	and $\mathbf{E_{v,q}\sumsq b_{p,q} \sumsq b_{r,q} \lowereqval{q} a_{q,p}}$. In this case, set:
	\begin{align*}
		&X'_r \takes a_{r,q} \sumsq a_{r,v}\sumsq b_{p,r}\\
		&X'_p \takes b_{q,p} \sumsq b_{v,p} \sumsq a_{p,r} \sumsq E_{v,q} \sumsq b_{r,v} \sumsq b_{r,q} \\
		&X'_q \takes a_{q,p} \\
		&X'_v \takes a_{v,p}.
	\end{align*}
	
	\noindent \textbf{No one strongly envies agent $v$.} We have that $X'_v = a_{v,p}\subseteq E_{v,p}$, 
	so agents $q$ and $r$ do not envy agent $v$. Moreover, agent $p$ does not strongly envy
	$a_{v,p}$ with respect to bundle $b_{v,p}$.
	
	\noindent \textbf{No one strongly envies agent $q$.} We have that $X'_q = a_{q,p}\subseteq E_{p,q}$, 
	so agents $v$ and $r$ do not envy agent $q$. Moreover, agent $p$ does not strongly envy
	$a_{q,p}$ with respect to bundle $b_{q,p}$.
	
	\noindent \textbf{No one envies agent $r$.} We have:
	\begin{align*}
		X'_r \cap E_q &= a_{r,q} \lowereqval{q} a_{q,p} = X'_q, \\
		X'_r \cap E_p &= b_{p,r} \lowereqval{p} a_{p,r} \lowereqval{p} X'_p, \\
		X'_r \cap E_v &= a_{r,v}  \lowereqval{v} a_{v,p}= X'_v.
	\end{align*}
	
	\noindent \textbf{No one envies agent $p$.} We have:
	\begin{align*}
		X'_p \cap E_q &= b_{q,p} \sumsq E_{v,q} \sumsq b_{r,q} \lowereqval{q}  a_{q,p} = X'_q, 
	\end{align*}
	where the inequality comes from the subcase's condition. Moreover,
	\begin{align*}
		X'_p \cap E_v &= b_{v,p} \sumsq E_{v,q} \sumsq b_{r,v} \lowereqval{v}  a_{v,p} = X'_v,
	\end{align*}
	where the last inequality comes from subcase's condition. In addition,
	\begin{align*}
		X'_p \cap E_r &=  a_{p,r} \sumsq b_{r,v} \sumsq b_{r,q} \lowereqval{r} b_{p,r} \sumsq a_{r,q} \sumsq b_{v,r}  
		\lowereqval{r} X'_r, 
	\end{align*}
	where the first inequality comes from \cref{lem:aq>ap}. Hence, $\xp$ is a complete \efx\ allocation in this case.
	Thus, the proof is complete.
\end{proof}

By \cref{lem: stageD>4} and \cref{lem: stageD=4}, given an allocation in \stage\ D, we get an \efx\ allocation. 
\begin{corollary}\label{lem: stageD}
	If an allocation is in \stage\ D, we can construct a complete \efx\ allocation in polynomial time.
\end{corollary}

\section{\Stage\ E}\label{sec:e}
	In this section, we prove that whenever allocation is in \stage\ E, we can construct a complete $\efx$ allocation.

\begin{lemma}\label{lem: stageE}
	If an allocation is in \stage\ E, we can construct a complete \efx\ allocation in polynomial time.
\end{lemma}
\begin{proof}
	Since the allocation is in \stage\ E,  there exist two agents $s, t$ such that $s\to t$ and
        $$A_s(\X) \sumsq b_{s,t} \greaterval{t} X_t .$$
	First, note that since allocation is \pone, and agent $t$ is resented,
	agent $t$ does not resent anyone, so $X_t=a_{t,s}$ is agent $t$'s favorite unit bundle.

	Then, update the allocation with update rule (U1), defined in \cref{up rule:u1}, for the pair $(s, t)$.
	By \cref{lem:cycle_support}, the allocation remains \pone, and $(s,t)$ becomes a support pair.
	Moreover, agent $t$ does not resent anyone.

    Since the original allocation was not in \stage\ D, there existed some other resent tree, and it still exists since their bundle has not changed.
    Hence, by \cref{lem:weak_support}, the proof is complete.
\end{proof}

\section{\Stage\ F}\label{sec:f}

	In this section, we prove that whenever an allocation is in \stage\ F, we can construct a complete $\efx$ allocation.
	We first prove a couple of auxiliary lemmas, and then, we introduce the main lemma of this section, which proves the existence of $\efx$ allocations for any allocation in \stage\ F.

\begin{lemma}\label{lem:support_triple}
    If we have a \pone\ allocation \X\ with three agent $s,t,k$ such that
    \begin{itemize}
        \item $(s,t)$ is a support pair.
        \item $s$ and $t$ do not resent any agent.
        \item $k$ is not resented.
        \item If $k$ possesses $a_{k,j}$, then
                $(A_j(\X) \setminus E_{j,k}) \sumsq a_{k,j} \lowereqval{j} X_j$
    \end{itemize}
    Then, we can construct a complete \efx\ allocation in polynomial time.
\end{lemma}

\begin{proof}
    First, note that if there exists a $p \to q$ such that $A_q(\X) \cup b_{p, q} \greaterval{q} X_q$, we update the allocation using (U1), defined in \cref{up rule:u1}. Note that since agents $s$ and $t$ are non-resented and also do not resent any other agent, we have that $p, q \notin \{s, t\}$. Therefore, by \cref{lem:cycle_support}, $(p, q)$ will become a support pair that is disjoint from $(s, t)$, then, by \cref{lem:two_support}, we would get a full $\efx$ allocation. Thus, without loss of generality, we assume that for any $p \to q$, we have that $A_q(\X) \cup b_{p, q} \lowerval{q} X_q$. Moreover, without loss of generality, by our assumption at the end of \stage\ A, for every $k \to i$ and $\ell \to j$, the function $R(i, j)$ is well-defined using \cref{lem:strong_A}.

    Next, we proceed to the dumping phase.

    \paragraph{Dumping Phase.}
	We allocate every remaining unallocated unit bundle by the following rules.
	Here, a \emph{root} means a non-resented agent, and for a root $p$, the set $R_p$ denotes the agents resented by $p$ in $\X$. Every unallocated good gets allocated based on one of the following rules:

    \begin{enumerate}[label=(\roman*), leftmargin=2.2em]
        \item \textbf{Between $s$ and other roots.} Since $(s,t)$ is a support pair, we have $X_s \subseteq E_{s, t}$.
        We allocate the other unit bundle in $E_{s, t}$ to $t$. Moreover, let $p \notin \{s, t\}$ be a root. We assign \dumpu{a_{p,s}}{p} and \dumpp{b_{p,s}}{s}.

        \item \textbf{Between $k$ and other roots.} Let $p\ne s$ be a root. 
        Then, if $a_{k,p}\subseteq X_k$, then \dumpp{b_{k,p}}{p}.
        Otherwise, \dumpu{a_{p,k}}{p} and \dumpp{b_{p,k}}{k}.

		\item \textbf{Root to root (except $k$ and $s$).} 
        Now, let $p$ and $r$ be two distinct roots such that $p,r \notin \{s,k\}$.
        Then, if at least one of the unit bundles in $E_{p, r}$ is already allocated to one of the roots, allocate the other unit bundle to the other root. 		
		Otherwise, arbitrarily allocate one unit bundle in $E_{p,r}$ to $p$ and the other to $r$. 

		\item \textbf{Root to its own children.} For every root $p$ and every $u\in R_p$, assign \dumpp{b_{u,p}}{p}. Also, we already have $X_u = a_{u, p}$.
		
		\item \textbf{Root to other children.} Let $p$ and $r$ be two distinct roots and let $u\in R_r$.
        If $p \notin \{s,t\}$, we assign \dumpu{a_{p,u}}{p} and \dumpp{b_{p,u}}{s}.
        When $p$ is either $s$ or $t$, 
        we instead assign \dumpu{a_{p,u}}{p} and \dumpp{b_{p,u}}{k}.
		
		\item \textbf{Between two children of the same tree.} Let $p$ be a root and let $u,v\in R_p$ be two distinct children of $p$.
        Since $s,t$ do not resent anyone, we have $p \notin \{s,t\}$.
        We assign \dumpp{a_{u,v}}{p} and \dumpp{b_{u,v}}{s}.
		
		\item \textbf{Between two children of two different trees.}         Let $p$ and $r$ be two distinct roots, and let $u\in R_p$ and $v\in R_r$. Note that we have $p,r\notin \{s,t\}$.
        Assign \dumpp{a_{u,v}}{s} and \dumpp{b_{u,v}}{R(i,j)}.

	\end{enumerate}
	
	This completes the allocation.

    \paragraph{Correctness of Dumping Phase.}Denote the allocation right before dumping phase  by $\X$, and denote the allocation at the end by $\xp$. 
	Next, we prove that the final allocation, $\xp$, is \efx.

    Note that by construction, it is clear that $\X$ follows global properties before the dumping phase (\cref{pbd}), and our dumping rules follow our general dumping structure. Thus, by \cref{lem:dumprule},
	\begin{enumerate}
		\item For every agent $p$, we have $X'_p \greatereqval{p} X_p$.
		\item Nobody strongly envies any agent who was resented in $\X$.
		\item For every resent $p\to q$ in $\X$, agent $q$ does not envy agent $p$ in $\xp$.
		\item No one envies agent $s$, and agent $s$ does not strongly envy anyone.  
	\end{enumerate}

    To prove that the final allocation is $\efx$, we show that any agent $p \ne s$ who was non-resented in $\X$, cannot be envied in $\xp$. To do so, we prove the following. Note that by the fourth statement of \cref{lem:dumprule}, we do not need to consider envies from or toward agent $s$. 

    \begin{itemize}
    \item \textbf{Agent $k$ is not strongly envied by any agent in $\xp$.} 
    Note that we know that $s$ does not strongly anyone.
    Consider an arbitrary agent $q\ne s$. We show that $q$ does not envy $k$ in $\xp$.
    
    If agent $k$ possessed $a_{k,q}$ in $\X$, then by Lemma's condtions, we get 
    \begin{align*}
        X'_q &\greatereqval{q} X_q \\
        &\greatereqval{q} a_{k,q} \cup (A_q(\X) \setminus E_{q,k})
        = a_{k,q} \cup \bigcup_{r\notin \{q,k\}} (A_q(\X) \cap E_{q,r})  \\ 
        &\greatereqval{q} (X'_k\cap E_{k,q}) \cup \bigcup_{r\notin \{q,k\}} (X'_k\cap E_{r,q}) = X'_k \cap E_q \greatereqval{q} X'_k
    \end{align*}
    Hence, assume that $k$ did not possess $a_{k,q}$ in $\X$.
    First, let $q$ be a resented agent. By the third statement of \cref{lem:dumprule}, we only need to consider the case where $i \to q$ such that $i \neq k$. 
    Since $k$ did not possess $a_{k,q}$ in $\X$, $a_{q,k}$ was unallocated in $\X$.
    This means that $a_{q,k} = A_q (\X)\cap E_{q,k}$.
    In allocation \xp, $k$ drops $b_{q,k}$ if she previously had it and picks up $a_{k,q}$, so
    $ X'_k \cap E_{k,q} = a_{k,q} \lowereqval{q} a_{q,k} = A_q (\X)\cap E_{q,k}$.
    For every resented agent $r\ne q$, we have 
    $A_q(\X) \cap E_{q,r}= a_{q,r} \greatereqval{q} X'_k \cap E_{q,r}$.
    Moreover, for every non-resented $r\ne k$, we have 
    $X'_k\cap E_{q,r} \lowereqval{q} b_{r,q} \lowereqval{q} A_q(\X)\cap E_{r,q}$.
    Hence, for every agent $r\ne q$, we have 
    $X'_k\cap E_{q,r} \lowereqval{q} A_q(\X)\cap E_{r,q}$. Hence,
\begin{align*}
    X'_q &= X_q  \\
         &\greatereqval{q} A_q(\X)=\bigcup_{r\ne q} (A_q(\X)\cap E_{r,q}) \\ 
         &\greatereqval{q}  \bigcup_{r\ne q} (X'_k\cap E_{r,q}) = X'_k \cap E_q \greatereqval{q}
    X'_k,
\end{align*}
    which means $k$ is not envied by any resented $q$.
    

    Finally, suppose $q$ is non-resented, and since $a_{k,q}$ is not in $X_k$,
    according to our dumping rules, $a_{q,k}$ is allocated to $q$ in $\xp$.
	Also, if for some $r\notin \{k,q\}$, a unit bundle in $E_{q,r}$ is allocated to $k$, then $r$ is resented,
	and $a_{q,r}$ is allocated to $q$. Hence,
    $$X'_q \greatereqval{q} a_{q,k} \sumsq \bigcup_{r \in R(\X)} a_{q,r} 
			\greatereqval{q} b_{q,k} \sumsq \bigcup_{r \in R(\X)} b_{q,r} \greatereqval{q} X'_k$$
            
    \item \textbf{Agent $t$ is not envied by any agent in $\xp$.}   
    According to our dumping rules, $t$ is only allocated unit bundles incident to her and only receives one unit bundle from each agent, 
    so for any agent $i$, we get $X'_t \lowereqval{i} X'_t \cap E_{i,t} \lowereqval{i} a_{i,t} \lowereqval{i} X_i \lowereqval{i} X'_i$, where the last inequality comes from the fact that 
    $i$ did not resent $t$ in $\X$.
    

    \item \textbf{Any non-resented agent $p \neq \{s, t, k\}$ in $\X$ is not envied by any agent in $\xp$.} For a resented agent $v$, by \cref{lem:dumprule}, we only need to consider the case where $r\to v$ for $r\ne p$.
	If there exists a resent edge $p\to u$ such that $R(u,v)=p$, then 
	$$X'_p \cap E_v \subseteq a_{p,v} \sumsq   \bigsumsq_{q \in R_p} (X'_p \cap E_{v,q})
	\lowereqval{v}  a_{v,p} \sumsq D_v(\X)  \lowereqval{v} X_v = X'_v,$$
	where the last inequality comes from the fact  $R(u,v)=p$.
	Also, if for every $p\to u$, we have $R(u,v)=r$, then 
	$$X'_p \cap E_v =a_{p,v} \lowereqval{v} X'_v.$$

    Now consider a non-resented agent $v$.
	The only unit bundle incident to $v$ that can be allocated to $p$ during the dumping phase
	is the unit bundle from $E_{v,p}$ assigned by Rule~(iii).
	Therefore, $X'_p \lowereqval{v} X'_p \cap E_{v,p} \lowereqval{v} a_{v,p} \lowereqval{v} X_v \lowereqval{v} X'_v$, where the last inequality comes from the fact that 
    $v$ did not resent $p$ in $\X$.
	Hence, $v$ does not envy $p$ in $\xp$.
    \end{itemize}
	Hence, the final allocation is a complete \efx\ allocation.
\end{proof}

We define nine properties (denoted by \textsc{FP$i$} for $i \in [9]$) for an allocation, which are desirable for an allocation that is initially in \stage\ F. Consider an allocation $\X$ and agents $k, u, v$. We want the allocation to satisfy the following properties:

\begin{itemize}
    \item \textsc{FP$1$}: $\X$ is a \pone\ allocation.
    \item \textsc{FP$2$}: $k$ is non-resented and only $k$ may have multiple unit bundles.
    \item \textsc{FP$3$}: For any agent $p\ne u$, agent $k$ does not possess unit bundle $a_{k,p}$.
    \item \textsc{FP$4$}: $v$ \weakresents\ $u$, and $a_{v,u}$ is a most valuable unit bundle for agent $v$.
    \item \textsc{FP$5$}: $X_u \greatereqval{u} (A_u(\X)\setminus (E_{u,v}\sumsq E_{u,k})) \sumsq b_{v,u} \sumsq a_{u,k}$.
    \item \textsc{FP$6$}: For all $k \to r$, we have $X_r \greatereqval{r} A_r(\X) \sumsq b_{k,r}$
    \item \textsc{FP$7$}: For all $j\to i$ with $j\ne k$, we have either $a_{j,k} \sumsq D_k (\X)\lowereqval{k} X_k$ or agent $k$ possesses $a_{k,j}$. 
    \item \textsc{FP$8$}: For all $k\to r$ and $j$ \weakresents\ $i  $, we have $$a_{j,r} \cup D_r(\X) \lowereqval{r} X_r.$$
    \item \textsc{FP$9$}: For all $p\to q$, we have $A_q(\X) \sumsq b_{p,q} \lowereqval{q} X_q$.
\end{itemize}

We next introduce two update rules, which we may use in our proof:

\begin{definition}[Update Rule F1\label{up rule:f1}] 
    Given a \pone\ allocation $\X$ and agents $i, j, k$ such that $j \to i$, $k$ is non-resented, $a_{j, k} \sumsq D_k(\X) \greaterval{k} X_k$,  agent $k$ does not possess $a_{k, j}$, and only agent $k$ may hold more than one unit bundle, let $i$ be the most-resented agent by $j$, which is possible since the inequality above is independent of agent $i$. Then, let: (Use $\hat{\X}$ to denote the allocation right before executing \textsc{Reduce Trees}.)
    \begin{align*}
    &\text{Agents } k,j, \text{and } i \text{ drop their previous bundles}\\
    &\hat{X}_k \takes a_{k,j}  \sumsq (B_k\setminus E_{k,j}) \\
    &\forall r \notin\{j,i,k\}  \text{ with } X_r \subseteq E_{r,k}: \hat{X}_r = a_{r,k}\\
    &\hat{X}_j \takes a_{j,i}\\
    &\hat{X}_i \takes \textsc{choose}(i)\\
	&Run \text{ Reduce Trees}(\hat{\X}, i).
\end{align*}

\end{definition}

\begin{lemma}\label{lem:f_aux1}

    If we execute update (F1) on an allocation $\X$ that is \pone, agent $k$ is non-resented and only $k$ may hold more than one unit bundle,
	and the following property holds:
    \begin{itemize}
        \item For every resent edge $p\to q$ where $p\ne k$, and $q$ is a most-resented agent by $p$, we have $a_{p,q} \greatereqval{p} A_p(\X) \sumsq X_p,$
    \end{itemize}
    then, the resulting allocation satisfies all properties \textsc{FP$1$-FP$5$} with respect to $(k,u=j,v=i)$.
\end{lemma}

\begin{proof}
    Let $\X$ and $\xp$ denote the allocation before and after the update.
    Set $H= R_k \cup \{k,j\}$, $p=j$, and $h_0=i$. Note that before the execution of 
    \textsc{Reduce Trees}$(\hat{\X},i)$, the allocation is \pzero. We first show that the conditions of \cref{lem:treebreaker2} hold.
	Since $i$ was the most-resented agent by $j$, and now agent $j$ possesses $a_{j,i}$, agent $p=j$ does not resent anyone. 
	Also, $h_0=i$ is not resented, since she is getting an unallocated unit bundle that nobody resents since the allocation was initially \pzero. Hence, all of the conditions of \cref{lem:treebreaker2} hold. Therefore, after execution of \textsc{Reduce Trees}$(\hat{\X},i)$, we get a \pone\ allocation, i.e., \textsc{FP$1$} holds.
	Also, the bundles of agents in $H$ do not change. Since agent $k$ was the only agent that may have had multiple unit bundles, and during the execution of \textsc{Reduce Trees}$(\hat{X},i)$, agents only receive a single unit bundle, $k$ is the only agent who may have more than one unit bundle in $\xp$. Moreover, $k\in H$, and $k$ was non-resented in $\hat{X}$, so \textsc{FP$2$} holds. 
	We have $X'_k = \hat{X}_k = a_{k,j}  \sumsq (B_k\setminus E_{k,j}) $, which satisfies \textsc{FP$3$}. Because $v=i$ was resented in $\X$,
    $a_{v,u}$ is most valuable unit bundle for agent $v$. Moreover, $v=i=h_0$ is non-resented in $\hat{\X}$, $X'_u =\hat{X_u}=a_{u,v}$, and $u=j$ does not resent anyone in $\xp$ since she
    has her most valuable unit bundle, so $v$ \weakresents\ $u$, satisfying \textsc{FP$4$}.
	
	Finally, by the assumption of the lemma, we should have $X'_u = a_{u,v} \greatereqval{u} A_u(\X) \sumsq X_u$.
	Moreover, since (F1) was applicable, we have that agent $k$ did not possess $a_{k,u}$ in $\X$, so 
	$a_{u,k} \subseteq  A_u(\X) \sumsq X_u$. 
	Also, since $X_v  =a_{v,u}$, we get $b_{v,u} \subseteq  A_u(\X) \sumsq X_u$. 
	Hence, we get 
	$A_u(\X)
    \cup X_u = a_{u, k} \cup b_{v, u} \cup (A_u(\X) \cup X_u) \setminus (E_{u, v} \cup E_{u, k})$. 
    Then, we have 
	$(A_u(\X) \cup X_u) \setminus (E_{u, v} \cup E_{u, k}) = A_u(\hat{\X}) \setminus (E_{u, v} \cup E_{u, k})$. Moreover, by  \cref{lem:treebreaker2}, during the execution of \textsc{Reduce Trees}$(\hat{\X},i)$, 
	for every $t\ne h_0=v$, no allocated unit bundle in $E_{u,t}$ gets unallocated, so any unit bundle currently in 	
	$A_u(\X')$ must have been unallocated right before executing \textsc{Reduce Trees}$(\hat{\X},i)$, 
	so $A_u(\hat{\X}) \setminus (E_{u, v} \cup E_{u, k}) \greatereqval{u} A_u(\X') \setminus (E_{u, v} \cup E_{u, k})$.
	 Combining these facts, we get
\begin{align*}
		X'_u = a_{u, v} &\greatereqval{u} A_u(\X) \cup X_u \\
		&= (A_u(\X) \cup X_u) \setminus (E_{u, v} \cup E_{u, k}) \cup a_{u, k} \cup b_{v, u} \\
		&= A_u(\hat{\X}) \setminus (E_{u, v} \cup E_{u, k})  \cup a_{u, k} \cup b_{v, u} \\
		&\greatereqval{u} A_u(\X') \setminus (E_{u, v} \cup E_{u, k}) \cup a_{u, k} \cup b_{v, u},
\end{align*}
meaning that \textsc{FP$5$} holds. Hence, the proof of claim is complete.    
\end{proof}

\begin{lemma}\label{lem:f_aux2}
    Consider an allocation $\X$ that satisfies properties \textsc{FP$1$-FP$5$}. 
	If there exists $p\to q$ where $p\ne k$, and $q$ is the most-resented agent by $p$ and $a_{p,q} \lowerval{p} A_p(\X) \sumsq X_p$, 
	we can construct a complete $\efx$ allocation in polynomial time.
\end{lemma}

\begin{proof}
	We update the allocation as follows: (let $\X$ and $\xp$ denote the allocation before and after the update.)
\begin{align*}
    X_p' &\takes A_p(\X) \sumsq X_p
\end{align*}
    Next, we show that the conditions of \cref{lem:support_triple} are satisfied with $(k,s=q,t=p)$.
	Since $ a_{p,q} \lowerval{p} A_p \sumsq X_p$, agent $q$ is no longer resented.
	Hence, $(q,p)$ is a support pair.
	Note that since $q$ was the most-resented agent by $p$ and $a_{p,q} \lowerval{p} A_p(\X) \sumsq X_p$, 
	we get that agent $p$ does not resent any agent in $\xp$. 
	Moreover, agent $q$ did not resent any agent in $\X$, and $X'_q=X_q$, 
    so agent $q$ does not resent any agent in $\xp$.  Also, since no new resent relation has been created, the allocation remains \pone, and since $k$ was non-resented in $\X$, she remains 
    non-resented.
    
	Agent $k$ can only have $a_{k,j}$ for $j = u$,
    and according to \textsc{FP$5$}, we had 
    $X_u \greatereqval{u}  (A_u(\X)\setminus (E_{u,v}\sumsq E_{u,k})) \sumsq b_{v,u} \sumsq a_{u,k}$.
    Moreover, no allocated unit bundle gets unallocated after the update. Henec,
    $$X'_u \greatereqval{u} X_u \greatereqval{u} (A_u(\X)\setminus (E_{u,v}\sumsq E_{u,k})) \sumsq b_{v,u} \sumsq a_{u,k} \greatereqval{u} (A_u(\xp)\setminus E_{u,k}) \sumsq a_{u,k}.$$
    This means that we have satisfied the conditions of \cref{lem:support_triple},
	and we can reach a complete $\efx$ allocation.
\end{proof}

\begin{lemma}\label{lem:f_aux3}
    Consider an allocation $\X$ that satisfies properties \textsc{FP$1$-FP$5$}. If there exists $k \to r$ such that $X_r \lowerval{r} A_r(\X) \sumsq b_{k,r}$, we can construct a complete $\efx$ allocation in polynomial time.
\end{lemma}
\begin{proof}
	We update the allocation as follows: (let $\X$ and $\xp$ denote the allocation before and after the update.)
\begin{align*}
    X_r' &\takes A_r(\X) \sumsq b_{k,r}\\
	X_k' &\takes a_{k,r}.
\end{align*}

	Now, $k$ and $r$ are non-resented, and therefore, $(k,r)$ is a support pair.
	Agent $r$ does not resent anyone as she did not resent anyone before, and now she is receiving a more valuable bundle.
	Also, $a_{r,k}$ was agent $r$'s most valuable unit bundle since while possessing it, she did not resented anyone.
    Moreover, $v$ still \weakresents\ $u$, and $v,u \notin \{k,r\}$.
	Thus, we have satisfied the conditions of \cref{lem:weak_support}, so we construct a complete $\efx$ allocation.
\end{proof}

\begin{lemma}\label{lem:f_aux4}
    If an allocation is in \stage\ F, we can construct an allocation that is either a complete $\efx$ allocation or satisfies all properties \textsc{FP$1$-FP$8$}.
\end{lemma}
\begin{proof}
	Let $\X$ be our initial allocation in \stage\ F. 
 	By definition of \stage\ F, $\X$ is \pthree, and hence, 
	\pone\ and all agents have a single unit bundle.
	Moreover, there exist three agents $i, j, k$ such that $j \to i$, $k$ is non-resented, 
	$k$ does not have $a_{k,j}$, and
        $$a_{k,j} \sumsq D_k(\X) \greaterval{k} X_k.$$
	Then, we execute update rule F1 (defined in \cref{up rule:f1}) to get allocation $\xp$.
	Since $\X$ is in \stage\ F, it not in \stage\ B, 
	so for every resent edge $p \to q$, we must have $A_p(\X) \cup X_p \lowereqval{p} a_{p, q}$. 
	Therefore, all the conditions of \cref{lem:f_aux1} are satisfied, so allocation $\xp$ satisfies properties \textsc{FP$1$-FP$5$}. 
	
	Rename $\X'$ to be $\X$. 
	By \cref{lem:f_aux3}, if there exists a resent relation $k \to r$, such that $X_r \lowerval{r} A_r(\X) \sumsq b_{k,r}$, 
	we can construct a complete $\efx$. 
	Therefore, without loss of generality, we assume that for all $k \to r$, 
	we have $X_r \greatereqval{r} A_r(\X) \sumsq b_{k,r}$, which means that \textsc{FP$6$} is also satisfied. 
	Moreover, if \textsc{FP$8$} is violated, we can construct a complete $\efx$ allocation by \cref{lem: stageCprop}. 
	Thus, without loss of generality, we assume that \textsc{FP$8$} is also satisfied.

	Therefore, we only need to satisfy the property \textsc{FP$7$}. To do so, we proceed as follows. We execute the update rule (F1) as long 
	as there exists a pair of agents $i, j$, violating the property \textsc{FP$7$} and the allocation is not complete. We maintain the following 
	invariants:
\begin{itemize}
    \item For every resent edge $p\to q$ where $p\ne k$ and $q$ and is the most-resented agent by $p$, we have $a_{p,q} \greatereqval{p} A_p(\X) \sumsq X_p$.
    \item Properties \textsc{FP$1$-FP$5$}, \textsc{FP$6$} and \textsc{FP$8$} remain satisfied.
\end{itemize}
	Note that initially, both invariants are satisfied. After every execution of the update rule (F1), properties \textsc{FP$1$-FP$5$} remain
	satisfied, by \cref{lem:f_aux1}. We rename agents $u, v$ in every execution of (F1) according to the \cref{lem:f_aux1} so that the 
	mentioned properties are still satisfied for agents $k, u, v$. Also, if one of the properties \textsc{FP$6$} or \textsc{FP$8$} becomes violated, 	we can construct an $\efx$ allocation by \cref{lem:f_aux3} or \cref{lem: stageCprop}, respectively. Also, note that if the first invariant gets 	violated, then we can construct an $\efx$ allocation by \cref{lem:f_aux2}. Therefore, by each execution of (F1), we either find a complete 
	$\efx$ allocation or we maintain the above invariants. 
    
    Finally, note that agent $k$ is fixed in these iterations, and after each execution of (F1),  $v_k(X_k)$ strictly increases, and after every iteration agent $k$ gets bundle 
    $a_{k,j}  \sumsq (B_k\setminus E_{k,j})$ for some agent $j$.
    Note that this bundle is only dependent to agents $k$ and $j$.
    As a result, this procedure must terminate after at most $n$ iterations since there are $n$ agents. Upon the termination of the procedure, we have an allocation that is either a 
	complete $\efx$ allocation or satisfies all properties \textsc{FP$1$-FP$8$}.  
\end{proof}

\begin{definition}[Update Rule F2\label{up rule:f2}] 
    Given a \pone\ allocation $\X$ and agents $p, q$ such that $p \to q$, $A_q(\X) \sumsq b_{p,q} \greaterval{q} X_q$, let $q$ be the most-resented vertex of $p$ that has this property. Then, let: (denote the allocation after the update by $\xp$)
    \begin{align*}
    X'_q &\takes A_q(\X) \sumsq b_{p,q} \\
    X'_p &\takes (X_p\setminus b_{p,q}) \sumsq a_{p,q}.
	\end{align*}
\end{definition}

\begin{lemma}\label{lem:f_aux5}
    If we execute update (F2) on an allocation $\X$ satisfying properties \textsc{FP$1$} and \textsc{FP$3$-FP$8$}, we get an allocation that still satisfies properties \textsc{FP$1$} and \textsc{FP$3$-FP$8$} with strictly lower number of resented agents.
\end{lemma}

\begin{proof}
    Consider agents $p$ and $q$ such that $p\to q$,  $A_q(\X) \sumsq b_{p,q} \greaterval{q} X_q$, and $q$ is the most-resented vertex of $p$ that has this property. 
    Let $\X$ and $\xp$ denote the allocations before and after the update, respectively. 

    First, note that every agent in $\xp$ is getting a bundle at least as valuable as in $\X$,
    so no new resent is created. Moreover, $p$ does not resent $q$ anymore, so the number of resented agents has strictly decreased.

    Next, we show that $\xp$ satisfies properties \textsc{FP$1$} and \textsc{FP$3$-FP$8$} with resepect to agents $(k,u,v)$.
    Since, no new resent is created, allocation remains \pone, satisfying (FP1).

    By property \textsc{FP$6$} for $\X$, we get $p \ne k$, and since $q$ was resented, we get
    $q\ne k$.
    Hence, agent $k$'s bundle does not change, so (FP3) holds.

    by property \textsc{FP$5$} for $\X$, $u\neq q$, and since $u$ does not resent anyone by \textsc{FP$4$} (for $\X$),
    we get $u\ne p$. Thus,
    since no new resent is created, and $u\notin \{p,q\}$, we get that (FP4) still holds, too.


    Note that no previously allocated unit bundle becomes unallocated by this update. 
    Also, no new resent edge is created in the resent graph. 
    This means that for any agent $i$, the sets $A_i(\X)$ and $D_i(\X)$ cannot increase in value with respect to agent $i$'s 
    valuation function. Therefore, properties \textsc{FP$5$-FP$8$} are still held,
    which completes the proof.
\end{proof}

\begin{lemma}\label{lem:f_aux6}
    If an allocation is in \stage\ F, we can construct an allocation that is either a complete $\efx$ allocation or satisfies all properties \textsc{FP$1$} and \textsc{FP$3$-FP$9$} in polynomial time.
\end{lemma}

\begin{proof}
    By \cref{lem:f_aux4}, we obtain an allocation that is either a complete $\efx$ allocation or satisfies all properties \textsc{FP$1$-FP$8$}. 
	In case the allocation is not a complete \efx\ yet, we execute the update rule (F2) as long as the property \textsc{FP$9$} 
	is violated and the allocation is not complete. By \cref{lem:f_aux5}, this procedure will either output a complete $\efx$ allocation 
	or maintain properties \textsc{FP$1$} and \textsc{FP$3$-FP$8$} while strictly decreasing the number of resented agents. 
	Therefore, this procedure will terminate after at most $n$ iterations, and we will finally have an allocation that is either a complete $\efx$ allocation or satisfies 
	all properties \textsc{FP$1$} and \textsc{FP$3$-FP$9$}.
\end{proof}

Next, we provide the main lemma of this section.

\begin{lemma}
\label{lem: stageF}
	If an allocation is in \stage\ F,
    we can construct a complete $\efx$ allocation in polynomial time.
\end{lemma}

\begin{proof}
    By \cref{lem:f_aux6}, we can construct an allocation that is either a complete $\efx$ allocation or satisfies all properties \textsc{FP$1$} and \textsc{FP$3$-FP$9$}. Denote this allocation by $\X$. Assume the allocation is not complete yet. Moreover, without loss of generality, by our assumption at the end of \stage\ A, for every $k \to i$ and $\ell \to j$, the function $R(i, j)$ is well-defined using \cref{lem:strong_A}. We now provide a dumping phase, which completes our $\efx$ allocation.

\paragraph{Dumping Phase.}
    We allocate every remaining unallocated unit bundle by the following rules.
	Here, a \emph{root} means a non-resented agent, and for a root $p$, the set $R_p$ denotes the agents resented by $p$ in $\X$. Every unallocated good gets allocated based on one of the following rules:
    
    \begin{enumerate}[label=(\roman*), leftmargin=2.2em]

		\item \textbf{Root to root.} 
		 Let $p$ and $r$ be two distinct roots.
    		We orient the pair so that $p\neq k$.
    		For $r=k$, assign \dumpu{a_{p,k}}{p} and \dumpp{b_{p,k}}{k}. Note that by \textsc{FP$4$}, agent $u$ is a resented agent in $\X$, and by \textsc{FP$4$}, for any agent $p\ne u$, 
		agent $k$ does not possess unit bundle $a_{k,p}$ in $\X$. For $r\ne k$, if at least one of the unit bundles in $E_{p, r}$ is already allocated to one of the roots, 
		allocate the other unit bundle to the other root. 		
		Otherwise, allocate one unit bundle in $E_{p,r}$ to $p$ and the other one to $r$. 

		\item \textbf{Root to its own children.} For every root $p$ and every $i\in R_p$, assign \dumpp{b_{i,p}}{p}. Also, we already have $X_i = a_{i, p}$.
		
		\item \textbf{Root to other children.} Let $p$ and $r$ be two distinct roots and let $i\in R_r$. Let \dumpu{a_{p,i}}{p}.
    If $p \neq k$,
    assign \dumpp{b_{p,i}}{k}.
    If $p = k$,
    assign \dumpp{b_{k,i}}{r}.

		\item \textbf{Between two children of the same tree.} Let $p$ be a root and let $i,j\in R_p$ be distinct. Let \dumpp{a_{i,j}}{p}
    If $p\neq k$, assign \dumpp{b_{i,j}}{k}.
    If $p=k$, assign \dumpp{b_{i,j}}{v}.

		\item \textbf{Between two children of two different trees.}         
            Let $p$ and $r$ be two distinct roots, and let $i\in R_p$ and $j\in R_r$.
    Assign \dumpp{a_{i,j}}{k} and \dumpp{b_{i,j}}{R(i,j)}.
    By property \textsc{FP$8$}, we may suppose that if $p=k$, then $R(i,j) = r$.
    \end{enumerate}

    This completes the allocation.

    \paragraph{Correctness of Dumping Phase.}Denote the allocation right before dumping phase  by $\X$, and denote the allocation at the end by $\xp$. Next, we prove that the final allocation, $\xp$, is \efx.

    Note that by FP1 and FP9, $\X$ follows global properties before the dumping phase (\cref{pbd}), and by construction, 
	our dumping rules follow our general dumping structure. Thus, by \cref{lem:dumprule},
	\begin{enumerate}
		\item For every agent $p$, we have $X'_p \greatereqval{p} X_p$.
		\item Nobody strongly envies any agent who was resented in $\X$.
		\item For every resent $p\to q$ in $\X$, agent $q$ does not envy agent $p$ in $\xp$.

	\end{enumerate}

    To prove that the final allocation is $\efx$, we show that any agent $p$ who was non-resented in $\X$, cannot be envied in $\xp$. To do so, we prove the following. 

    \begin{itemize}
        \item \textbf{Agent $k$ is not envied by any agent.} 
        First, we show that $u$ does not envy $k$ in $\xp$. By property \textsc{FP$5$}, we have that
        \begin{align*}
            X_u' &\greatereqval{u} X_u \greatereqval{u}
            a_{k,u} \sumsq (A_u(\X)) \setminus (E_{u,k}\sumsq E_{u,i})) \sumsq  b_{i,u} \\
            &\greatereqval{u} (X'_k\cap E_{k,u}) \cup \bigcup_{r\notin \{u,k\}} (X'_k\cap E_{r,u}) = X'_k \cap E_u \greatereqval{u} X'_k,            
        \end{align*}
    	where the third inequality comes from the fact that for every $r\notin \{k,u\}$, at most one unit bundle in $E_{u,r}$ is allocated to $k$.
        
        Next, consider a resented agent $q \ne u$. By the third statement of \cref{lem:dumprule}, we only need to consider the case where $i \to q$ such that $i \neq k$. 
    Since $q\ne u$, by FP3, we get that $a_{q,k}$ was unallocated in $\X$.
    This means that $a_{q,k} = A_q (\X)\cap E_{q,k}$.
    In allocation \xp, $k$ drops $b_{q,k}$ if she previously had it and picks up $a_{k,q}$, so
    $ X'_k \cap E_{k,q} = a_{k,q} \lowereqval{q} a_{q,k} = A_q (\X)\cap E_{q,k}$.
    For every resented agent $r\ne q$, we have 
    $A_q(\X) \cap E_{q,r}= a_{q,r} \greatereqval{q} X'_k \cap E_{q,r}$.
    Moreover, for every non-resented $r\ne k$, we have 
    $X'_k\cap E_{q,r} \lowereqval{q} b_{r,q} \lowereqval{q} A_q(\X)\cap E_{r,q}$.
    Since for all such $q$, by \textsc{FP$9$}, we have $A_q(\X) \sumsq b_{i,q}  \lowereqval{q} X_q$,
    we get:
\begin{align*}
    X'_q &= X_q \greatereqval{q} A_q(\X) \sumsq b_{i,q} \\
         &\greatereqval{q} A_q(\X)=(A_q(\X)\cap E_{k,q}) \cup \bigcup_{r\notin \{q,k\}} (A_q(\X)\cap E_{r,q}) \\ 
         &\greatereqval{q} (X'_k\cap E_{k,q}) \cup \bigcup_{r\notin \{q,k\}} (X'_k\cap E_{r,q}) = X'_k \cap E_q \greatereqval{q}
    X'_k,
\end{align*}
    which means $k$ is not envied by any resented $q$.

	For a non-resented agent $p \notin \{k,u\}$, by rules (i) and (ii), we get that if $X_k'\cap E_{p,q}\ne \emptyset$ for some agent $q$, then $q$ should be a resented agent who is not resented by $p$; therefore, by rule (iii), 
    we get $X_k'\cap E_{p,q} = b_{p,q} \lowereqval{p} a_{p,q} = X_p'\cap E_{p,q}$.
    Moreover, we have $X_k'\cap E_{p,k} = b_{p,k} \lowereqval{p} a_{p,k} = X_p'\cap E_{p,k}$, so for every agent $q\ne p$, we have $X_k'\cap E_{p,q} \lowereqval{p} X_p'\cap E_{p,q}$.
	This results in
	$$X'_p \greatereqval{p} \bigcup_{q\ne p} (X'_p \cap E_{p,q}) \greatereqval{p} 
		\bigcup_{q\ne p} (X'_k \cap E_{p,q}) \greatereqval{p} X'_k.$$


        \item \textbf{Any non-resented agent $p \neq k$ is not envied by any agent.}
        First, we show that no resented agent $q$ envies $p$ in $\xp$.
        By \cref{lem:dumprule}, we only need to consider the case where $r\to q$ for $r\ne p$. 
        
        If $p = v$ and $r = k$, then 
        $$X'_v \cap E_q \lowereqval{q} a_{v,q} \sumsq D_q(\X) \lowereqval{q} X_q \lowereqval{q} X'_q,$$
        where the second inequality follows from \textsc{FP$8$}.

        Otherwise, we have either $p \neq v$ or $r \neq k$, so 
        	$$X'_p \cap E_q \subseteq a_{p,q} \sumsq   \bigsumsq_{i \in R_p} (X'_p \cap E_{i,q})
            \lowereqval{q}  a_{q,p} \sumsq D_q(\X).$$
        If there exists a resent edge $p\to i$ such that $R(i,q)=p$, then 
    	$$X'_p \lowereqval{q}  a_{q,p} \sumsq D_q(\X)  \lowereqval{q} X_q = X'_q,$$
    	where the last inequality comes from the fact  $R(i,q)=p$.
    	Also, if for every $p\to i$, we have $R(i,q)=r$, then 
    	$$X'_p \cap E_q =a_{p,q} \lowereqval{q}  X_q= X'_q,$$
        where the inequality comes from the fact that $q$ did not resent $p$ in $\X$.
        
        Next, we show that no non-resented agent $q$ envies $p$ in $\xp$.
        If $q \ne k$, then the only unit bundle incident to $q$ that can be allocated to $p$ during the dumping phase is the unit bundle from $E_{q,p}$ assigned by Rule~(i).
        As a result, we get 
        $X'_q \greatereqval{q} X_q \greatereqval{q} a_{q,p} \greatereqval{q} X'_p \cap E_{p,q} = X'_p \cap E_q$,
        where the second inequality is due to $p$ being non-resented in $\X$.
        Hence, $q$ does not envy $p$ in $\xp$.

        
        Next, we show that $k$ does not envy $p$ in $\xp$.
        Note that whenever a unit bundle in $E_{k, i}$ has been assigned to $p$, it can only be allocated via Rules (i) and (iii). Hence,
        $$X_p' \cap E_k= (X_p' \cap E_{k,p}) \cup \bigcup_{r \in R_p}(X_p' \cap E_{r,k}).$$ 
        By FP7, we have for all $j\to i$, we have either $a_{j,k} \sumsq D_k (\X)\lowereqval{k} X_k$ or agent $k$ possesses $a_{k,j}$.
        If $k$ possesses $a_{k,j}$, then $X_p' \cap E_{k,p}\lowereqval{k} b_{p,k}$, so 
            $$X_p' \lowereqval{k} b_{p,k} \cup \bigcup_{r \in R_p} b_{k,r} \lowereqval{k}
            a_{p,k} \cup \bigcup_{r \in R_p} a_{k,r} \lowereqval{k} X'_k.$$ 
        Otherwise, if $k$ does not possess $a_{k,j}$, we get $a_{p,k} \sumsq D_k (\X)\lowereqval{k} X_k$. Hence, 
           $$X_p' \lowereqval{k} a_{p,k} \cup \bigcup_{r \in R_p} b_{k,r} 
           \lowereqval{k} a_{p,k} \sumsq D_k (\X)   \lowereqval{k}  X_k \lowereqval{k} X'_k.$$

        Therefore, $k$ does not envy $p$.            
    \end{itemize}
    	Hence, the final allocation is a complete \efx\ allocation.
\end{proof}

\section{\Stage\ G}\label{sec:g}

In this section, we prove that whenever an allocation is in \stage\ G, we can construct a complete $\efx$ allocation.
	We first prove a couple of auxiliary lemmas, and then, we introduce the main lemma of this section, which proves the existence of $\efx$ allocations for any allocation in \stage\ G.

We proceed in three steps.
First, \Cref{lem:final_support} shows that, once a single support pair exists,
we can complete the allocation via a clean dumping procedure.
Second, \Cref{lem:bundle_order} rules out a problematic configuration of unit bundles and, when it occurs,
reduces us to the setting of \Cref{lem:final_support}.
Finally, in the proof of the main lemma of \stage\ G, we perform a case analysis on the witness quadruple
and in each case reduce to either \cref{lem:final_support} or \cref{lem:strong_support}.

\begin{lemma}\label{lem:final_support}
    Given a \pthree\ allocation, if there exists a support pair $(s,t)$,
	we can construct a complete \efx\ allocation.
\end{lemma}

\begin{proof}
	We may assume that we are past \stage\ F, as otherwise the algorithm would already terminate. Now, we provide a dumping phase.
	
	\paragraph{Dumping Phase.}
    We allocate every remaining unallocated unit bundle by the following rules.
	Here, a \emph{root} means a non-resented agent, and for a root $p$, the set $R_p$ denotes the agents resented by $p$. Every unallocated good gets allocated based on one of the following rules:
	
	\begin{enumerate}[label=(\roman*), leftmargin=2.2em]
		
		\item \textbf{The edge between $s$ and $t$.}
		Since $(s,t)$ is a support pair and $\X$ is \pthree, we have $X_s \subseteq E_{s,t}$ and $X_s$ contains exactly one unit bundle. We assign the other unit bundle in $E_{s, t}$ to agent $t$. 
		
		\item \textbf{Root to its own children.}
		For every root $p$ and every $u\in R_p$, assign \dumpp{b_{p,u}}{p}.
		
		\item \textbf{Root to others' children.}
		Let $p$ and $r$ be two distinct roots and let $u\in R_r$.
		Assign \dumpu{a_{p,u}}{p} and \dumpp{b_{p,u}}{r}.
		
		\item \textbf{Root to root.}
		Let $p$ and $r$ be two distinct roots and assume $\{p,r\}\neq \{s,t\}$.
		If at least one of the unit bundles is already allocated to one of the roots,
		allocate the other unit bundle to the other root.
		Otherwise, assume w.l.o.g.\ that $p\neq s$ and assign \dumpp{a_{p,r}}{p} and \dumpp{b_{p,r}}{r}.

		\item \textbf{Between two children of one tree.}
		Let $p$ be a root and let $u,v\in R_p$ be distinct.
		If $p\neq s$, assign \dumpp{a_{u,v}}{p} and \dumpp{b_{u,v}}{s}.
		If $p=s$, assign \dumpp{a_{u,v}}{s}.
		Let $r$ be the root of another resent tree with non-zero height and assign \dumpp{b_{u,v}}{r}. Note that such a tree exists as we are not in \stage\ D.
		
		\item \textbf{Between two children of two trees.}
		Let $p$ and $r$ be two distinct roots such that $p\neq s$, and let $u\in R_p$ and $v\in R_r$.
		Assign \dumpp{a_{u,v}}{s} and \dumpp{b_{u,v}}{p}.
		
	\end{enumerate}
	
	This completes the allocation.
	
	\paragraph{Correctness of Dumping Phase.}
	Let $\X$ be the allocation right before the dumping phase and let $\X'$ be the allocation after dumping.
	We prove that $\X'$ is \efx.
	
	Since we are in a \pthree\ allocation and we are past \stage\ F,
	the global properties before dumping (\Cref{pbd}) hold for $\X$.
	Moreover, our dumping rules satisfy the general dumping-structure requirements \ref{pd1}--\ref{pd8}.
	Therefore, all conclusions of \Cref{lem:dumprule} apply:
	\begin{enumerate}
		\item For every agent $p$, $X'_p \greatereqval{p} X_p$.
		\item Nobody strongly envies any agent who was resented at the start of dumping.
		\item For every resent edge $p\to q$ in $\X$, agent $q$ does not envy agent $p$ in $\X'$.
		\item No one envies agent $s$, and agent $s$ does not strongly envy anyone.
	\end{enumerate}

        By the second statement of \cref{lem:dumprule}, in order to prove that the final allocation is $\efx$, it suffices to show that any agent $p \neq s$ who was non-resented in $\X$, cannot be strongly envied in $\xp$.

   Let $q$ be a resented agent and $r$ be the unique root such that $r\to q$ in $\X$. By the third statement of \cref{lem:dumprule}, we can only consider the case where $r \neq p$. If $p$ is a root of a tree with height zero, then agent $p$ can only possess $a_{p, q}$ from $E_q$ in $\X'$, which was either allocated to $p$ in $\X$ or unallocated. In either case, $q$ does not resent $p$ in $\xp$. Otherwise, if we had $p \to \ell$ for some agent $\ell$, since the allocation $\X$ was not at \stage\ C, for the quadruple $p \to \ell$ and $r \to q$, we must have $D_q(\X) \cup a_{q, p} \lowereqval{q} X_q$. Therefore, we have $$X_p' \cap E_q \lowereqval{q} a_{p, q} \cup D_q(\X) \lowereqval{q} a_{q, p} \cup D_q(\X) \lowereqval{q} X_q = X_q'.$$
   Hence, $p$ cannot be envied by a resented agent $q$.

   Now, let $r \neq s$ be a non-resented agent in $\X$. Note that if $X_p' \cap E_r \subseteq E_{p, r}$, then $X_p' \cap E_r$ must be a single unit bundle from $E_{p, r}$, which was either allocated to $p$ in $\X$ or unallocated. In either case, $r$ does not resent $p$ in $\xp$. Otherwise, if there exists a unit bundle from $E_{r, j}$ for $j \neq p$ that has been dumped to $p$, it must have been done by Rule (iii). Since $\X$ was not at \stage\ F, we get that either $r$ possesses $a_{r, p}$ in $\X$ or $a_{r, p} \cup D_r(\X) \lowereqval{r} X_r$. If $r$ possesses $a_{r, p}$ in $\X$, then $$X'_p \cap E_r \lowereqval{r} b_{r, p} \sumsq \bigsumsq_{i \in R_p} b_{r, i} \lowereqval{r} a_{r, p} \sumsq \bigsumsq_{i \in R_p} a_{r, i} \lowereqval{r} X'_r,$$
   and hence, $r$ does not envy $p$. Otherwise, if $a_{r, p} \cup D_r(\X) \lowereqval{r} X_r$, we have $$X'_p \cap E_r \lowereqval{r} a_{p, r} \cup D_r(\X) \lowereqval{r} a_{r, p} \cup D_r(\X) \lowereqval{r} X_r \lowereqval{r} X'_r.$$
   Hence, $r$ does not envy $p$.

        	Therefore, the final allocation is a complete \efx\ allocation.
\end{proof}

\begin{lemma}\label{lem:bundle_order}
	In a \pthree\ allocation, suppose there exist agents $i \to j$ and a non-resented agent $k$
	such that $k$ holds $a_{k,j}$ and $j$ is \emph{not} the most resented agent of $i$.
	Then, we can construct a complete \efx\ allocation.
\end{lemma}

\begin{proof}
	We may assume that we are past \stage\ F, as otherwise the algorithm would already terminate.
	Let $r$ be the most resented agent of $i$.
	We update the allocation as follows: (let $\X$ and $\xp$ denote the allocation before and after the update)
	\begin{align*}
		&X'_i \takes a_{i,r},\\
		&X'_r \takes \textsc{choose}(r),\\
		&\text{Run }\textsc{Reduce Trees}(r).
	\end{align*}
	Let $\X'$ denote the allocation right before running \textsc{Reduce Trees}$(r)$.
	
	Let $h_0 = r$, $H = \{i,j,k\}$, and $p = i$. Note that all the conditions of \Cref{lem:treebreaker2} were satisfied before the execution of \textsc{Reduce Trees}$(r)$. Therefore, running \textsc{Reduce Trees}$(r)$
	results in a \pone\ allocation while the bundles of agents in $H$ do not change, and hence every agent in $H$ has only a single unit bundle. Moreover, since agents in $N \setminus H$ can only have a single unit bundle after the execution of \textsc{Reduce Trees}$(r)$, the allocation is in fact \pthree. 
    
	In particular, in the resulting allocation, agent $k$ is still non-resented and still holds $a_{k,j}$.
	Also, agent $j$ becomes non-resented since after the update, agent $i$ holds $a_{i,r}$,
	and by assumption $i$ resented $r$ the most in $\X$. Hence, $(p,j)$ is a support pair, and the claim follows from \Cref{lem:final_support}.
\end{proof}


\begin{lemma}\label{lem: stageG}
	If we are in \stage\ G, we can construct a complete \efx\ allocation.
\end{lemma}

\begin{proof}
    Let $\X$ be the initial allocation in \stage\ G.
	Fix a witness quadruple $k\to i$ and $\ell\to j$ for \stage\ G, and define $Y_i,Y_j$ as follows:
	\begin{align*}
		Y_j &:= [C_j(\X)\cap E_{j,R_k}] \sumsq [D_j(\X)\setminus E_{j,R_k}] \sumsq a_{j,k},\\
		Y_i &:= [C_i(\X)\cap E_{i,R_\ell}] \sumsq [D_i(\X)\setminus E_{i,R_\ell}] \sumsq a_{i,\ell}.
	\end{align*}
	Then, the condition of \stage\ G becomes $Y_j \greaterval{j} X_j$ and $Y_i \greaterval{i} X_i$.
	
	Since we are not in \stage\ E, we must have $X_k=a_{k,j}$,
	Otherwise, we would have $Y_j \lowereqval{j} A_j(\X)$ and the allocation will be in \stage\ E.
	Similarly, $X_\ell=a_{\ell,i}$.
	Moreover, we can assume that agent $k$ resents agent $i$ the most and agent $\ell$ resents agent $j$ the most. Otherwise, by \Cref{lem:bundle_order}, we can construct a complete $\efx$ allocation. 
	
	Consider the set
	\[
	U_j(\X) := \{u \in N \setminus k: X_u = a_{u,j} \text{ and } u \text{ resents at least one agent in } \X\}.
	\]

    Next, we distinguish between two cases:

    \begin{itemize}
        \item \textbf{Case 1. $\mathbf{a_{j,i}\greatereqval{j}  a_{j,k}}$.}
        We update the allocation as follows (let $\X$ and $\xp$ denote the allocation before and after the update, respectively):
        \begin{align*}
    		&X'_\ell \takes a_{\ell,j},\\
    		&X'_j \takes \textsc{choose}(j),\\
    		&\text{Run }\textsc{Reduce Trees}(j).
	    \end{align*}
        First, we show that $\xp$is \pone.
        Let $H=\{k,i,\ell\}$, $h_0=j$, and $p=\ell$.
        Denote the allocation right before executing $\textsc{Reduce Trees}(j)$ by $\xz$.
        We show that the conditions of \cref{lem:treebreaker2} hold.

        The only agent $h\in H\setminus p$ holding some $a_{h,s}$ for some agent $s\notin H$ is agent $k$, who possesses $a_{k,j}$. However, by definition of case 1, we have $a_{j,k}\lowereqval{j}  a_{j,i}$.
        Hence $\forall r \in N\setminus H$, for $q= \arg\max_{\ell \in H \setminus p} v_r(a_{r,\ell})$,
        unit bundle $a_{q,r}$ is not allocated to agent $q$.

        Note that in $\xz$ agent $\ell$ does not resent anyone, since she is receiving the unit bundle that she was previously resenting the most. Agent $h_0=j$ is non-resented, since the allocation was unitary, and $j$ got an unallocated unit bundle. Moreover, every new created resent is from agent $j$, so any path not including $h_0=j$ has length one. 
        Moreover, agent $j$ got an unallocated unit bundle with highest value with respect to her valuation function, and $a_{j,i}$ was unallocated when she chose.
        Thus, $X''_j \greatereqval{j} a_{j,i} \greatereqval{j} a_{j,k}=X''_k$.
        Therefore, agent $k$ is non-resented in $\xz$.

        As a result, by \cref{lem:treebreaker2}, we get that $\xp$ is \pone, agents in $H$ still hold their initial bundle, agents $k$ remains non-resented, and every agent in $N\setminus H$
        possesses exactly one unit bundle. Since every agent in $H$ also possesses exactly one unit bundle, we get that $\xp$ is \pthree.

        Since in $\X$, agent $j$ possessed $a_{j,\ell}$, and since $j$ was envied, 
        $a_{j,\ell}$ is the most valuable unit bundle for agent $j$. Hence, $j$ resents $\ell$ in $\xp$,
        so $j$ is non-resented in $\xp$.
        Hence, $(k, j)$ is a support pair in $\xp$, and $\xp$ is \pthree. Thus, we can compute a complete $\efx$ allocation by \Cref{lem:final_support}.

        \item \textbf{Case 2. $\mathbf{a_{j,i}\lowerval{j}  a_{j,k}}$.}
        We update the allocation as follows (let $\X$ and $\xp$ denote the allocation before and after the update, respectively):
\begin{align*}
        &X'_\ell \takes a_{\ell,j},\\
		&X'_k \takes a_{k,i},\\
		&X'_i \takes Y_i(\X),\\
        &X'_j =\textsc{choose}(j),\\
        &\text{Run }\textsc{Reduce Trees}(j).
	    \end{align*}
        First, we show that $\xp$is \pone.
        Let $H=\{k,i,\ell\}\cup R_\ell(\X)\setminus j$, $h_0=j$, and $p=\ell$.
        Denote the allocation right before executing $\textsc{Reduce Trees}(j)$ by $\xz$.
        We show that the conditions of \cref{lem:treebreaker2} hold.
        
        The only agent $h\in H\setminus p$ holding some $a_{h,s}$ for some agent $s\notin H$ is agent $i$, who possesses $a_{i,j}$. However, by definition of case 2, 
        we have $a_{j,i}\lowerval{j}  a_{j,k}$.
        Hence, $\forall r \in N\setminus H$, for $q= \arg\max_{\ell \in H \setminus p} v_r(a_{r,\ell})$,
        unit bundle $a_{q,r}$ is not allocated to agent $q$.

        Agent $i$ did not resent anyone in $\X$ because the allocation $\X$ was \pthree\ and is now getting better-off. Hence, agent $i$ does not resent any agent in $\xz$. Also, since $X'_k=a_{k,i}$, agent $k$ is non-resented in $\xz$.
        
        Note that agents $\ell$, $k$, and $R_\ell(\X)\setminus j$ do not resent anyone in $\xz$ since they are receiving the unit bundle that they value the most. 

        Also, note that by case 2, $a_{j,i}\lowereqval{j} a_{j,k}$, 
        and $a_{j,k}$ was unallocated before agent $j$ chooses her most valuable unallocated unit bundle, so we get that agent $j$ does not resent agent $i$ in $\xz$.
        Hence,  agent $i$ is non-resented.

        
        As a result, by \cref{lem:treebreaker2}, we get that $\xp$ is \pone, agents in $H$ still hold their initial bundle, and agents $i$ and $k$ remain non-resented.
        Hence, $(k,i)$ is a support pair in $\xp$, and
        agent $k$ does not resent anyone because she possesses her most valuable unit bundle. Moreover, the only agent who may have more than one unit bundle is agent $i$.
        Therefore, we can construct a complete $\efx$ allocation by \Cref{lem:strong_support}.\qedhere
    \end{itemize}
\end{proof}

\end{document}